%% file: main.tex
\documentclass[a4paper,USenglish,cleveref, autoref, thm-restate]{lipics-v2021}

\usepackage[utf8]{inputenc}
\usepackage{amsthm}
\usepackage{amsmath}
\usepackage{amsfonts}
\usepackage{amssymb}
\usepackage{graphicx}
\usepackage[pdflatex=true,recompilepics=false]{gastex}   

\usepackage[shortlabels,inline]{enumitem}
\usepackage{comment}
\usepackage[ruled, vlined]{algorithm2e}

\ifpdf
\else
  \AtBeginDocument{%
        \RemoveFromHook{shipout/firstpage}[hyperref]%
        \RemoveFromHook{shipout/before}[hyperref]%
  }
\fi

\dontprintsemicolon
\let\emptyset\varnothing
\let\epsilon\varepsilon

\def\abs#1{\ensuremath{\lvert #1 \rvert}} 

\def\myparagraph#1{\noindent\textbf{#1}~}
\newcommand{\real}{\mathbb R}

\newcommand{\nat}{\mathbb N}

\newcommand{\tuple}[1]{\langle #1 \rangle}

\newcommand{\Cyl}{\mathrm{Cyl}}

\newcommand*{\pr}{\mathbb{P}}

\newcommand*{\expect}{\mathbb{E}}

\newcommand\restr[2]{\ensuremath{\left.#1\right\rvert_{#2}}}
\newcommand\Supp{\textrm{\sf Supp}}
\newcommand\almostsure{\textrm{\sf AS}}
\newcommand\limitsure{\textrm{\sf LS}}
\newcommand\ASParity{\textrm{\sf AS\_Parity}}

\newcommand\Safe{\textrm{\sf Safe}}
\newcommand\Reach{\textrm{\sf Reach}}
\newcommand\Parity{\textrm{\sf Parity}}

\newcommand\D{\mathcal{D}}
\renewcommand\H{{\sf Hist}}
\newcommand\R{{\sf Hist^{\omega}}}

\newcommand\ecs{\textrm{\sf EC}}
\newcommand\Inf{\textrm{\sf Inf}}
\newcommand\mecs{\textrm{\sf MEC}}

\newcommand{\last}{\operatorname{\mathsf{last}}}

\newcommand\Val{\ensuremath{{\textrm{\sf Val}}}}

\def\loseabsorb{q_{\sf lose}}
\def\winabsorb{q_{\sf win}}

\def\Good{\textsf{Good}}

\newcommand\straa{\sigma}

\def\Act{A}

\newcommand*{\variance}{\mathbb{V}}
\newcommand*{\cov}{\textrm{\sf Cov}}
\def\Choice{\textrm{\sf Choice}}
\def\wgt{\textrm{\sf wgt}}
\def\occ{\textrm{\sf occ}}
\newcommand\freshaction[1]{F_{#1}}
\newcommand\purge[1]{\textrm{\sf purge}(#1)}
\newcommand\purgeinv[1]{\textrm{\sf purge}^{-1}(#1)}
\def\actionstay{\textrm{\sf stay}}

\title{The Value Problem for Multiple-Environment MDPs with Parity Objective}

\author{Krishnendu Chatterjee}{IST Austria}{}{https://orcid.org/0000-0002-4561-241X}{ERC CoG 863818 (ForM-SMArt) and Austrian Science Fund (FWF) 10.55776/COE12;}%
\author{Laurent Doyen}{CNRS \& LMF, ENS Paris-Saclay, France}{}{https://orcid.org/0000-0003-3714-6145}{}%
\author{Jean-Fran\c cois Raskin}{Universit\'e Libre de Bruxelles, Belgium}{}{https://orcid.org/0000-0002-3673-1097}{PDR Weave project FORM-LEARN-POMDP funded by FNRS and DFG, and the support of the Fondation ULB;}%
\author{Ocan Sankur}{Universit\'e de Rennes, CNRS, Inria, France \& Mitsubishi Electric R\&D Centre Europe, France}{}{https://orcid.org/0000-0001-8146-4429}{ANR BisoUS (ANR-22-CE48-0012) and ANR EpiRL (ANR-22-CE23-0029).}%
\authorrunning{K. Chatterjee \textit{et al.}}

\Copyright{Krishnendu Chatterjee and Laurent Doyen and Jean-Fran\c cois Raskin and Ocan Sankur}



\nolinenumbers
\begin{document}
\maketitle

\begin{abstract}
We consider multiple-environment Markov decision processes (MEMDP), which consist 
of a finite set of MDPs over the same state space, representing different scenarios
of transition structure and probability. 
The value of a strategy is the probability to satisfy the objective, here a parity objective,
in the worst-case scenario, and the value of an MEMDP is the supremum of the values
achievable by a strategy.

We show that deciding  whether the value is~$1$ is a PSPACE-complete problem, 
and even in~P when the number of environments is fixed,
along with new insights to the almost-sure winning problem, which is to decide if there exists 
a strategy with value~$1$. Pure strategies are sufficient for theses problems,
whereas randomization is necessary in general when the value is smaller than~$1$. 
We present an algorithm to approximate the value, running in double exponential space.
Our results are in contrast to the related model of partially-observable MDPs
where all these problems are known to be undecidable.
\end{abstract}

\section{Introduction}\label{sec:intro}

We consider Markov decision processes (MDP), a well-established state-transition model 
for decision making in a stochastic environment. 
The decisions involve choosing an action from a finite set, which
together with the current state determine a probability distribution
over the successor state. The question of constructing a strategy that maximizes
the probability to satisfy a logical specification is a classical 
synthesis problem with a wide range of applications~\cite{Puterman,HBMDP,BK08,BvD17}.

The stochastic transitions in MDPs capture the uncertainty in the effect of
an action. Another form of uncertainty arises when the states are (partially) 
hidden to the decision-maker, as in the classical model of partially-observable
MDPs (POMDP)~\cite{Ast65,PT87}. Recently, an alternative model of MDPs with partial 
information has attracted attention, the multiple-environment MDPs (MEMDP)~\cite{RS14}, 
which consists of a finite set of MDPs over the same state space.
Each MDP represents a possible environment, but the decision-maker
does not know in which environment they are operating. 
The synthesis problem is then to construct a single strategy 
that can be executed in all environments to ensure the objective be satisfied
independently of the environment.
This model is natural in applications where the structure of the transitions and their probability 
are uncertain such as in robust planning or population models with individual variability~\cite{BS19,CCK0R20,BSSJ23,vdVJJ23,SVJ24}.

In contrast to what previous work suggest, the two models of POMDP and MEMDP
are (syntactically) incomparable: the choice of the environment in MEMDP is 
adversarial, which cannot be expressed in a POMDP, and the partial observability 
of POMDP can occur throughout the execution, whereas the uncertainty in MEMDP 
is only initial. In particular, MEMDP are \emph{not} a subclass of POMDP
since pure strategies are sufficient in POMDPs~\cite{Mar98,CDGH15} while 
randomization is necessary in general in MEMDPs~\cite[Lemma~3]{RS14}.

The synthesis problem has been considered for traditional $\omega$-regular
objectives, defined as parity~\cite{RS14} or Rabin~\cite{SVJ24} condition,
in three variants: the almost-sure problem is to decide whether there exists a strategy 
that is winning with probability~$1$ in all environments, the limit-sure
problem is to decide whether, for every $\epsilon > 0$, there exists a strategy 
that is winning with probability at least~$1-\epsilon$ in all environments,
and the gap problem, which is an approximate version of the quantitative problem to decide, given a threshold $0 < \lambda \leq 1$,
whether there exists a strategy 
that is winning with probability at least~$\lambda$ in all environments. 
The limit-sure problem is also called the value-$1$ problem, 
where the value of an MEMDP is defined as the supremum of the values
achievable by a strategy. The value is $1$ if and only if the answer to 
the limit-sure problem is $\textsf{Yes}$.

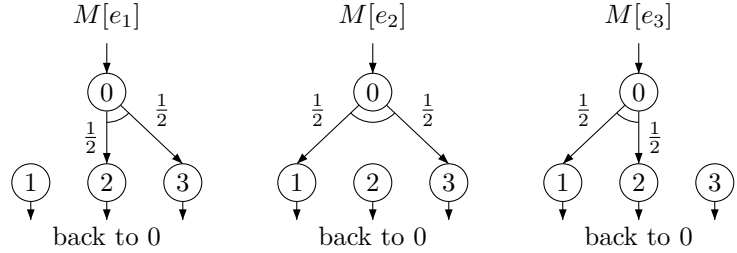
\begin{figure}[!t]
\hrule
\begin{center}
    \input{figures/as3.tex}
\end{center} 
\hrule
 \caption{Multiple-environment MDP for the missing card (over 3-card deck). 
 Each $M[e_i]$ represents the behavior of the MEMDP under environment $e_i$ where card $i$ has been removed.
 The environment can be identified almost-surely (with probability~$1$). \label{fig:as3}}
\end{figure}

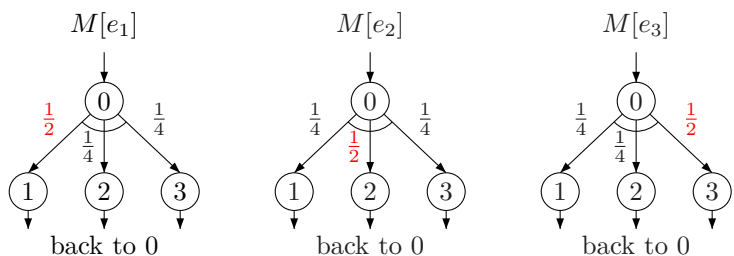
\begin{figure}[!t]
\hrule
\begin{center}
    \input{figures/ls3.tex}
\end{center} 
\hrule
 \caption{Multiple-environment MDP for the duplicate card (over 3-card deck). 
 Each $M[e_i]$ represents the behavior of the MEMDP under environment $e_i$ where card $i$ has been duplicated.
 The environment can be identified limit-surely (with probability arbitrarily close to~$1$).\label{fig:ls3}}
\end{figure}

A classical example to illustrate the difference between almost-sure and limit-sure 
winning is to consider an environment consisting of $51$ cards, obtained by removing 
one card from a standard $52$-card deck (see \figurename~\ref{fig:as3}). The decision-maker has two possible
actions: the action \emph{sample} reveals the top card of the deck and then 
shuffles the cards (including the top card, which remains in the deck);
the action \emph{guess($x$)}, where $1 \leq x \leq 52$ is a card, 
stops the game with a win if $x$ is the missing card, and a lose otherwise. 
If no guess is ever made, the game is also losing.
An almost-sure winning strategy is to sample until each of the $51$ cards has been 
revealed at least once, then to make a correct guess. It is easy to see that
the strategy wins with probability $1$, even if there exist scenarios (though with probability~$0$)
where some of the $51$ cards are never revealed and no correct guess is made.
Hence the MEMDP is almost-sure winning, and we say that it is not sure winning
because a losing scenario exists in every strategy.
Consider now an environment consisting of $53$ cards, obtained by adding
one duplicate card $c$ to the standard deck, and the same action set and rules of the game,
except that a correct guess is now the duplicate card $x = c$ (see \figurename~\ref{fig:ls3}).
The strategy that samples for a long time and then makes a guess based on the 
most frequent card wins with probability close to $1$ -- and closer to $1$
as the sampling time is longer -- but not equal to $1$, since no matter how long is the sampling phase
there is always a nonzero probability that the duplicate card does not have the highest frequency
at the time of the guess. In this case, the MEMDP is limit-sure winning, but not almost-sure winning.
Intuitively, the solution of almost-sure winning relies on the analysis of
\emph{revealing} transitions, which give a sure information allowing to exclude 
some environment (seeing card $c$ is a guarantee that we are not in the environment
where $c$ is missing); the solution of limit-sure winning involves \emph{learning} 
by sampling, which also allows to exclude some environment, but possibly with a 
nonzero probability to be mistaken. 

For MEMDPs with two environments, it is known 
that the almost-sure and limit-sure problem for parity objectives are solvable in 
polynomial time~\cite[Theorem~33, Theorem~40]{RS14}, while the gap problem is decidable in 2-fold exponential 
space~\cite[Theorem~30]{RS14} and is NP-hard, even for acyclic MEMDPs with two environments~\cite[Theorem~26]{RS14}. 
With an arbitrary number of environments,
the almost-sure problem becomes PSPACE-complete~\cite[Theorem~41]{SVJ24}, even for reachability objectives~\cite[Lemma~11]{vdVJJ23}.  
 For comparison, in the close model of POMDP, the decidability frontier lies between limit-sure winning and almost-sure winning: with reachability objectives, the almost-sure problem is decidable (and EXPTIME-complete~\cite{BGB12}), whereas the limit-sure problem is undecidable~\cite{GO10}. The gap problem is also undecidable~\cite{MHC03}. 

In this paper, we consider the limit-sure problem and the gap problem 
for parity objectives in MEMDPs with an arbitrary number of environments.
Our main result is to show that $(a)$ the limit-sure problem is PSPACE-complete and can be solved in polynomial time for a fixed number of environments, 
and $(b)$ the gap problem can be solved in double exponential space.
Correspondingly,
our algorithms significantly extend the solutions that are known for two environments, relying on a non-trivial recursive (inductive) analysis.

The PSPACE upper bound is obtained by a characterization
of limit-sure winning for a subclass of MEMDPs, in terms of almost-sure 
winning conditions (Lemma~\ref{lemma:ls-parity-charac}). A pre-processing phase transforms general MEMDPs into
the subclass. We present a PSPACE algorithm to compute the pre-processing and
verify the characterization. 
Since our algorithm relies on almost-sure winning, we also give a new characterization 
of almost-sure winning for parity objectives in MEMDPs (Lemma~\ref{lem:char-as-parity}), 
which gives a conceptually simple alternative algorithm to the known solution~\cite{SVJ24}.
The PSPACE lower bound straightforwardly follows from the same reduction as for 
almost-sure winning~\cite[Theorem~7]{SVJ24}.
A corollary of our characterizations is a refined strategy complexity: pure (non-randomized) strategies 
are sufficient for both limit-sure and almost-sure winning, which was known 
only for acyclic MEMDPs and almost-sure reachability objectives~\cite[Lemma~12]{vdVJJ23},
and exponential memory is sufficient.
In the last part of the paper, we present an algorithm running in double exponential space
for solving the gap problem, by computing an approximation of the value of the MEMDP.
To win with probability at least~$\lambda$ in all environments, 
randomized strategies are more powerful~\cite[Lemma~3]{RS14}, and thus need 
to be considered for solving the gap problem.

In conclusion, 
the model of MEMDP is a valuable alternative to POMDPs, from a theoretical
perspective since the limit-sure problem and gap problem are undecidable 
for POMDPs  whereas our results establish decidability for MEMDPs, and from a practical perspective since
many applications of POMDPs can be expressed by MEMDPs, as was observed previously~\cite{BSSJ23,vdVJJ23}.

\section{Definitions}

A \emph{probability distribution} on a finite set $Q$ is a function 
$d: Q \to [0, 1]$ such that $\sum_{q \in Q} d(q) = 1$. 
The support of $d$ is $\Supp(d) = \{q \in Q \mid d(q) > 0\}$. 
A Dirac distribution assigns probability~$1$ to some $q \in Q$.
We denote by $\D(Q)$ the set of all probability distributions on~$Q$.

\subsection{Markov Decision Processes}
A \emph{Markov decision process (MDP)} over a finite set $A$ of actions
is a tuple $M = \tuple{Q, (A_q)_{q \in Q}, \delta}$ consisting
of a finite set $Q$ of \emph{states}, 
a nonempty set $A_q \subseteq A$ of actions for each state $q \in Q$, 
and a partial probabilistic transition function $\delta : Q \times A \to \D(Q)$.
We say that $(q,a,q')$ is a transition if $\delta(q,a)(q') > 0$.
A state $q \in Q$ is a \emph{sink} if $\delta(q,a)(q) = 1$ for all $a \in A_q$. 

A \emph{run} of~$M$ from an initial state $q_0 \in Q$ is an infinite sequence 
$\pi = q_0 a_0 q_1 a_1 \ldots $ of interleaved states and actions such that 
$a_i \in A_{q_i}$ and $\delta(q_i,a_i)(q_{i+1}) > 0$ for all~$i \geq 0$.
Finite prefixes $\rho = q_0 a_0 \ldots q_n$ of runs ending in a state 
are called \emph{histories} and we denote by $\last(\rho) = q_n$ the last 
state of $\rho$.
We denote by $\R(M)$ (resp., $\H(M)$) the set of all runs (resp., histories) of~$M$,
and by $\Inf(\pi)$ the set of states that occur infinitely often along the run $\pi$.

A \emph{sub-MDP} of~$M$ is an MDP $M' = \tuple{Q',(A'_q)_{q \in Q'},\delta}$ such that 
$Q' \subseteq Q$ and $\Supp(\delta(q,a)) \subseteq Q'$ for all states $q \in Q'$ and
actions~$a \in A'_q$ (recall the requirement that $A'_q \neq \emptyset$). 
Consider a set $Q' \subseteq Q$ such that for all $q \in Q'$, there exists $a \in A_q$ with 
$\Supp(\delta(q,a)) \subseteq Q'$. We define the \emph{sub-MDP of~$M$ induced by $Q'$}, 
denoted by $\restr{M}{Q'}$, as the sub-MDP 
$M' = \tuple{Q',(A'_q)_{q \in Q'},\delta}$ where
$A'_q = \{a \in A_q\mid \Supp(\delta(q,a)) \subseteq Q'\}$ for all $q \in Q'$.

\smallskip
\myparagraph{End-components}
An \emph{end-component} of $M = \tuple{Q, (A_q)_{q \in Q}, \delta}$ is a 
pair $(Q',(A'_q)_{q \in Q'})$ such that $(Q',(A'_q)_{q \in Q'}, \delta')$
is a sub-MDP of $M$, where $\delta'$ denotes the restriction of $\delta$ to $\{(q,a) \mid q \in Q', a \in A'_q\}$,
and where the graph $\tuple{Q',E'}$
with $E'=\{(q,q') \in Q'\times Q' \mid \exists a \in A'_q: \delta(q,a)(q') > 0\}$
is strongly connected~\cite{DeAlfaro-phd97,BK08}.
We often identify an end-component as the 
set $Q' \cup \{(q,a) \mid q \in Q', a \in A_q\}$ of states and state-action pairs, 
and we say that it is \emph{supported} by the set $Q'$ of states. 
The (componentwise) union of two end-components with nonempty intersection is
an end-component, thus one can define \emph{maximal} end-components.
We denote by $\mecs(M)$ the set of maximal end-components of~$M$, 
which is computable in polynomial time~\cite{DeAlfaro-phd97},
and by $\ecs(M)$ the set of all end-components of~$M$.

\smallskip
\myparagraph{Histories and Strategies}
A~\emph{strategy} is a function~$\sigma: \H(M) \to \D(A)$ such that
$\Supp(\sigma(\rho)) \subseteq A_q$ for all histories~$\rho \in \H(M)$ ending in~$\last(\rho) = q$.
A strategy is \emph{pure} if all histories are mapped to Dirac distributions.
A strategy~$\sigma$ is \emph{memoryless} if $\sigma(\rho) = \sigma(\rho')$
for all histories $\rho, \rho'$ such that $\last(\rho) = \last(\rho')$.
We sometimes view memoryless strategies as functions $\sigma: Q \to \D(A)$.
A strategy~$\sigma$ uses \emph{finite memory} (of size $k$)
if there exists a right congruence $\approx$ over $\H(M)$ (i.e., such that
if $\rho \approx \rho'$, then $\rho \cdot a \cdot q  \approx \rho' \cdot a \cdot q$
for all $\rho, \rho' \in \H(M)$ and $(a,q) \in A \times Q$)
of finite index $k$ such that $\sigma(\rho) = \sigma(\rho')$
for all histories $\rho \approx \rho'$ with $\last(\rho) = \last(\rho')$.

\smallskip\myparagraph{Objectives}
An objective $\varphi$ is a Borel set of runs.
We denote by $\pr_{q}^\sigma(M, \varphi)$ the standard probability measure 
on the sigma-algebra over the set of (infinite) runs of $M$ with initial state~$q$, 
generated by the cylinder sets spanned by the histories~\cite{BK08}.
Given a history $\rho = q_0 a_0 q_1 \ldots q_k$, the cylinder set
$\Cyl(\rho) = \rho (AQ)^{\omega} $ has probability
\(
\pr_{q}^{\sigma}(M, \Cyl(\rho)) = \prod_{i=0}^{k-1} \sigma(q_0 a_0 q_1 \ldots q_i)(a_i) \cdot  \delta(q_i,a_i)(q_{i+1})
\) if $q_0 = q$, and probability~$0$ otherwise.
We say that a run $\rho$ is compatible with strategy $\sigma$ if $\pr_{q}^{\sigma}(M, \Cyl(\rho)) > 0$.

We consider the following standard objectives for an MDP $M$:
\begin{itemize}
\item safety objective: given a set $T \subseteq Q$ of states, let $\Safe(T)= \{
q_0 a_0 q_1 a_1 \ldots \in \R(M) \mid \forall i\geq 0: q_i \in T \}$;

\item reachability objective: given a set $T \subseteq Q$ of states, let $\Reach(T)= \{
q_0 a_0 q_1 a_1 \ldots \in \R(M) \mid \exists i\geq 0: q_i \in T \}$;

\item parity objective: given a priority function $p: Q \to \nat$,
let $\Parity(p)= \{\pi \in \R(M) \mid \min\{ p(q) \mid q \in \Inf(\pi) \} \text{ is even}\}$.

\end{itemize}
It is standard to cast safety and reachability objectives as special 
cases of parity objectives, using sink states.
Given an objective $\varphi$, we denote by $\lnot \varphi = \R(M) \setminus \varphi$ 
the complement of $\varphi$. We say that a run $\pi \in \R(M)$ \emph{satisfies}
$\varphi$ if $\pi \in \varphi$, and that it \emph{violates} $\varphi$ otherwise.

It is known that under arbitrary strategies,  with probability~$1$ the set $\Inf(\pi)$ of states occurring 
infinitely often along a run $\pi$ is the support of an end-component~\cite{CY-acm95,DeAlfaro-phd97}.

\begin{lemma}[\cite{CY-acm95,DeAlfaro-phd97}]\label{lem:basic-ec}
	Given an MDP $M$, for all states $q \in Q$ and all strategies $\sigma$,
	we have $\pr_{q}^{\sigma}(M, \{\pi \mid \Inf(\pi) \text{ is the support of an end-component}\}) = 1$.
\end{lemma}

An end-component $D \in \ecs(M)$ is \emph{positive} under strategy $\sigma$ 
from $q$ if $\pr_{q}^{\sigma}(M, \{\pi \mid\Inf(\pi)  = D\}) > 0$.
By Lemma~\ref{lem:basic-ec}, we have $\sum_{D \in \ecs(M)} \pr_{q}^{\sigma}(M, \{\pi \mid\Inf(\pi)  = D\}) = 1$.

\smallskip
\myparagraph{Value and qualitative satisfaction}
A strategy~$\sigma$ is winning for objective $\varphi$ from $q$ with probability
(at least) $\alpha$ if  $\pr_{q}^{\sigma}(M, \varphi) \geq \alpha$.
We denote by $\Val_{q}^*(M,\varphi) = \sup_{\sigma} \pr_{q}^{\sigma}(M, \varphi)$
the \emph{value} of objective $\varphi$ from state $q$.
A strategy~$\sigma$ is \emph{optimal} if $\pr_{q}^{\sigma}(M, \varphi) = \Val_{q}^*(M,\varphi)$.

We consider the following classical qualitative modes of winning. 
Given an objective $\varphi$, a state $q$ is:
\begin{itemize}
 
\item \emph{almost-sure winning} if there exists a strategy $\sigma$ such that 
is winning with probability~$1$, that is $\pr_{q}^{\sigma}(M, \varphi) = 1$.

\item \emph{limit-sure winning} if $\Val_{q}^*(M,\varphi) = 1$, or equivalently for all $\epsilon > 0$
there exists a strategy~$\sigma$ such that $\pr_{q}^{\sigma}(M, \varphi) \geq 1-\epsilon$.
\end{itemize}

We denote by $\almostsure(M,\varphi)$ and $\limitsure(M,\varphi)$
the set of almost-sure and limit-sure winning states, respectively.
In MDPs, it is known that $\almostsure(M,\varphi) = \limitsure(M,\varphi)$ and pure memoryless optimal strategies exist for parity objectives $\varphi$~\cite{Puterman,CY-acm95}.

We recall that the value of a parity objective $\varphi = \Parity(p)$ 
from every state of an end-component $D$ is the same, and is either $0$ or $1$, 
which does not depend on the precise value of the (non-zero) transition probabilities, 
but only on the supports $\Supp(\delta(q,a))$ of the transition function 
at the state-action pairs $(q,a)$ in $D$~\cite{DeAlfaro-phd97}.
When the value $1$, there exists a pure memoryless strategy $\sigma$
such that $\pr_{q}^{\sigma}(M, \varphi) = 1$
for all states $q\in D$.
If such a strategy exists, then $D$ is said to be \emph{$\varphi$-winning}, and otherwise \emph{$\varphi$-losing}.

\subsection{Multiple-Environment MDP}
\label{section:memdp}
A \emph{multiple-environment MDP (MEMDP)} over a finite set $E$ of environments
is a tuple $M = \tuple{Q, (A_q)_{q \in Q}, (\delta_e)_{e \in E}}$,
where $M[e] = \tuple{Q,(A_q)_{q \in Q},\delta_e}$ is an MDP that models the
behaviour of the system in the environment~$e \in E$. 
The state space is identical in all $M[e]$ ($e \in E$), only the 
transition probabilities may differ.
We sometimes refer to the environments of $M$ as the MDPs $\{ M[e] \mid e \in E\}$
rather than the set $E$ itself.
For $E' \subset E$, let $M[E']$ be the MEMDP $M$ over set $E'$ of environments.
We denote by $M[\lnot e]$ the MEMDP $M$ over environments $E \setminus \{e\}$,
and by $\cup_{e \in E} M[e]$ the MDP $\tuple{Q, (A_q)_{q \in Q}, \delta_{\cup}}$
such that $\delta_{\cup}(q,a)$ is the uniform distribution over $\bigcup_{e \in E} 
\Supp(\delta_{e}(q,a))$ for all $q \in Q$ and $a \in A$.

A transition $t = (q,a,q')$ is \emph{revealing} in $M$ 
if $K_{t} = \{ e \in E \mid q' \in \Supp(\delta_e(q,a))\}$
is a strict subset of $E$ ($K_{t} \subsetneq E$).
We say that $K_{t}$, which is the set of environments where the transition $t = (q,a,q')$
is possible, is the \emph{knowledge} after observing transition $t$.
An MEMDP is in \emph{revealed form} if for all revealing transitions $t = (q,a,q')$,
the state $q'$ is a sink in all environments, that is 
$\Supp(\delta_e(q',a)) = \{q'\}$ for all environments $e \in E$ and all actions $a \in A_{q'}$.
By extension, we call knowledge after a history $\rho$ the set of
environments in which all transitions of $\rho$ are possible.

\smallskip
\myparagraph{Decision Problems}
We are interested in synthesizing a \emph{single} strategy~$\sigma$ with guarantees in
\emph{all} environments, without knowing in which 
environment~$\sigma$ is executing.
We consider reachability, safety, and parity objectives.

A state $q$ is \emph{almost-sure winning} in $M$ for objective $\varphi$
if there exists a strategy $\sigma$ such that in all environments $e \in E$,
we have $\pr_{q}^{\sigma}(M[e], \varphi) = 1$, and we call such a strategy 
$\sigma$ almost-sure winning. 
A state $q$ is \emph{limit-sure winning} in $M$ for objective $\varphi$
if for all $\epsilon > 0$, there exists  a strategy $\sigma$ such that in all 
environments $e \in E$ we have $\pr_{q}^{\sigma}(M[e], \varphi) \geq  1-\epsilon$,
and we say that such a strategy $\sigma$ is $(1-\epsilon)$-winning.

We denote by $\almostsure(M, \varphi)$ (resp., $\limitsure(M, \varphi)$) the set
of all almost-sure (resp., limit-sure) winning states in $M$ for objective $\varphi$.
We consider the \emph{membership problem for almost-sure (resp., limit-sure) winning},
which asks whether a given state $q$ is almost-sure (resp., limit-sure) winning
in $M$ for objective $\varphi$. We refer to these membership problems as 
\emph{qualitative} problems.

We are also interested in the \emph{quantitative} problems.
Given MEMDP $M$, a parity objective $\varphi$, and probability threshold $\alpha\geq 0$,
we are interested in the existence of a strategy $\sigma$ satisfying
$\pr^\sigma_q(M[e], \varphi)\geq \alpha$ for all $e\in E$.
We present an approximation algorithm for the quantitative problem, solving the \emph{gap problem} consisting,
given MEMDP~$M$, state~$q$, parity objective $\varphi$, and
	thresholds~$0 < \alpha < 1$ and $\epsilon > 0$, in answering
	\begin{itemize}
		\item \textsf{Yes} if there exists a strategy $\sigma$ such that for all $e \in E$, we have
			$\pr_{q}^{\sigma}(M[e], \varphi)\geq \alpha$,
		\item \textsf{No} if for all strategies $\sigma$, there exists $e \in E$ with $\pr_{q}^{\sigma}(M[e],\varphi)<
			\alpha - \epsilon$,
		\item and arbitrarily otherwise.
	\end{itemize}

The gap problem is an instance of promise problems
which  guarantee a correct answer in two disjoint sets of
inputs, namely positive and negative instances -- which do not necessarily cover all
inputs, while giving no guarantees in the rest
of the input~\cite{ESY-ic84,goldreich2005promise}.

\smallskip
\myparagraph{Results}
We solve the membership problem for limit-sure winning with parity objectives $\varphi$ (i.e.,
deciding whether a given state $q$ is limit-sure winning, that is $q \in \limitsure(M, \varphi)$),
providing a PSPACE algorithm with a matching complexity lower bound, and showing that the problem is solvable in polynomial time when 
the number of environments is fixed.
Our solution relies on the solution of almost-sure winning, 
which is known to be PSPACE-complete for reachability~\cite{vdVJJ23} and Rabin objectives~\cite{SVJ24}. 
We revisit the solution of almost-sure winning and give a simple 
characterization for safety objectives (which is also PSPACE-complete),
that can easily be extended to parity objectives.
A corollary of our characterization is that pure (non-randomized) strategies are sufficient for both limit-sure
and almost-sure winning, which was known only for acyclic MEMDPs and reachability
objectives~\cite[Lemma~12]{vdVJJ23}.

For the gap problem, we present an double exponential-space procedure 
to approximate the value $\alpha$ that can be achieved in all environments,
up to an arbitrary precision $\epsilon$.

\section{Almost-Sure Winning}\label{section:as}
It is known that the membership problem for almost-sure winning in MEMDPs is PSPACE-complete 
with reachability objectives~\cite{vdVJJ23} as well as with Rabin objectives~\cite{SVJ24}, an
expressively equivalent of the parity objectives.
We revisit the membership problem for almost-sure winning with parity and safety objectives,
as it will be instrumental to the solution of limit-sure winning.
We present a conceptually simple characterization of the winning region
for almost-sure winning, from which we derive a PSPACE algorithm, thus 
matching the known complexity for almost-sure Rabin objectives.
A corollary of our characterization is that pure (non-randomized) strategies 
are sufficient for both limit-sure and almost-sure winning, which was known 
only for acyclic MEMDPs and reachability objectives~\cite[Lemma~12]{vdVJJ23}. 

\begin{theorem}[\cite{vdVJJ23},\cite{SVJ24}]\label{thm:pspacehard}
    The membership problem for almost-sure winning in MEMDPs with a reachability, safety, or Rabin  
    objective is PSPACE-complete.
\end{theorem}

To solve the membership problem for a safety or parity objective $\varphi$, 
we first convert $M$ into an MEMDP $M'$ 
in revealed form with state space $Q \uplus \{\winabsorb, \loseabsorb \}$ and 
each revealing transition $t=(q,a,q')$ in $M$ is redirected in $M'$ to $\winabsorb$
if $q' \in \almostsure(M[K_t], \varphi)$ is almost-sure winning when the 
set of environments is the knowledge $K_t$ after observing transition $t$,
and to $\loseabsorb$ otherwise. In order to decide if $q' \in \almostsure(M[K_t], \varphi)$,
we need to solve the membership problem for an MEMDP with strictly fewer
environments than in $M$ as $K_t \subsetneq E$, 
which will lead to a recursive algorithm. 
The base case of the solution is MEMDPs with one environment,
which is equivalent to plain MDPs.

It is easy to see that $\almostsure(M, \varphi) \cup \{\winabsorb\} = 
\almostsure(M', \varphi \cup \Reach(\winabsorb))$
for all prefix-independent objectives $\varphi$,
and we can transform the objective $\varphi \cup \Reach(\winabsorb)$
into an objective of the same type as $\varphi$ (for example, if $\varphi$
is a parity objective then assigning the smallest even priority to 
$\winabsorb$ turns the objective $\varphi \cup \Reach(\winabsorb)$
into a pure parity objective).

Hence, the main difficulty is to solve the membership problem for MEMDP 
in revealed form.

\subsection{Safety}\label{section:as-safety}

Although safety objectives are subsumed by parity objectives which we solve in the next
section, we give here a simpler algorithm specifically for safety, and also prove PSPACE-hardness
in this case.

The safety objective has the property that almost-sure winning is equivalent
to sure winning, where a strategy is sure winning if all runs compatible with 
the strategy satisfy the objective. Intuitively, if some runs does not
satisfy the safety objective $\Safe(T)$, then it contains a state outside $T$ 
after a finite prefix, thus with positive probability (the probability
of the finite prefix). 
In the sure-winning mode, we can consider
the probabilistic choices to be adversarial, which entails that only 
the support of the probability distributions in the transition function
is relevant.

It follows that, as long as the knowledge remains the set $E$ of all environments
a winning strategy for a safety objective can play all actions that 
are safe (i.e., that ensure the successor state remains in the winning region) 
in all environments. We obtain the following property:
almost-sure winning for a safety objective in a MEMDP $M$ in revealed form
is equivalent to almost-sure winning in the MDP $\cup_{e \in E} M[e]$.

An algorithm for solving almost-sure safety is as follows: 
$(1)$ for each revealing transition $t=(q,a,q')$ in $M$, 
decide if $q' \in \almostsure(M[K_t], \Safe(T))$ (using a recursive call),
and redirect the transition $t$ to $\winabsorb$ or $\loseabsorb$ accordingly,
transforming~$M$ into revealed form;
$(2)$ assuming $M$ is in revealed form, compute the almost-sure winning 
states $W = \almostsure(M_{\cup}, \Safe(T))$ where $M_{\cup} = \cup_{e \in E} M[e]$
is an MDP. Return $W \setminus \{\winabsorb\}$. The depth of recursive calls
is bounded by the number of environments, and the almost-sure safety in MDPs 
can be solved in polynomial time,
namely, in time $O(\abs{Q}^2\abs{A})$.
It follows that almost-sure safety in MEMDPs can be solved in PSPACE,
and in time $O(\abs{Q}^2\cdot\abs{A}\cdot 2^{\abs{E}})$.
A PSPACE lower bound can be established by a similar reduction from QBF 
as for reachability, the constructed MEMDP being acyclic~\cite{vdVJJ23}.

Note that for a fixed number of environments, the membership problem for almost-sure 
safety in MEMDPs is solvable in polynomial time by our algorithm
since the depth of the recursion is then constant.
This is also the case in Theorem~\ref{thm:pspacehard}
as shown in \cite{vdVJJ23}.

\subsection{Parity}\label{section:as-parity}
By definition, the almost-sure winning region $W = \almostsure(M, \Parity(p))$ 
for a parity objective in an MEMDP $M$ is such that there exists a strategy $\sigma$
that is almost-sure winning for the parity objective from every state $q \in W$
in every MDP $M[e]$ (where $e$ is an environment of $M$).  
In contrast, we show the following characterization (note the order of the quantifiers).

\begin{restatable}{lemma}{charasparity} \label{lem:char-as-parity}
	Given an MEMDP $M$ in revealed form with state space $Q$, if $W \subseteq Q$ is such that 
	in every environment $e$, 
	from every state $q\in W$, there exists a strategy $\sigma_e$
	that is almost-sure winning for the parity objective $\Parity(p)$ in $\restr{M}{W}[e]$ from~$q$,
	then $W \subseteq \almostsure(M, \Parity(p))$.
	Moreover, for all $q \in W$, there exists a pure ($\abs{Q}\cdot\abs{E}$)-memory strategy ensuring
	$\Parity(p)$ from~$q$ in $M$.
\end{restatable}

\begin{proof}
	For each environment $M[e]$, consider a memoryless strategy $\sigma_e$
	almost-surely winning for the objective $\Parity(p)$ in $\restr{M}{W}[e]$ from every state of $W$.
	Recall that almost-sure winning strategies can be assumed to be memoryless
	in MDPs with single environments; and that one can build a single memoryless strategy that is almost-surely winning
	from all winning states.
	Let $\ecs(\sigma_e) = \{ D \in \ecs(M[e]) \mid \exists q \in W: \pr_{q}^{\sigma_e}(M[e], \Inf = D) > 0 \}$ 
	be the set of positive end-components under strategy $\sigma_e$.
	Note that the least priority in an end-component $D \in \ecs(\sigma_e)$ is even
	since the parity objective is satisfied with probability~$1$.

	Let $E = \{1,\dots,k\}$ be the set of environments of $M$. 
	We construct a pure almost-sure winning strategy $\sigma$ for the MEMDP $M$ as follows,
	where initially $e=1$:

	$(1)$ play according to $\sigma_e$ for $\abs{W}$ steps;

	$(2)$ if the current state is $\winabsorb$ or belongs to a positive end-component
	$D \in  \ecs(\sigma_e)$, keep playing according to $\sigma_e$ forever.
	Otherwise, increment $e$ (modulo $k$) and go to $(1)$.
	\smallskip

The strategy $\sigma$ uses memory of size at most  $\abs{Q}\cdot\abs{E}$ since $W \subseteq Q$. 

	Fix environment $f \in E$. We show that strategy 
	$\sigma$ is almost-sure winning in $M[f]$. Because all strategies $\sigma_e$ are defined in $\restr{M}{W}$,
	the region $W$
	is never left while playing $\sigma$, and during phase $(1)$ of the strategy
	there is a lower-bounded probability to reach an end-component $D \in  \ecs(\sigma_e)$ 
	when $e=f$.

	We show that eventually phase $(2)$ is executed forever with probability~$1$,
	that is, some end-component $D \in  \ecs(\sigma_e)$ for some $e$ is reached with probability 1.
	Towards contradiction,
	assume that phase $(1)$ of the strategy $\sigma$ is executed infinitely often with positive probability $p$.
	Then phase $(1)$ for~$e=f$ and $\sigma_f$ is also executed infinitely often and
	it follows that, conditioned on phase $(1)$ being executed infinitely often, 
	a positive end-component $D \in  \ecs(\sigma_f)$ is reached with 
	probability~$1$; 
	hence phase $(2)$ is executed forever from that point on. 
	Thus with probability~$1-p + p = 1$ phase~$(1)$ is executed
	only finitely often, contradicting our assumption.  

	As phase $(2)$ of the strategy $\sigma$ is eventually executed forever with probability~$1$, let $e$
	be the corresponding environment (i.e., such that $\sigma$ plays according to $\sigma_e$)
	and let $D \neq \{\winabsorb\}$ be the reached end-component of $M[e]$ (the other case where $\winabsorb$
	is reached is trivial). 
	If some transition of $D$ is not present in $f$,
	then it must be a revealing transition in $e$, thus leading in $M[e]$ to $\winabsorb$ outside $D$,
	which is impossible since $D$ is an end-component in $M[e]$.
	Hence all transitions of $D$ are present in all environments.

	We show that $\sigma$ is almost-sure winning in $f$. 
	The result is immediate if $D$ is an end-component of $M[f]$ (in particular if $f = e$). 
	If $D$ is not an end-component of $M[f]$, then in $M[f]$ the strategy would leave $D$ and 
	reach $\winabsorb$, thus $\sigma$ is almost-sure winning as well in that case.
\end{proof}

\begin{algorithm}[t]
\caption{${\ASParity}(M,p)$ \label{alg:solve-as-ws}}
{
  \AlgData{$M = \tuple{Q, (A_q)_{q \in Q}, (\delta_e)_{e \in E}}$ an MEMDP, $p: Q \to \nat$ a priority function.}
  \AlgResult{The winning region $\almostsure(M,\Parity(p))$ for almost-sure parity.}

  \Begin{
        \hfill {\tt /* pre-processing */} \;
       \nl $M' \gets M$  \;
       \nl add two sink states $\winabsorb, \loseabsorb$ to $M'$ \;
       \nl define $p(\winabsorb)=0$ and $p(\loseabsorb)=1$ \;
       \nl \ForEach{{\rm revealing transition} $t = (q,a,q')$ {\rm in} $M$}
       {
           \hfill {\tt /* } $K_t \subsetneq E$ {\tt */}

	   \nl \If{$q' \in {\ASParity}(M[K_t],p)$\label{line:memdp-as-parity-revealing-recall}}
               {\nl replace $t$ by $(q,a,\winabsorb)$ in $M'$ \;}
               \Else
               {\nl replace $t$ by $(q,a,\loseabsorb)$ in $M'$ \;}
       }
       \nl $M \gets M'$  \;
       \hfill {\tt /* $M$ is in revealed form */} \;
       \nl $P \gets \emptyset$; $P' \gets \emptyset$\;
       \nl \Repeat{$P'$ is unchanged}
       {
	 \nl $P \gets \cap_{e \in E} \, \almostsure(M[e],\Parity(p))$  \hfill{{\tt /* } $M[e]$ and $\cup_{e \in E} M[e]$} \label{alg:solve-as-ws-P} \;
	 \nl $P' \gets \almostsure(\cup_{e \in E} M[e],\Safe(P))$  \hfill{are MDPs {\tt */}} \label{alg:solve-as-ws-Q} \;
	 \nl $M \gets \restr{M}{P'}$ \;
       }
       \nl \KwRet{$P' \setminus \{\winabsorb\}$} \;
  }
}
\end{algorithm}

The characterization in the first part of Lemma~\ref{lem:char-as-parity} 
holds simply because parity objectives are prefix-independent (runs that differ
by a finite prefix are either both winning or both losing), and thus the 
characterization holds for all prefix-independent objectives.

The converse of Lemma~\ref{lem:char-as-parity} is immediate, which entails
that the almost-sure winning region $W = \almostsure(M, \Parity(p))$ is the largest
set of states satisfying the condition in Lemma~\ref{lem:char-as-parity}.
We exploit this characterization in Algorithm~\ref{alg:solve-as-ws} to compute
the winning region for almost-sure parity. After transforming the MEMDP into
revealed form (through recursive calls to the algorithm), we compute 
the winning region for almost-sure parity in each environment (line~\ref{alg:solve-as-ws-P}), 
and then the set $P'$ of states from which we can remain in the intersection $P$ of all these winning regions
(line~\ref{alg:solve-as-ws-Q}). We iterate this process on $\restr{M}{P'}$
until a fixpoint $P=P'$ is reached. 

It is easy to see that the fixpoint satisfies the characterization 
of Lemma~\ref{lem:char-as-parity}, and thus 
$P' \subseteq \almostsure(M,\Parity(p)) \cup \{\winabsorb\}$.
Also by the proof of Lemma~\ref{lem:char-as-parity}, we can construct a pure almost-sure winning ($\abs{Q}\cdot\abs{E}$)-memory strategy from all states in $P'$, and define (recursively, for each subset of the environments) a pure almost-sure winning strategy from the states that were replaced by $\winabsorb$ in the revealed form, with a total memory size at most $\abs{Q}\cdot\abs{E} \cdot 2^{\abs{E}}$,
corresponding to the memory bound
from Lemma~\ref{lem:char-as-parity}
for each subset $K \subseteq E$ of environments (representing the belief, i.e., the set of environments where the current history is possible).

To show the converse inclusion, we show the invariant that
every state $q \in Q \setminus P'$ is not almost-sure winning in $M$:
for all strategies $\sigma$ from $q$, in some environment $M[e]$ the set $P$ is left 
with positive probability (along some history $\rho$). Given a state $q' \in Q \setminus P$ reached in $M[e]$,
there is an environment $f \in E$ where the parity objective is violated
with positive probability under $\sigma$ from $q'$.  The crux is to show that 
the state $q'$ is reached with positive probability
in $M[f]$ as well.
Towards contradiction, assume that the history $\rho$ from $q$ to $q'$ (in $M[e]$)
is not possible in $M[f]$. Then $\rho$ contains a revealing transition in $M[e]$,
and $q' = \winabsorb \in P$, which is a contradiction since $q' \in Q \setminus P$. 
Hence, in $M[f]$ with strategy $\sigma$ the parity objective is violated
with positive probability.

Algorithm~\ref{alg:solve-as-ws} can be implemented in PSPACE by a similar argument
as for almost-sure safety: the depth of recursive calls is bounded by the number
of environments, both almost-sure safety and almost-sure parity can be solved 
in polynomial time in MDPs, and the repeat-loop runs at most $\abs{Q}$ times.
The algorithm runs in polynomial time if the number of environments is fixed.
The PSPACE-hardness follows from Theorem~\ref{thm:pspacehard}.

\begin{restatable}{theorem}{asparity}
\label{th:as-parity}
	The membership problem for almost-sure parity in MEMDPs is PSPACE-complete.
	Pure exponential-memory strategies are sufficient for almost-sure winning in MEMDPs with parity (thus also reachability and safety)
	objectives. When the number of environments is fixed, the problem is solvable in polynomial time.
\end{restatable}

The time complexity of Algorithm~\ref{alg:solve-as-ws} is established as follows. 
Each recursive call, corresponds to a subset of the initial environment set $E$ that we can compute once and tabulate.
In each call, the second loop runs at most $\abs{Q}$ times,
and the set of almost-sure winning states for parity conditions (that is, the set $\almostsure(M[e],\Parity(p))$) can be computed in time $O(\abs{Q}\cdot \abs{\delta})$~\cite{BK08}.
Since $\abs{\delta}$ is in $O(\abs{Q}^2\cdot\abs{A})$,
each recursive call takes $O(\abs{Q}^4\cdot\abs{E}\cdot\abs{A})$ time,
and overall, this is $O(\abs{Q}^4\cdot\abs{E}\cdot\abs{A}\cdot 2^{\abs{E}})$.

Note that pure exponential-memory strategies for almost-sure parity in MEMDPs are provided
by Lemma~\ref{lem:char-as-parity}.
The algorithm for almost-sure parity can be used to solve almost-sure safety
with optimal PSPACE complexity, although the specific algorithm for safety
is slightly simpler (the repeat-loop can be replaced by just line~\ref{alg:solve-as-ws-Q}
where $P = T$ is the set of states defining the safety objective $\Safe(T)$).
 
The PSPACE procedure can be implemented in exponential time by solving
all subproblems and storing their solutions. Moreover,
for large numbers of environments, the exponent in the complexity can be made to depend only on the size of $M$. In fact, intuitively, two environments with identical supports yield the same result so one can derive a dynamic programming solution where at most one environment per support is solved.

Define the \emph{support} of a probabilistic transition relation $\delta : Q\times A \rightarrow \D(Q)$
as the family of supports of its transitions, that is, $\Supp(\delta) = (\Supp(\delta(q,a)))_{(q,a) \in Q\times A}$.
Define the support of a family of transition relations
as $\Supp((\delta_e)_{e \in E}) = \{ \Supp(\delta_e) \mid e  \in E\}$.

Two environments $\delta_e$ and $\delta_f$ are said to be \emph{equivalent} if they have the same support.
One can check whether two environments are equivalent in polynomial time,
by going through all triples $(q,a,q')$ and verifying that $\delta_e(q,a,q')= 0$
iff $\delta_f(q,a,q')=0$.

Almost sure parity in MEMDPs does not depend on the precise probability values in the given environments in $M$ but only on their supports.

In addition to Theorem~\ref{th:as-parity}, we can obtain a complexity bound whose exponent
is independent of the number of environments (Theorem~\ref{th:as-parity-exptime}),
using the following result: if in two
environments, the support of the transition relation is the same, 
we can discard one of the environment (all strategies that are almost-sure winning
in one are also almost-sure winning in the other one, as shown in Lemma~\ref{lemma:as-parity-equivalent-envs}) 
and thus consider at most one environment for each support.
Here, we denote by $\Supp((\delta_e)_{e \in E}) = (\Supp(\delta_e))_{e \in E}$ where $\Supp(\delta_e)$ denotes the set of transitions with positive probability under $\delta_e$.

\begin{restatable}{lemma}{asparityequivalentenvs}
	\label{lemma:as-parity-equivalent-envs}
	Consider two MEMDPs $M_i = \tuple{Q, (A_q)_{q \in Q}, (\delta_e)_{e \in E_i}}$ for $i=1,2$,	
	with the same state and action sets, and with the same supports of their
	transition relation, $\Supp((\delta_e)_{e \in E_1}) = \Supp((\delta_e)_{e \in E_2})$. 
	Given a parity condition $\Parity(p)$, for all states~$q$ and all finite-memory strategies $\straa$, 
	the following equivalence holds:
	$\pr_{q}^{\straa}[M_1[e], \Parity(p)] = 1$ for all $e\in E_1$ if and only if 
	$\pr_{q}^{\straa}[M_2[e], \Parity(p)] = 1$ for all $e\in E_2$.
	In particular, $\almostsure(M_1, \Parity(p)) = \almostsure(M_2, \Parity(p))$. 
\end{restatable}

\begin{proof}
	Given state $q$ and finite-memory strategy $\straa$, assume that $\pr_{q}^{\straa}[M_1[e_1], \Parity(p)] = 1$ for all $e_1\in E_1$.
	Consider any $e_2 \in E_2$, and let $e_1\in E_1$ be such that $\Supp(\delta_{e_1}) = \Supp(\delta_{e_2})$; 
	such a $e_2$ exists by the hypothesis $\Supp((\delta_e)_{e \in E_1}) = \Supp((\delta_e)_{e \in E_2})$.
	Consider the Markov chain obtained as the product of the MDP $M_1[e_1]$ with the Moore machine describing the finite-memory strategy $\straa$.
	Because $\pr_{q}^{\straa}[M_1[e_1], \Parity(p)] = 1$, all bottom strongly connected components (BSCC) 
	in this product are winning for $\Parity(p)$ (i.e., the smallest priority of their states is even).
	But the product of $M_2[e_2]$ and the Moore machine for $\straa$ have the same set of BSCCs since the supports are identical.
	It follows that $\pr_{q}^{\straa}[M_2[e_2], \Parity(p)]=1$. By symmetry, this proves the first statement.

	It follows that $\almostsure(M_1, \Parity(p)) = \almostsure(M_2, \Parity(p))$ since
	finite-memory strategies suffice for almost-sure parity in MEMDPs by Theorem~\ref{th:as-parity}.
	
\end{proof}

\begin{restatable}{theorem}{asparityexptime}
	\label{th:as-parity-exptime}
	The membership problem for almost-sure parity for an MEMDP
	$M = \tuple{Q, (A_q)_{q \in Q}, (\delta_e)_{e \in E}}$
	can be solved in time 
    $O((\abs{E}^2+\abs{Q}^4\cdot\abs{E}\cdot\abs{A})\cdot2^{\min(\abs{E}, 2^{\abs{Q}^2\cdot\abs{A}})})$.
\end{restatable}

\begin{proof}
	Consider an MEMDP $M = \tuple{Q, (A_q)_{q \in Q}, (\delta_e)_{e \in E}}$ and a parity objective $\varphi$. If $\abs{E} \leq 2^{\abs{Q}^2\abs{A}}$, then we apply the PSPACE procedure
	from Theorem~\ref{th:as-parity}. The number of recursive calls is then bounded by $2^{\abs{E}}$,
    and each call itself takes polynomial time, so the result follows.

	Otherwise, we scan the set of environments given as input, and store 
	a subset $E'$ of these: we include an environment~$e$ in $E'$ if and only 
	if none of the previously stored environments is equivalent to~$e$.
    This takes $O(\abs{E}^2)$ time.
	This yields a subset with at most $2^{\abs{Q}^2\abs{A}}$ environments,
	with at most one representative for each possible support. 	
	We then apply the recursive algorithm on the MEMDP $M[E']$,
	which yields the same result as if it was applied to $M=M[E]$ by Lemma~\ref{lemma:as-parity-equivalent-envs}.
\end{proof}

\section{Limit-Sure Winning}\label{section:limitsure}
We refer to the examples of the duplicate card and the missing card in Section~\ref{sec:intro} to illustrate the
difference between limit-sure and almost-sure winning. We present in Section~\ref{sec:example-as-ls} two other scenarios where limit-sure winning and almost-sure winning do not coincide, which will be useful to illustrate the key ideas in the algorithmic solution.

\subsection{Examples}\label{sec:example-as-ls}

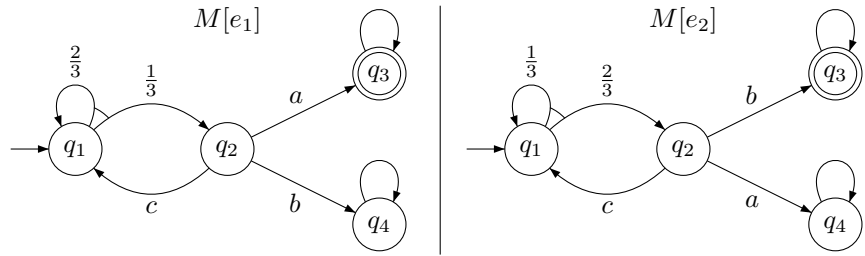
\begin{figure}[!t]
\hrule
\begin{center}
    \input{figures/limit-sure1.tex}
\end{center} 
\hrule
 \caption{An end-component $\{q_1,q_2\}$ with different transition probabilities in environments $e_1$ and $e_2$.\label{fig:limit-sure1}}
\end{figure}

 In the example of \figurename~\ref{fig:limit-sure1},
the set $D = \{q_1,q_2\}$ is an end-component in both environments $e_1$ and $e_2$ (the actions are shown in the figures only
when relevant, that is in $q_2$). However, the transition probabilities from $q_1$
are different in the two environments $e_1$ and~$e_2$, and intuitively we can learn (with high probability) in 
which environment we are by playing $c$ for a long enough (but finite) time
and collecting the frequency of the visits to $q_1$ and $q_2$. 
Then, in order to reach the target $q_3$, if there are more $q_1$'s than $q_2$'s in the history we play $a$ in $q_2$,
otherwise $b$. The intuition is that 
the histories with more $q_1$'s than $q_2$'s have a high probability (more than $1-\epsilon$)
in $M[e_1]$ and a small probability (less than $\epsilon$) in $M[e_2]$,
where $\epsilon$ can be made arbitrarily small (however not $0$) 
by playing $c$ for sufficiently long. Hence $q_1$ is limit-sure winning,
but not almost-sure winning.

In the second scenario (\figurename~\ref{fig:limit-sure2}), the transition
probabilities do not matter. The objective is to visit some state in $\{q_3,q_4,q_5\}$
infinitely often (those states have priority $0$, the other states have priority $1$).
The state $q_1$ is limit-sure winning, but not almost-sure winning. 
To win with probability $1-\epsilon$, a strategy can play $a$ (in $q_2$) 
for a sufficiently long time, then switch to playing $b$ (unless $q_5$ was 
reached before that).
The crux is that playing $a$ does not harm, as it does not leave the limit-sure 
winning region, but ensures in at least one environment (namely, $e_1$)
that the objective is satisfied with probability~$1$ (by reaching $q_5$).
This allows to ``discard'' the environment $e_1$ if $q_5$ was not reached,
and to switch to a strategy that is winning with probability at least $1-\epsilon$ in $e_2$,
namely by playing $b$. 
With an arbitrary number of environments, the difficulty
is to determine in which order the environments can be ``discarded''.

Note that the transition $(q_2,a,q_1)$ is not revealing, since it is present in both environments.
However, after crossing this transition a large number of times, we can still learn that the environment 
is $e_2$ (and be mistaken with arbitrarily small probability).
In contrast, the transition $(q_2,a,q_5)$ is revealing and the environment is $e_1$ with certainty
upon crossing that transition.

To solve the membership problem for limit-sure parity, 
we first convert $M$ into a revealed-form MEMDP $M'$, similar to the case of almost-sure
winning, with the obvious difference that revealing transitions $t=(q,a,q')$ of $M[e]$
are redirected in $M'[e]$ to $\winabsorb$ if $q' \in \limitsure(M[K_t], \varphi)$
is limit-sure winning when the set of environments is the knowledge $K_t$ after observing 
transition $t$. Thus, we aim for a recursive algorithm,  
where the base case is limit-sure winning in MEMDPs with one environment,
which are equivalent to plain MDPs, for which limit-sure and almost-sure parity coincide.
Note that the examples of \figurename~\ref{fig:limit-sure1} and 
\figurename~\ref{fig:limit-sure2} are in revealed form.

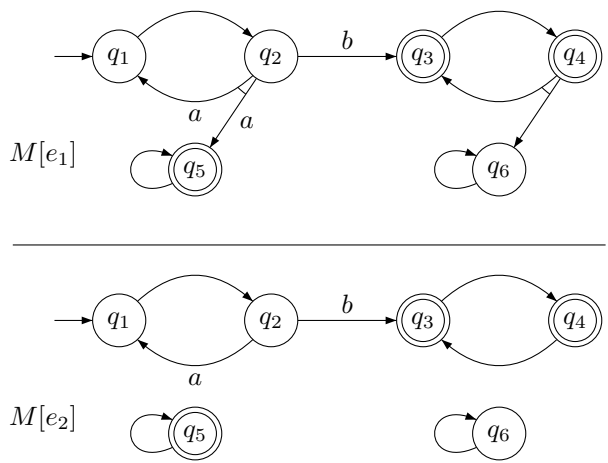
\begin{figure}[!t]
\hrule
\begin{center}
    \input{figures/limit-sure2.tex}
\end{center} 
\hrule
 \caption{The set $\{q_1,q_2\}$ is an end-component in $e_2$, not in $e_1$.\label{fig:limit-sure2}}
\end{figure}

\subsection{Common End-Components and Learning}
\label{sec:cec}
A \emph{common end-component (CEC)} of an MEMDP $M = \tuple{Q, (A_q)_{q \in Q}, (\delta_e)_{e \in E}}$
is a pair $(Q',A')$ that is an end-component in $M[e]$ for all environments $e \in E$. 
A CEC~$D$ is \emph{trivial} if it contains a single state.
$D$ is said \emph{winning} for a parity condition $\Parity(p)$,
if for all $e\in E$, there is a strategy in $M[e]$ which, when started inside $D$,
ensures $\Parity(p)$ with probability~1. Notice that since $D$ is a common end-component,
such a strategy ensures $\Parity(p)$ with probability 1 in $M[e]$ iff it does in $M[e']$.

We note that the common end-components of an MEMDP are the end-components
of the MDP $\cup_{e \in E} M[e]$ assuming $M$ is in revealed form,
and thus can be computed using standard algorithm for end-components~\cite{DeAlfaro-phd97}.
\begin{restatable}{lemma}{cec}
	\label{lemma:cec}
	Consider an MEMDP~$M$ in revealed form.
	The common end-components of~$M$ are exactly the end-components of~$\cup_{e \in E} M[e]$.
\end{restatable}

\begin{proof}
	Consider a common end-component~$D$ of~$M$. Because in each $M[e]$, all state-action pairs in $D$
	stay inside~$D$, and~$D$ is strongly connected, this is also the case in $\cup_{e \in E} M[e]$;
	thus~$D$ is an end-component of the latter.

	Conversely, consider an end-component~$D$ of $\cup_{e \in E} M[e]$.	
	If~$D$ consists of a single sink state, then it is indeed
	a common end-component. Otherwise $D$ contains more than one state.
	We show that all state-action pairs~$(q,a)$ of~$D$ must have the same support in all environments,
	and it follows that $D$ is an end-component in every environment, thus a common end-component.
	By contradiction, if a transition $(q,a,q')$ with $(q,a) \in D$ exists in $M[e]$ but not in~$M[f]$,
	then it is revealing and $q'$ is a sink state. Hence $D$ is not strongly connected in $\cup_{e \in E} M[e]$
	because~$D$ does not consist of a single sink state. 
\end{proof}

A CEC may have different transition probabilities in different environments.
We call a CEC \emph{distinguishing} if it contains a transition $(q,a,q')$ (called a distinguishing transition) such that
$\delta_e(q,a)(q') \neq   \delta_{f}(q,a)(q')$
for some environments $e,f \in E$. 
Given a distinguishing transition $(q,a,q')$ and environment~$e$, define
$K_1 = \{f \in E \mid \delta_f(q,a)(q') = \delta_e(q,a)(q')\}$ and $K_2 = E \setminus K_1$.
We say that $(K_1,K_2)$ is a \emph{distinguishing partition} of~$D$ that is \emph{induced} by the distinguishing 
transition $(q,a,q')$ and environment~$e$.

Distinguishing transitions can be used to learn the partition $(K_1,K_2)$, that is to guess (correctly with high probability) whether the current environment is in $K_1$ or $K_2$, as in the example of \figurename~\ref{fig:limit-sure1}, where the set $D = \{q_1,q_2\}$ is a distinguishing end-component with distinguishing transition $(q_1,\cdot,q_2)$ and partition $(\{e_1\},\{e_2\})$. A distinguishing CEC may have several distinguishing transitions and induced partitions.

We formalize how a strategy can distinguish between $K_1$ and $K_2$ with high probability inside a distinguishing CEC.
First let us recall Hoeffding's inequality.

\begin{theorem}[Hoeffding's Inequality~\cite{Hoeffding-1963}]
	\label{thm:Hoeffding}
	Let $X_1,X_2,\ldots,X_n$ be a sequence of independent and identical Bernoulli variables with $\pr[X_i] = p$,
	and write $S_n = X_1+\ldots +X_n$. For all~$t> 0$,
	\(
		\pr[S_n - \expect[S_n] \geq t] \leq e^{-2t^2/n},
	\)
	and
	\(
		\pr[\expect[S_n] -S_n\geq t] \leq e^{-2t^2/n}.
	\)
\end{theorem}

Given a distinguishing CEC with distinguishing partition  $(K_1,K_2)$ induced by a transition $(q,a,q')$, a strategy can sample the distribution
$\delta_e(q,a)$ by repeating the following two phases: first, use a pure memoryless strategy to almost-surely visit $q$, then play action $a$; by repeating this long enough (the precise bound depends on a given $\epsilon$ and is derived from Theorem~\ref{thm:Hoeffding}) while storing the frequency of visits to $q'$ in the second phase, we can learn and guess in which block $K_i$ belongs the environment, with sufficiently small probability of mistake to ensure winning with probability $1-\epsilon$.

\begin{restatable}{lemma}{strategyfordistmecs}
  \label{lemma:strategy-for-dist-mecs}
  Given an MEMDP~$M$ containing a distinguishing common end-component~$D$ with partition $(K_1,K_2)$
  induced by a distinguishing transition, and parity objective~$\varphi$,
	for all states $q_0$ in~$D$, all pairs of strategies $\sigma_1,\sigma_2$, and all $\epsilon>0$, there exists 
	a strategy~$\sigma$ such that:
  \begin{align*}
		\pr_{q_0}^{\sigma}(M[e], \varphi) \geq (1-\epsilon)\pr_{q_0}^{\sigma_1}(M[e], \varphi) \text{ for all } e \in K_1,\\
		\pr_{q_0}^{\sigma}(M[e], \varphi) \geq (1-\epsilon)\pr_{q_0}^{\sigma_2}(M[e], \varphi) \text{ for all } e \in K_2.
	\end{align*}
	
    Moreover, the strategy $\sigma$ is pure if both $\sigma_i$ are pure; and if each strategy $\sigma_i$ uses a memory of size $m_i$,
	then $\sigma$ uses finite memory of size $m_1+m_2+\lceil 8\frac{\log(1/\epsilon)}{\eta^2}^2\rceil$ where \linebreak $\eta=\min\left(\{\abs{\delta_e(q,a)(q') - \delta_{f}(q,a)(q')} \mid e,f \in E, q,q'\in Q, a \in \Act\}\setminus\{0\}\right)$.
\end{restatable}

\begin{proof}
	Consider $M=\tuple{Q,(A_q)_{q \in Q}, (\delta_e)_{e \in E}}$ and $D = (Q',(A'_q)_{q \in Q'})$ as in the statement of the lemma,
	and let $q_0 \in D$. 
        
	Consider a distinguishing transition $(q,a,q')$ and environment~$e_0$ that induces the distinguishing partition $(K_1,K_2)$.
	Consider~$\epsilon>0$, and define~$N = \lceil\frac{2\log(1/\epsilon)}{\eta^2}\rceil$,

	The strategy~$\sigma$ runs in two phases. In the first phase, the
	goal is to estimate the distribution of $(q,a,q')$. For this, it executes
	a pure memoryless strategy which has a nonzero probability of reaching $q$
	while staying in $D$ (such a strategy can be defined based on the supports of state-action pairs of $D$)
	and keeps two counters: $c_{q,a}$ that counts the number of times the state-action pair $(q,a)$ is selected;
	and $c_{q,a,q'}$ the number of times the transition $(q,a,q')$ is observed.
	The second round of the strategy starts when $c_{q,a} = N$. Note that this happens with probability 1.
	Then, we go back to~$q_0$ (with probability 1), and we switch to 
	\begin{itemize}
		\item $\sigma_1$ if $\left\lvert\frac{c_{q,a,q'}}{c_{q,a}} - \delta_{e_0}(q,a)(q')\right\rvert < \eta/2$,
		\item $\sigma_2$ otherwise.
	\end{itemize}
	We now analyze this strategy and show that because $N$ is sufficiently large, 
	the estimation error is bounded, so that we obtain the desired result.
	
	In each environment~$e$, at each visit at~$q$ and choice of~$a$, we have a
	Bernoulli trial with mean $\delta_e(q,a)(q')$, and~$c_{q,a,q'}$ is the
	number of successful trials.
	By Hoeffding's inequality (Theorem~\ref{thm:Hoeffding}), we have
	\[
		\pr_{q_0}^{\sigma}\left(M[e], \left\lvert c_{q,a,q'}/c_{q,a} - \delta_e(q,a)(q')\right\rvert \geq 
		\eta/2 \bigm\vert c_{q,a} = N \right)
		\leq e^{-2N(\frac{\eta}{2})^2} \leq \epsilon.
	\]
	Thus, in $M[e]$ with $e \in K_i$, the probability of not switching to $\sigma_i$ is at most $\epsilon$.
	It follows that $\pr_{q_0}^{\sigma}(M[e],\varphi) \geq (1-\epsilon)\pr_{q_0}^{\sigma_i}(M[e], \varphi)$.

	The memory requirement comes from the fact that $\sigma$ must store two counters up to $N$ values, and it has two modes
	(before and after reaching $c_{q,a} = N$).
	\end{proof}

It follows that the membership problem for limit-sure winning can be decomposed into subproblems where the set of environments is one of the blocks $K_i$ in the partition.

\begin{restatable}{lemma}{strategyformecas}
  \label{lemma:strategy-for-mec-as}
  Given an MEMDP~$M$ containing a distinguishing common end-component $D$ with a partition $(K_1,K_2)$
  induced by a distinguishing transition, and a parity objective $\varphi$ the following equivalence holds:
  $D \subseteq \limitsure(M,\varphi)$ if and only if $D \subseteq \limitsure(M[K_1],\varphi)$ and 
  $D \subseteq \limitsure(M[K_2],\varphi)$. 
\end{restatable}

\begin{proof}
	Immediate consequence of Lemma~\ref{lemma:strategy-for-dist-mecs}.
\end{proof}

\subsection{Characterization and Algorithm}

Here, we assume that MEMDPs are in revealed form
with sink states $\winabsorb$ and $\loseabsorb$.

We show that the winning region $W = \limitsure(M,\varphi)$ 
for limit-sure parity is a closed set: from every state $q \in W$,
there exists an action $a$ ensuring in all environments
that all successors of~$q$ are in $W$.
We call such actions \emph{limit-sure safe} for $q$.
We show in Lemma~\ref{lemma:only11-inf} that a limit-sure safe action always exists in 
limit-sure winning states. 
Note that playing actions that are \emph{not} limit-sure safe may be useful 
for limit-sure winning, as in the example of \figurename~\ref{fig:limit-sure1} 
where action $a$ is limit-sure safe, but action $b$ is not (from $q_2$).

By definition of limit-sure winning, if a state $q$ is not limit-sure winning,
there exists $\epsilon_q > 0$ such that for all strategies $\sigma$, there
exists an environment $e \in E$ such that $\pr_{q}^{\sigma}(M[e], \varphi) < 1-\epsilon_q$.
We denote by $\epsilon_0 = \min \{\epsilon_q \mid q \in Q \setminus \limitsure(M,\varphi)\}$
a uniform bound.

\begin{restatable}{lemma}{onlyinf}
\label{lemma:only11-inf}
	Given an MEMDP~$M$ (in revealed form) over environments $E$, a parity objective~$\varphi$, and a state~$q$,
	if~$q \in \limitsure(M,\varphi)$ is limit-sure winning, then there exists an action~$a$ such that
	for all environments $e\in E$, all successors of $q$ are limit-sure winning, 
	\textit{i.e} $\Supp(\delta_e(q,a))\subseteq \limitsure(M,\varphi)$.
\end{restatable}

\begin{proof}
	Consider $q \in \limitsure(M,\varphi)$ and let $0 < \epsilon < \frac{\nu \epsilon_0}{\abs{A}}$, 
	where $A$ is the set of actions in $M$, and $\nu$ is a lower bound on the smallest nonzero transition probability (in all environments),
	and $\epsilon_0$ is the uniform bound defined above.
	Let $\sigma$ be a strategy ensuring $\varphi$ from~$q$ with probability at least $1-\epsilon$ in all environments.
		
	Towards contradiction, assume that there is no limit-sure safe action from state $q$.
	Let $a$ be the action chosen by $\sigma$ with the highest probability at the history~$q$, that is $a = \arg\max_a \sigma(q)(a)$, and
	thus $\sigma(q)(a) \geq \frac{1}{\abs{A}}$.
	By our assumption, there exists an environment $e \in E$ and a state 
	$t \not \in \limitsure(M,\varphi)$ (in particular $t \neq \winabsorb$) such that 
	$\delta_e(q,a)(t) > 0$, hence $\delta_e(q,a)(t) \geq \nu$. It is immediate that $t \neq \loseabsorb$
	as otherwise the strategy $\sigma$ would ensure $\varphi$ with probability at most $1-\nu \leq 1-\epsilon$ from $q$.
	So $t \not\in \{\winabsorb,\loseabsorb\}$ and therefore $\delta_e(q,a)(t) \geq \nu$ in all environments~$e$.
	By definition of the uniform bound $\epsilon_0$, there exists an environment $e$ such that 
	$\pr_{t}^\sigma(M[e],\varphi) \leq 1-\epsilon_0$, 
	hence from $q$ we have $\pr_{q}^\sigma(M[e],\lnot \varphi) \geq \frac{\nu\epsilon_0}{\abs{A}} > \epsilon$,
	in contradiction to $\sigma$ ensuring $\varphi$ with probability at least $1-\epsilon$ from~$q$.
	We conclude that there exists a limit-sure safe action from $q$.
\end{proof}

Given an MEMDP $M$, consider the limit-sure winning region $W = \limitsure(M,\varphi)$ 
for $\varphi = \Parity(p)$.
For the purpose of the analysis, consider the (memoryless) randomized strategy $\sigma_{\limitsure}$ that plays uniformly at random 
all limit-sure safe actions in every state $q \in W$, 
which is well-defined by Lemma~\ref{lemma:only11-inf}. 

Consider an arbitrary environment $e$, and an end-component $D$ in $M[e]$ that is positive under $\sigma_{\limitsure}$
(recall Lemma~\ref{lem:basic-ec} and the definition afterward). There are three possibilities:

\vspace{-0.1cm}
\begin{enumerate}
\item $D$ is not a common end-component (as in the example of \figurename~\ref{fig:limit-sure2}, for $D = \{q_1,q_2\}$ in $M[e_2]$),
that is, $D$ is not an end-component in some environment $e'$ (in the example $e' = e_1$), then we can 
learn (and be mistaken with arbitrarily small probability) that we are not in $e'$,
reducing the problem to an MEMDP with fewer environments (namely, $M[\lnot e']$);

\item $D$ is a common end-component and is distinguishing (as in the example of \figurename~\ref{fig:limit-sure1}, 
for $D = \{q_1,q_2\}$), then we can also learn a distinguishing partition $(K_1,K_2)$
and reduce the problem to MEMDPs with fewer environments (namely, $M[K_1]$ and $M[K_2]$);

\item $D$ is a common end-component and is non-distinguishing, then 
we show in Lemma~\ref{lem:ndCEC-is-winning} below that $D$ is almost-sure winning 
($D \subseteq \almostsure(M, \varphi)$), obviously in all environments.
\end{enumerate}
\vspace{-0.2cm}

\begin{restatable}{lemma}{ndCECiswinning}
\label{lem:ndCEC-is-winning}
	Given an MEMDP~$M$ over environments $E$ (in revealed form), a parity objective~$\varphi$, and a state~$q$,
	if~$q \in \limitsure(M,\varphi)$, then all non-distinguishing common end-components $D$
	that are positive under strategy $\sigma_{\limitsure}$ from $q$ in $M[e]$ (for some $e \in E$)
	are almost-sure winning for $\varphi$ (that is $D \subseteq \almostsure(M,\varphi)$).
\end{restatable}

\begin{proof}
	Consider a positive non-distinguishing common end-component $D$ as in 
	the statement of the lemma. Using Lemma~\ref{lemma:only11-inf}, note that $D \subseteq \limitsure(M,\varphi)$
	since $\sigma_{\limitsure}$ plays only limit-sure safe actions and $D$ is a common end-component.

	Assume towards contradiction that $D$ is not almost-sure winning for the parity objective $\varphi$.
	It follows that in $M$, all strategies that play only limit-sure safe 
	actions ensure the parity objective $\varphi$ with probability~$0$
	from all states in $D$ (in all environments since $D$ is a common
	end-component).

	Denote by $\Omega_{safe}$ the set of all runs that contain only limit-sure safe actions.
	For all strategies $\sigma$ (in $M$), and $q \in D$ we have 
	$\pr_{q}^{\sigma}(M[e], \varphi \mid \Omega_{safe}) = 0$ (for all $e\in E$)
	and therefore:

	\begin{align*}
		\pr_{q}^{\sigma}(M[e], \varphi)  = \ & \pr_{q}^{\sigma}(M[e], \varphi \mid \Omega_{safe})
		\cdot \pr_{q}^{\sigma}(M[e], \Omega_{safe}) \\
		& + \pr_{q}^{\sigma}(M[e], \varphi \mid \lnot \Omega_{safe})
		\cdot \pr_{q}^{\sigma}(M[e], \lnot \Omega_{safe}) \\
		= \ & \pr_{q}^{\sigma}(M[e], \varphi \mid \lnot \Omega_{safe})
		\cdot \pr_{q}^{\sigma}(M[e], \lnot \Omega_{safe}) \\
		\leq  \ & 1 - \pr_{q}^{\sigma}(M[e], \lnot \varphi \mid \lnot \Omega_{safe})
	\end{align*}

	Given $\epsilon < \frac{\epsilon_0 \cdot \nu }{\abs{E}}$ 
	where $\nu$ is the smallest positive probability in $M$,
	we show that there exists an environment $e \in E$
	such that $\pr_{q}^{\sigma}(M[e], \varphi) < 1-\epsilon$, which entails that 
	$q$ is not limit-sure winning for $\varphi$, establishing a contradiction since
	$q \in D \subseteq \limitsure(M,\varphi)$. It will follow that $D$ is almost-sure winning for $\varphi$ and conclude the proof.

	By definition of limit-sure safe actions,
	to every pair $(q,a)$ such that $a \in A_q$ is not limit-sure safe in $q$, 
	we can associate an environment $e$ such that: 
	$$ \Supp(\delta_{e}(q,a)) \cap (Q \setminus \limitsure(M, \varphi)) \neq \emptyset,$$
	and thus from some state $q' \in \Supp(\delta_{e}(q,a))$, we have
	$\pr_{q'}^{\sigma}(M[e], \varphi) \leq 1-\epsilon_0$
	where $\epsilon_0$ is the uniform bound for non-limit-sure winning states.
	Assuming that a non-limit-sure safe action is played by $\sigma$, since there are finitely 
	many environments, by the pigeonhole principle there is an environment $e$ 
	such that with probability at least $\frac{1}{\abs{E}}$ an action
	that is not limit-sure safe and associated with $e$ is played,
	which leads with probability at least $\nu$ to a state outside $\limitsure(M, \varphi)$. 
	It follows that 
	$\pr_{q}^{\sigma}(M[e], \lnot \varphi \mid \lnot \Omega_{safe}) 
	\geq \epsilon_0 \cdot \frac{\nu}{\abs{E}} > \epsilon$
	and thus $\pr_{q}^{\sigma}(M[e], \varphi) < 1-\epsilon$,
	which concludes the proof.
\end{proof}

Our approach to compute the limit-sure winning states is to first identify the
distinguishing CECs that are limit-sure winning. We can compute
the maximal CECs using Lemma~\ref{lemma:cec}, and note that
a maximal CEC containing a distinguishing CEC is itself distinguishing, so it is sufficient
to consider maximal CECs. By Lemma~\ref{lemma:strategy-for-mec-as}, we can decide
if a given distinguishing CEC is limit-sure winning using 
a recursive procedure on MEMDPs with fewer environments.
We show in Lemma~\ref{lemma:remove-lsd} below that we can replace the limit-sure 
CECs by a sink state $\winabsorb$.

\begin{restatable}{lemma}{removelsd}
\label{lemma:remove-lsd}
	Given an MEMDP~$M$ with parity objective $\varphi$
	and a set $T \subseteq \limitsure(M, \varphi)$ of limit-sure winning states, 
	we have
	\(\limitsure(M, \varphi) = \limitsure(M, \varphi \cup \Reach(T)).\)
\end{restatable}

\begin{proof}
	The inclusion $\limitsure(M, \varphi) \subseteq \limitsure(M, \varphi \cup \Reach(T))$
	is immediate since $\varphi \subseteq \varphi \cup \Reach(T)$.

	To show the converse inclusion, consider $q \in \limitsure(M, \varphi \cup \Reach(T))$
	and show that $q\in \limitsure(M, \varphi)$. 
	Given $\epsilon > 0$, let $\epsilon_1 = \frac{\epsilon}{2}$ and
	let $\sigma$ be a strategy such that 
	$\pr_{q}^{\sigma}(M, \varphi \cup \Reach(T)) \geq 1-\epsilon_1$.
	We construct a strategy $\tau$ that satisfies the objective $\varphi$ 
	with probability at least $1-\epsilon$ as follows:
	for all histories $\rho$, if $\rho$ does not visit $T$, then let $\tau(\rho) = \sigma(\rho)$;
	otherwise, consider the suffix $\rho'$ of $\rho$ after the first visit to a state $t \in T$, 
	and let~$\sigma_t$ be strategy that ensures $\varphi$ is satisfied with probability 
	at least $1-\epsilon_1$ from $t$ (such a strategy exists since $T \subseteq \limitsure(M, \varphi)$).
	Define $\tau(\rho) = \sigma_t(\rho')$.
	We easily show below that $\pr_{q}^{\tau}(M, \varphi) \geq 1-\epsilon$, 
	establishing that $q\in \limitsure(M, \varphi)$:
	\begin{align*}
	\pr_{q}^{\tau}(M, \varphi) & = \pr_{q}^{\tau}(M, \varphi \cap \Reach(T)) + \pr_{q}^{\tau}(M, \varphi \cap \lnot \Reach(T)) \\
	& = \pr_{q}^{\tau}(M, \varphi \mid \Reach(T)) \cdot \pr_{q}^{\tau}(M, \Reach(T)) 
	+ \pr_{q}^{\tau}(M, \varphi \cap \lnot \Reach(T)) \\
	& = \pr_{q}^{\tau}(M, \varphi \mid \Reach(T)) \cdot \pr_{q}^{\sigma}(M, \Reach(T)) 
	+ \pr_{q}^{\sigma}(M, \varphi \cap \lnot \Reach(T)) \\
	& \qquad\qquad\qquad\quad\  \text{(since $\tau$ agrees with $\sigma$ as long as $T$ is not reached)} \\
	& \geq (1-\epsilon_1)  \cdot \pr_{q}^{\sigma}(M, \Reach(T)) 
	+ \pr_{q}^{\sigma}(M, \varphi \cap \lnot \Reach(T)) \\
	& \geq (1-\epsilon_1)  \cdot \pr_{q}^{\sigma}(M, \Reach(T)) 
	+ (1-\epsilon_1)  \cdot \pr_{q}^{\sigma}(M, \varphi \cap \lnot \Reach(T)) \\
	& \geq (1-\epsilon_1)  \cdot \pr_{q}^{\sigma}(M, \varphi \cup \Reach(T)) \geq (1-\epsilon_1)^2   \geq 1-\epsilon.
	\end{align*}
\end{proof}

We can now assume that MEMDPs contain no limit-sure winning distinguishing CEC,
and present a characterization for the remaining possibility, illustrated by the scenario of 
\figurename~\ref{fig:limit-sure2}, where playing the action $a$ (in $q_2$, forever) 
ensures, in some environment (namely, $e_1$), almost-sure satisfaction of 
the parity objective while remaining inside the limit-sure winning region in all other environments. 

\begin{restatable}{lemma}{lsparitycharac}
\label{lemma:ls-parity-charac}
Consider an MEMDP $M$ (in revealed form) over environments~$E$ with $\abs{E} \geq 2$,
that contains no limit-sure winning distinguishing common end-component, and a parity objective $\varphi$. 
Writing $T_e = \limitsure(M[\lnot e], \varphi)$, we have the following:
\(
   \limitsure(M, \varphi) = \almostsure\Big(M, \Reach\Big(\bigcup_{e \in E} \almostsure(M[e], \varphi \cap \Safe(T_e)) \Big)\Big).
\)

\end{restatable}

\begin{proof}
First we show the inclusion 
\begin{center}
$\limitsure(M, \varphi) \subseteq \almostsure(M, \Reach(\bigcup_{e \in E} \almostsure(M[e], \varphi \cap \Safe(T_e)))).$
\end{center}
Consider the (memoryless) strategy $\sigma_{\limitsure}$ that plays all 
limit-sure safe actions uniformly at random from every state in $\limitsure(M, \varphi)$. 
The strategy $\sigma_{\limitsure}$ is well-defined by Lemma~\ref{lemma:only11-inf} and 
to establish the inclusion, we show that, from every state $q \in \limitsure(M, \varphi)$,
it is almost-sure winning (in all environments $e' \in E$) for the objective 
$\Reach(\bigcup_{e \in E} \almostsure(M[e], \varphi \cap \Safe(T_e)))$. 

Consider an arbitrary environment $e' \in E$ and an arbitrary 
end-component $D$ that is positive under $\sigma_{\limitsure}$ in $M[e']$. 
Since positive end-components are reached with probability~$1$ (Lemma~\ref{lem:basic-ec}), it is sufficient
to show that for all such $D$, there exists an environment $e \in E$ such that 
every state in $D$ is almost-sure winning for the objective $\varphi \cap \Safe(T_e)$
in $M[e]$. We consider two cases:

\begin{itemize}
\item if $D$ is a common end-component, then we show that $D$ is non-distinguishing.
Note that $D$ must be limit-sure winning, by definition of limit-sure safe actions (played by $\sigma_{\limitsure}$).
It follows by the assumption of the lemma that $D$ is non-distinguishing
and therefore almost-sure winning for $\varphi$ (in all environments) by Lemma~\ref{lem:ndCEC-is-winning}.
We take $e = e'$ and it is easy to see that there exists an almost-sure winning strategy 
for $\varphi$ from $D$ (that stays in $D$), which is also almost-sure winning 
for $\varphi \cap \Safe(T_e)$.

\item otherwise $D$ is not a common end-component, and there exists an environment $e$
where $D$ is not an end-component. 
We first show that all transitions of $D$ are present in $M[e]$, since otherwise $D$ 
would contain a revealing transition, thus leading to a state that is a sink in all environments (revealed form). 
Then $D$ being strongly connected would not contain another state, and thus in particular
all transitions in $D$ would be present in $M[e]$. 

It follows that playing $\sigma_{\limitsure}$ from $D$ in $M[e]$ ensures with
probability~$1$ that a (revealing) transition not present in $M[e']$ is executed,
which leads to $\winabsorb$ since $\sigma_{\limitsure}$ never leaves the 
limit-sure winning region (by definition of limit-sure safe actions).
Hence $\varphi$ is satisfied with probability~$1$ in $M[e]$ while playing
only limit-sure safe actions, thus remaining in the limit-sure winning region
$\limitsure(M, \varphi) \subseteq \limitsure(M[\lnot e], \varphi) = T_e$, 
thereby satisfying $\Safe(T_e)$ as well.
This shows that in $M[e]$, the states in $D$ are almost-sure winning 
for the objective $\varphi \cap \Safe(T_e)$.
\end{itemize}

For the converse inclusion, given a state $q$ and a pure\footnote{By Theorem~\ref{th:as-parity}, pure strategies are sufficient for almost-sure winning in MEMDPs.} strategy $\sigma$ 
that is almost-sure winning for objective 
$\Reach(\bigcup_{e \in E} \almostsure(M[e], \varphi \cap \Safe(T_e)))$
(in all environments), 
we show that for all $\epsilon > 0$ there is a pure strategy $\tau$
that ensures that $\varphi$ is satisfied with probability at least $1-\epsilon$
(from $q$ in all environments).

Given $\epsilon > 0$, let $\tau$ be the strategy that plays as follows:
\begin{enumerate}[label=(\theenumi)]
\item play like $\sigma$ until a state $t \in \bigcup_{e \in E} \almostsure(M[e], \varphi \cap \Safe(T_e))$
is reached, and let $e \in E$ be an environment
such that from $t$ there is a (pure memoryless) strategy $\sigma_t$
that is almost-sure winning in $M[e]$ for the objective $\varphi \cap \Safe(T_e)$; \label{tau-phase1}

\item play like $\sigma_t$ for $k \cdot \abs{Q}$ steps, where $k$ is such that
$(1-\nu^{\abs{Q}})^k \leq \epsilon$ (where $\nu$ is the smallest positive 
probability in $M$); \label{tau-phase2}

\item if the current state belongs to a positive end-component $D_t$ of $\sigma_t$ (in $M[e]$),
then keep playing like $\sigma_t$ (forever); otherwise switch to a strategy 
that ensures that $\varphi$ is satisfied with probability at least $1-\epsilon$
from the current state in all environments of $E \setminus \{e\}$ -- such a strategy 
exists because from $t$ the strategy $\sigma_t$ ensures the objective $\Safe(T_e)$
is satisfied almost-surely (and thus surely as well). \label{tau-phase3}
\end{enumerate}

Consider an arbitrary environment $e \in E$,
and show that $\pr_{q}^{\tau}(M[e], \varphi) \geq 1-\epsilon$,
which establishes that $q$ is limit-sure wining, $q \in \limitsure(M, \varphi)$.

First note that phase~\ref{tau-phase2} (and thus also phase~\ref{tau-phase3}) 
is reached with probability~$1$, and let $e_t$ be the environment corresponding to the state $t$ 
reached at the end of phase~\ref{tau-phase1}. We consider two cases:

\begin{itemize}
\item if $e_t = e$, then by standard analysis the probability that 
after phase~\ref{tau-phase2} a positive end-component of $\sigma_t$ is \emph{not yet} reached
is at most $(1-\nu^{\abs{Q}})^k \leq \epsilon$ since within $\abs{Q}$
steps a positive end-component is reached with probability at least 
$\nu^{\abs{Q}}$. Hence with probability at least $1-\epsilon$,
a positive (winning since $\sigma_t$ almost-sure winning in $M[e]$ for 
the objective $\varphi$) end-component of $\sigma_t$ is reached and
the strategy $\sigma_t$ is played forever in phase~\ref{tau-phase3},
thus winning with probability at least $1-\epsilon$.

\item otherwise $e_t \neq e$ and we consider the following cases in 
phase~\ref{tau-phase3}: 
\begin{itemize}
\item[$(a)$] if the strategy $\sigma_t$ is played forever, then either
the set $D_t$ (which is an end-component in $M[e_t]$) is never left, 
or it is left (via a revealing transition, as $D_t$
is not left in $M[e_t]$) and since $\sigma_t$ ensures $\Safe(T_{e_t})$
the sink $\winabsorb$ is reached in $M[e]$, thus in both cases
the objective $\varphi$ is satisfied (with probability~$1$); 
\item[$(b)$] otherwise, by construction the strategy $\tau$ switches to a strategy 
that ensures $\varphi$ is satisfied with probability at least $1-\epsilon$.
\end{itemize}

In all cases, the objective $\varphi$ holds with probability at least $1-\epsilon$,
showing that $\pr_{q}^{\tau}(M[e], \varphi) \geq 1-\epsilon$ as claimed.
\end{itemize}
\end{proof}

\begin{algorithm}[t]
\caption{${\sf LS\_Parity}(M,p)$ \label{alg:solve-ls-ws}}
{
  \AlgData{$M = \tuple{Q, (A_q)_{q \in Q}, (\delta_e)_{e \in E}}$ an MEMDP, $p: Q \to \nat$ a priority function.}
 \AlgResult{The winning region $\limitsure(M,\Parity(p))$ for limit-sure parity.}

  \Begin{
       \nl \lIf{$\abs{E}=1$}{\KwRet{$\almostsure(M, \Parity(p))$}} \;
       \hfill {\tt /* pre-processing */} \;
       \nl put $M$ in revealed form (defined in Section~\ref{section:as}) \;
       \nl $\mecs \gets$ maximal end-components of the MDP $\cup_{e \in E} M'[e]$ \;

       \nl \For{$D \in \mecs$}
       {
	 \nl \If{$D$ is distinguishing in $M$}
	 {
	   \nl Let $(K_1,K_2)$ be a distinguishing partition in $D$ \;
   	     \nl \If{$D \subseteq {\sf LS\_Parity}(M[K_1],p) \cap {\sf LS\_Parity}(M[K_2],p)$}
   	     {
   	        \nl replace $D$ by sink $\winabsorb$ in $M$ with $p(\winabsorb)=0$\;
   	     }
	 }
       }

       \hfill {\tt /* $M$ is in revealed form and Lemma~\ref{lemma:ls-parity-charac} applies */} \;

       \nl \For{$e \in E$}
       {
          \nl $T_e = {\sf LS\_Parity}(M[\lnot e],p)$ \;
       }
       
       \nl $Q \gets \almostsure\Big(M, \Reach\Big(\bigcup_{e \in E} \almostsure(M[e], \Parity(p) \cap \Safe(T_e)) \Big)\Big)$ \; 

       \nl \KwRet{$Q \setminus \{\winabsorb\}$} \;
  }
}
\end{algorithm}

\myparagraph{Algorithm Overview}
Given a MEMDP~$M=(Q,(A_q)_{q\in Q},(\delta_e)_{e\in E})$, the algorithm proceeds 
by recursion on the size of the environment set~$E$ (Algorithm~\ref{alg:solve-ls-ws}).
The base case is that of a  singleton set~$E$ where $\limitsure(M, \varphi)=\almostsure(M, \varphi)$
and this can be computed in polynomial time. 

Assume $\abs{E}\geq 2$.
We first convert $M$ into an MEMDP $M'$ 
in revealed form with state space $Q \uplus \{\winabsorb, \loseabsorb \}$ and 
each revealing transition $t=(q,a,q')$ in $M$ is redirected in $M'$ to $\winabsorb$
if $q' \in \limitsure(M[K_t], \varphi)$ is limit-sure winning when the 
set of environments is the knowledge $K_t$ after observing transition $t$,
and to $\loseabsorb$ otherwise. Notice that each query $q' \in \limitsure(M[K_t], \varphi)$
uses a set~$K_t$ that is strictly smaller than~$E$.

We now assume that $M$ is in revealed form and 
we compute the maximal end-components of the MDP $\cup_{e \in E} M[e]$; 
these are maximal common end-components of~$M$ by Lemma~\ref{lemma:cec}.
For each distinguishing maximal CEC~$D$,
we determine whether it is limit-sure winning 
using the condition of Lemma~\ref{lemma:strategy-for-mec-as},
namely that $D \subseteq \limitsure(M[K_i], \varphi)$ (for $i=1,2$)
where $(K_1,K_2)$ is a partition of $E$ induced by a 
distinguishing transition of $D$, which is computed
by a recursive calls to the algorithm.
We replace $D$ by $\winabsorb$ if it is limit-sure winning,
which yields an MEMDP without limit-sure winning distinguishing CECs,
and we can apply Lemma~\ref{lemma:ls-parity-charac}:
for each environment~$e \in E$, we compute $T_e = \limitsure(M[\lnot e], \varphi)$ which is done by
$\abs{E}$ separate recursive calls, and we compute the sets 
$\almostsure(M[e], \varphi \cap \Safe(T_e))$ using standard MDP algorithms
(we restrict the state space to $T_e$ and compute the almost-sure winning states for $\varphi$).
We then solve the almost-sure reachability problem in $M$ for the target set 
$\bigcup_{e\in E}\almostsure(M[e], \varphi \cap \Safe(T_e))$.

Thus each recursive step takes polynomial time (besides the recursive calls), 
and because each recursive call decreases the size of~$E$,
the depth of the recursion is bounded by $\abs{E}$.
It follows that the procedure runs in polynomial space.
The PSPACE lower bound follows from the same reduction as for 
almost-sure winning~\cite[Theorem~7]{SVJ24}, since the MEMDP constructed in the reduction
is acyclic, thus almost-sure and limit-sure winning coincide.

Note that Lemma~\ref{lemma:ls-parity-charac}
constructs a pure strategy that achieves the objective with probability at least $1-\epsilon$
from the limit-sure winning states, and that the strategies constructed in
Lemmas~\ref{lemma:strategy-for-mec-as} and \ref{lemma:remove-lsd} to witness
limit-sure winning are also pure (in Lemma~\ref{lemma:strategy-for-mec-as},
the construction assumes that pure strategies are sufficient for fewer environments, 
which allows a proof by induction since pure strategies are sufficient in MDPs, \textit{i.e.} in a single environment).

\begin{theorem}
	\label{thm:ls-parity}
	The membership problem for limit-sure parity objectives in MEMDPs is PSPACE-complete and 
	pure exponential-memory strategies are sufficient, i.e., if a state $q$ is limit-sure winning, then for all $\epsilon >0$
	there exists a pure exponential-memory strategy that ensures the objective is satisfied with probability at least $1-\epsilon$
	from $q$.
	When the number of environments is fixed, the problem is solvable in polynomial time.
\end{theorem}

The time complexity of Algorithm~\ref{alg:solve-ls-ws} is established as follows. 
Let us consider a single recursive call.
The maximal end-components of $\cup_{e \in E} M'[e]$ can be computed
in $O(\abs{Q}\cdot\abs{\delta})$ where $\abs{\delta}$ denotes the number of transitions. Then, determining whether each MEC is distinguishing, and replacing them with sink states can be done in time $O(\abs{\delta}\cdot \abs{E})$
since one needs to go over each transition and check whether their probability differs in two  environments. The last step requires solving almost-sure parity and safety for MDPs defined for each $e \in E$, which can be done in time $O(\abs{E}\cdot \abs{Q}\cdot \abs{\delta})$
(similarly as in the discussion following Theorem~\ref{th:as-parity}).
The most costly operation is almost-sure reachability for the MEMDP $M$, which 
by Theorem~\ref{th:as-parity} takes 
$O(\abs{Q}^4\cdot \abs{E}\cdot \abs{A}\cdot 2^{\abs{E}})$.
There are $2^{\abs{E}}$ recursive calls
(the algorithm can be run once for each subset of $E$ using memoization),
so overall 
we get $O(\abs{Q}^4\cdot \abs{E}\cdot \abs{A}\cdot 2^{2\abs{E}})$.

We do not know if a technique similar to that of Theorem~\ref{th:as-parity-exptime} can be used for the limit-sure case to obtain an exponent independent of $\abs{E}$.

	\section{The Gap Problem}\label{section:quantitative}
	The goal of this section is to give a procedure that solves the gap problem for parity objectives.
	For this, we show that an arbitrary strategy in $M$ can be imitated by a finite-memory one (with a computable bound on the memory size) 
	while achieving the same probability of winning up to $\epsilon$ in all environments.
	Once this is established, we show how to guess a finite-memory strategy of the appropriate size in order to solve the gap problem.

To establish the memory bound for such an $\epsilon$-approximation, we need a few intermediate lemmas. First, we define a transformation on MEMDPs 
consisting in collapsing non-distinguishing maximal CECs (MCECs) of the MEMDP~$M$; the resulting MEMDP is denoted $\purge{M}$.
We show that $M$ and $\purge{M}$ have the same probabilities of satisfaction of the considered parity objective under all environments.

Intuitively, removing non-distinguishing MCECs ensures that in $\purge{M}$, under all strategies, with high probability, within a fixed number of steps,
either a maximal CEC is reached (which is either distinguishing, or non-distinguishing but trivial  -- recall that a trivial CEC contains a single absorbing state.)
or enough samples are gathered to improve the knowledge about the current environment, as shown in Section~\ref{section:learning-while-playing}
This observation will help us constructing the finite-memory strategy inductively since in each case the knowledge can be improved correctly with high probability:
in trivial MCECs, the strategy is extended arbitrarily; inside distinguishing MCECs, the strategy can be extended so that it stays inside the MCEC while sampling distinguishing transitions with any desired precision as in Lemma~\ref{lemma:strategy-for-dist-mecs};
last, if no MCECs are reached but enough samples are gathered along the way, we prove that the knowledge can also be improved with high probability. 
The final strategy is obtained by combining finite-memory strategies constructed inductively for smaller sets of environments. 
This is done in Section~\ref{section:constructing-finite-memory-strategies}

\paragraph*{Maximal Common End-Components Revisited}
We extend 
the definition of common end-components (CEC) which, in Section~\ref{sec:cec}, were defined assuming MEMDPs are in revealed form. In this section, MEMDPs are \textbf{not} assumed to be in revealed form: in fact, upon observing a revealing transition, we cannot conclude recursively since we cannot determine which value vector must be achieved in the recursive call.
Here, we define a CEC for MEMDP $M=\tuple{Q,A,(\delta_e)_{e\in E}}$ as a pair $(Q',A')$ such that for all $e\in E$, 
$\tuple{Q',A',\delta_e}$ is an end-component of $M[e]$.
A maximal CEC (MCEC) is a CEC which does not contain a smaller CEC.

There are two types of MCECs:
\begin{itemize}
  \item MCEC $(Q',A')$ is non-distinguishing if for all $q \in Q'$, and $a \in A'(q)$, the distributions $\delta_e(q,a)$ and $\delta_{e'}(q,a)$ are identical
  for all $e,e' \in E$;
  \item MCEC $(Q',A')$ is distinguishing otherwise.
\end{itemize}
While non-distinguishing MCECs have state-action pairs with identical supports in all environments,
a distinguishing MCEC may contain revealing transitions, that is, state-action pairs $(q,a)$ with different supports in different environments.
This is the difference with Section~\ref{section:limitsure}.
The only result we need from Section~\ref{sec:cec} is Lemma~\ref{lemma:strategy-for-dist-mecs} which holds for the new definition of distinguishing MCECs:
in fact, we do require that $(Q',A')$ is an end-component (i.e., closed and strongly connected) in all environments, 
so revealing transitions are simply seen as distinguishing transitions, and thanks to the strong connectivity of $(Q',A')$ in all environments, one can 
define a strategy that samples a distinguishing transition a desired number of times.

As previously, a MCEC~$D$ is \emph{trivial} if it contains a single state.

In terms of computability, we cannot use Lemma~\ref{lemma:cec} to compute MCECs since this is only valid for MEMDPs in revealed form.
The $\epsilon$-gap procedure given in this section does not actually compute MCECs; these are only used in the proof of the existence of a finite-memory strategy (Lemma~\ref{lemma:quantitative-bounded-memory}).
Nevertheless, for completeness, let us describe how MCECs can be computed in polynomial time.
For $\abs{E}=1$, the MCECs are exactly the maximal end-components (MECs) of $M[e]$ where $E=\{e\}$. For $\abs{E}\geq 2$, we pick an environment $e \in E$, and compute the MECs of $M[e]$.
For each MEC $D$ of $M[e]$, we recursively compute the MCECs of $D$ in the MEMDP $M[E\setminus\{e\}]$. 
This is sound because a MCEC, being an end-component in all environments, is necessarily a subset of some MEC in each $M[e]$; so by restricting the search for MCECs to MECs of some $M[e]$, we do not discard any MCECs.
Furthermore, each recursive call splits the state space to disjoint sets, so we get an overall polynomial-time complexity.

Given an MEMDP $M$ over environments $E$, the notation $\pr_q^\sigma(M, \varphi)$ 
refers to the vector of probability values $(\pr_q^\sigma(M[e], \varphi))_{e \in E}$.

\subsection{Purge: Removing Non-Distinguishing MCECs}
\label{section:purge}
We first describe a transformation that collapses non-distinguishing MCECs, and keeps only trivial ones.
Since all trivial MCECs can be classified into winning and losing for the objective $\varphi$, we assume that
the only non-distinguishing MCECs in the resulting MEMDP are called $\winabsorb$ and $\loseabsorb$. 
The intuition is that non-distinguishing MCECs are not useful to refine information in order to distinguish environments,
so when a strategy visits such a MCEC, one can assume that it will either stay inside forever (either if the MCEC is $\varphi$-winning, or if there is no outgoing transition), or leave it as soon as possible (if the MCEC is $\varphi$-losing).

Observe that a distinguishing MCEC can contain a smaller non-distinguishing CEC. The transformation described here only collapses MCECs that are non-distinguishing, and not those smaller non-distinguishing CECs that are contained in MCECs.

Given an MEMDP~$M=\tuple{Q,A,(\delta_e)_{e\in E}}$, define
the MEMDP $\purge{M}=\tuple{Q',A',(\delta'_e)_{e\in E}}$ 
where $Q'$ contains all states of~$Q$ except those that belong to non-distinguishing MCECs;
and for each non-distinguishing MCEC $D$, we add a fresh state $s_D$ to $Q'$, and redirect all transitions that enter a state of $D$ in $M$ to $s_D$ in $M'$. 
We define the map $f : Q \rightarrow Q'$ by mapping all states of non-distinguishing MCECs $D$ to $s_D$, and as the identity for other states.

We add a fresh action $\actionstay$ which from $s_D$ goes to a winning absorbing state $\winabsorb$ if $D$ is $\varphi$-winning, and to a losing absorbing state $\loseabsorb$ otherwise.
For each pair $(q,a) \in D$ such that $\Supp(\delta(q,a))$ is not included in $D$, we add a fresh action $\freshaction{(q,a)}$ from $s_D$
with $\delta_e'(s_D, \freshaction{(q,a)})(q') = \sum_{q'' \in f^{-1}(q')} \delta_e(q, a, q'')$ for all $e \in E$.
(These state-action pairs can leave $D$ in some environments, so $F$ stands for the \emph{frontier} of~$D$.)

Given the set of MCECs, $\purge{M}$ can be computed in polynomial time.
However, the $\epsilon$-gap procedure we give does not actually compute $\purge{M}$;
this construction is only used for proving the existence of a finite-memory strategy of bounded memory size.

\begin{example}
  An example of this construction is given in Fig.~\ref{fig:removing-nondist-cec} for MEMDP $M$ with two environments $e_1,e_2$.
  Here $\{(q_2,b)\}$ is an end-component in $M[e_2]$ but not in $M[e_1]$ due to the edge to $q_3$ so this is not a CEC, and is not collapsed
  in $\purge{M}$.
  The MCEC $D$ defined by $\{(q_3,a), (q_4,a)\}$ has a single frontier action $F_{(q_4,b)}$.
  In $M[e_1]$, we have $\delta'_{e_1}(s_D, F_{(q_4,b)}, s_{D'})=2/3$ since 
  $\delta_{e_1}(q_4, b, q5)+\delta_{e_1}(q_4, b, q6)=2/3$ (since the probabilities are uniform),
  and $\delta'_{e_1}(s_D, F_{(q_4,b)}, s_{D})=1/3$.
  In $M[e_2]$, the latter edge is missing, so 
  $\delta'_{e_2}(s_D, F_{(q_4,b)}, s_{D'})=1$.
\end{example}

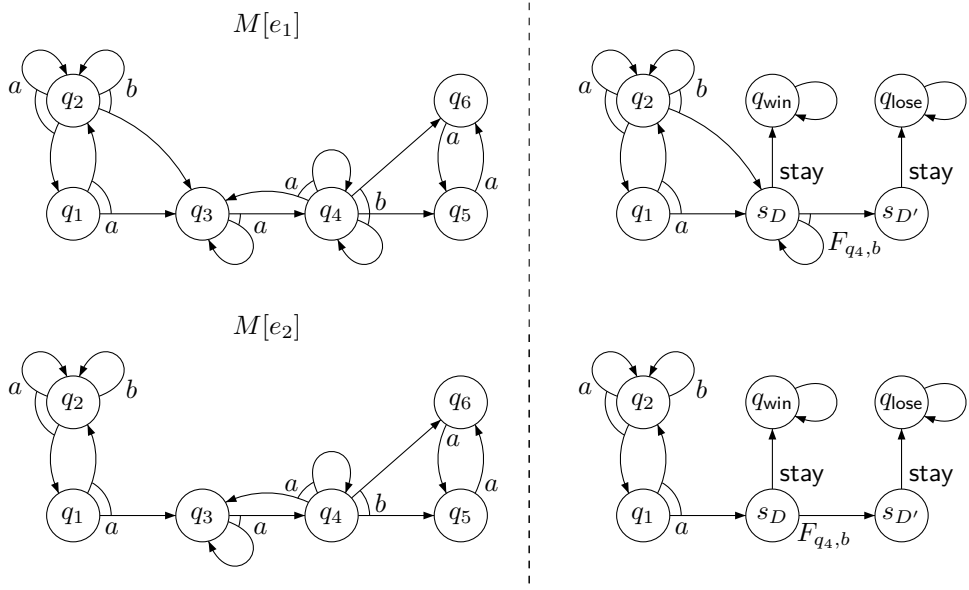
\begin{figure}[!t]
\hrule
  \centering
  \input{figures/gexample-purge}
\hrule
   \caption{
    An MEMDP~$M$ with two environments (left) and the construction $\purge{M}$ (right).
    Transition probabilities are uniform.
    Here $D$ is the MCEC defined by the pairs $\{(q_3,a), (q_4,a)\}$, and $D'$ is the MCEC defined by
    $\{(q_5,a), (q_6,a)\}$. The priority function is omitted, we assume that $D$ is winning (e.g., by assigning priority~$0$ to $q_3$ and $q_4$) and that $D'$ is losing (e.g., by assigning priority~$1$ to $q_5$ and $q_6$).  \label{fig:removing-nondist-cec}
  }
\end{figure}

\begin{lemma}
  \label{lemma:removing-nondist-cecs-1}
  For all MEMDPs $M$, the only non-distinguishing MCECs of $\purge{M}$ are the trivial $\winabsorb$ and $\loseabsorb$.
\end{lemma}

\begin{proof}
  Let $D=(Q',A')$ be any non-distinguishing MCEC in $M'$. 
  $D$ must contain a state of the form $s_{D'}$ since otherwise this is also a MCEC of $M$, and the construction would have collapsed it.
  We consider the component in $M$ given by the inverse image of $D$ by $f$.
  Formally, let $Q''=f^{-1}(Q') \subseteq Q$, and for each $q'' \in Q''$, define
  $A''(q'') = \{ a \in A(q'') \mid \forall e \in E: \Supp(\delta_e(q,a)) \subseteq Q''\}$.
  
  Then for each state of the form $q_{D'}$ in $D$, $(Q'',A'')$ contains all state-action pairs of $D'$.
  But $D$ is strongly connected in each $M'[e]$, and all non-distinguishing MCECs $D'$ of $M$ that were collapsed are
  also strongly connected in each $M[e]$ by definition, $(Q'',A'')$ is also strongly connected in each $M[e]$, thus
  a MCEC in $M$.

  Now, $(Q'',A'')$ cannot be distinguishing, since the construction only collapses MCECs, so no subset of $(Q'',A'')$
  can be collapsed in $M'$; and $(Q'',A'')$ would remain untouched and be distinguishing in $M'$ as well.
  So $(Q'',A'')$ is non-distinguishing; but in this case, it is collapsed into a trivial MCEC in $M'$, so $D$ is trivial.  
\end{proof}

To relate the histories of $M$ to those of $\purge{M}$, we introduce the function $h\mapsto \purge{h}$ which, 
intuitively, maps the state of a non-distinguishing MCECs $D$ to the state $s_D$,
removes the state-actions pairs that stay in~$D$, and replaces the state-action pairs $(q,a)$ having a transition that leaves $D$ by a new action $\freshaction{(q,a)}$.
Formally, $\purge{h}$ is obtained from $h=q_1a_1\ldots q_n$ by applying the following transformation:
for each non-distinguishing MCEC~$D=(Q',A')$ of $M$,
\begin{enumerate}
  \item Replace the maximal suffix of $h$ of the form $q_ia_i\ldots q_n$
  such that for all $i\leq k \leq n$, $q_k\in Q'$ and $a_k\in A'(q_k)$, 
  if such a suffix exists, by $s_D$;
  \item Remove all maximal factors of $h$ of the form $q_ia_i\ldots q_ja_j$
  satisfying $q_k \in Q'$ and $a_k\in A'(q_k)$ for all $i \leq k \leq j$;
  \item Replace each pair $q_ia_i$ with $q_i \in Q'$ and $a_i \not \in A'(q_i)$ by $s_D \freshaction{(q_i,a_i)}$;    
\end{enumerate} 

\begin{example}
  In the MEMDP of Fig.~\ref{fig:removing-nondist-cec},
  with $D$ containing the pairs $(q_3,a)$ and $(q_4,a)$, 
  for $h=q_1a q_3aq_4aq_3 a q_4 b q_4 b q_5$, we get 
  $\purge{h}=q_1a s_D F_{q_4,b} s_D F_{q_4,b} q_5$.
  Here we first apply rule~2 above to the factor $q_3aq_4aq_3 a$, and get
  $q_1 a q_4 b q_4 b q_5$; then an application of rule~3 yields
  $\purge{h}=q_1a s_D F_{q_4,b} s_D F_{q_4,b} q_5$.
  For the history $h'=q_1a q_3aq_4aq_3 a q_4$, we would get by rule~1,
  $\purge{h'}=q_1 a s_D$.
\end{example}

We establish a relation between $M$ and $\purge{M}$ in Lemmas~\ref{lemma:removing-nondist-cecs-stratmap} and \ref{lemma:removing-nondist-cecs-stratmap-2}.
These will be used to give a memory bound for strategies for the quantitative case in Lemma~\ref{lemma:quantitative-bounded-memory}.

In the Lemma~\ref{lemma:removing-nondist-cecs-stratmap}, we only establish an inequality. This is because a given strategy $\sigma$ of $M$
may not be optimal within a non-distinguishing MCEC, while the construction $\purge{M}$ is based on the assumption that optimal strategies are used within each MCECs.

\begin{lemma}
  \label{lemma:removing-nondist-cecs-stratmap}
  Consider an MEMDP $M=\tuple{Q,A,(\delta_e)_{e\in E}}$, and objective $\varphi=\Parity(p)$,
  and the map $f:Q\rightarrow Q'$ relating states of $M$ and those of $\purge{M} = \tuple{Q',A',(\delta_e')_{e\in E}}$.
  For all $q \in Q$, and strategy $\sigma$ for~$M$, there exists $\sigma'$ with
  $\pr_{q}^\sigma(M,\varphi) \leq \pr_{f(q)}^{\sigma'}(\purge{M},\varphi)$.
\end{lemma}
\begin{proof}
  Let us write $M'=\purge{M}$.
  Consider $q \in Q$, and a strategy $\sigma$ for $M$.
  We define $\sigma'$ for $M'$ as follows. For all histories $h$ of $M'$, and action $a \in A'(\last(h))$, we define
  \[
    \sigma'(h)(a) = \pr^\sigma_q\left[M[e], \purgeinv{ha} \mid \purgeinv{h}\right]
  \]
  for some arbitrary $e \in E$ for which $\pr^\sigma_q\left[M[e], \purgeinv{h}\right]>0$,
  if such $e \in E$ exists; and otherwise define $\sigma'(h)$ arbitrarily.
  This quantity does not depend on $e$ since, assuming $\pr^\sigma_q\left[M[e], \purgeinv{h}\right]>0$,
  \begin{align*}
    &   \pr^\sigma_q\left[M[e], \purgeinv{ha}\mid \purgeinv{h}\right]\\
    & = \sum_{\rho \in \purgeinv{h}} \pr^\sigma_q\left[M[e], \rho a' \mid \rho\right] \pr^\sigma_q\left[M[e], \rho \mid \purgeinv{h}\right] \\
    & = \sum_{\rho \in \purgeinv{h}} \sigma(\rho)(a') \pr^\sigma_q\left[M[e], \rho \mid \purgeinv{h}\right],
  \end{align*}
  where $a' = b$ if $a$ has the form $F_{(\_,b)}$,
  and $a'=a$ otherwise (in which case we have $a \in A(\last(\rho))$).
  Moreover, $\pr_q^\sigma[M[e], \rho \mid \purgeinv{h}]$ does not depend on $e$ here since 
  $\purgeinv{h}$ determines the outcomes of all transitions whose probability distributions differ among environments
  because these were not erased by $\purge{\cdot}$,
  and these probability distributions are identical in the remaining transitions since they belong to non-distinguishing MCECs.

  For a history $h$ of $M'$ that ends in a state of the form $s_D$, 
  and with $\pr^\sigma_q\left[M[e], \purgeinv{h}\right]>0$,
  we let $\sigma'$ take the action $\actionstay$
  with probability $\pr_q^\sigma[M[e], \purgeinv{h} D^\omega \mid \purgeinv{h}]$,
  where $D^\omega$ denotes the set of all runs that stay inside $D$.
  This probability is similarly independent from the particular choice of $e$.

  We prove that for all histories $h$ of $M'$ that do not contain $\actionstay$, $a \in A'(\last(h))\setminus\{\actionstay\}$, and $e \in E$,
  \begin{align}
    &\pr^{\sigma'}_{f(q)}\left[M'[e], h \right] = \pr^\sigma_q\left[M[e], \purgeinv{h} \right],\label{eqn:nondist-cec-eqn1}\\
    &\pr^{\sigma'}_{f(q)}\left[M'[e], h \cdot a \right] = \pr^\sigma_q\left[M[e], \purgeinv{h\cdot a} \right]\label{eqn:nondist-cec-eqn2},\\
    &\pr^{\sigma'}_{f(q)}\left[M'[e], h \cdot \actionstay \right] = \pr^\sigma_q\left[M[e], \purgeinv{h}D^\omega \right].\label{eqn:nondist-cec-eqn3}
  \end{align}
  We proceed by induction on the length of $h$ to prove the above three properties.
  
  Initially, if $\abs{h}=1$, then $h=f(q)$ and $\purgeinv{h} = \{q\}$.
  Then \eqref{eqn:nondist-cec-eqn1} follows trivially since both sides are equal to 1.
  To see \eqref{eqn:nondist-cec-eqn2}, note that, by definition of $\sigma'$,
  \[\sigma'(h)(a) = \pr^\sigma_q\left[M[e], \purgeinv{h\cdot a} \mid \purgeinv{h}\right] = \pr^\sigma_q\left[M[e], \purgeinv{h\cdot a} \right] \]
  since $h=f(q)$ here.
  Furthermore, 
  \begin{align*}
    \pr^{\sigma'}_{f(q)}\left[M'[e], h \cdot a \right]
    &=\pr^{\sigma'}_{f(q)}\left[M'[e], a \mid h \right] \pr^{\sigma'}_{f(q)}\left[M'[e], h\right]\\
    &=\sigma'(h)(a)
  \end{align*}
  since $\pr^{\sigma'}_{f(q)}\left[M'[e], h\right]=1$; which yields \eqref{eqn:nondist-cec-eqn2}.

  Last, assume that $\actionstay \in A'(f(q))$, that is $f(q)$ has the form $s_D$ for some non-distinguishing MCEC $D$.
  \begin{align*}
    \pr^{\sigma'}_{f(q)}\left[M'[e], h \cdot \actionstay \right] &= \sigma'(h)(\actionstay) \pr^{\sigma'}_{f(q)}\left[M'[e], h \right]\\
    & = \pr_q^\sigma[M[e], \purgeinv{h} D^\omega \mid \purgeinv{h}]\\
    & = \pr_q^\sigma[M[e], \purgeinv{h} D^\omega ],
  \end{align*}
  which proves \eqref{eqn:nondist-cec-eqn3}.

  Assume now that $\abs{h}>1$, and let us write $h=h'ar$ for a history $h'$, $a\in A'(\last(h'))$.
  \begin{align*}
    \pr^{\sigma'}_{f(q)}\left[M'[e], h'a \right] &= \pr^{\sigma'}_{f(q)}\left[M'[e], h'a \mid h'\right]\pr^{\sigma'}_{f(q)}\left[M'[e], h' \right]\\
    &=\pr^{\sigma'}_{f(q)}\left[M'[e], h'a \mid h'\right]\pr^{\sigma}_{q}\left[M'[e], \purgeinv{h'} \right] \\
    &=\sigma'(h')(a) \pr^{\sigma}_{q}\left[M'[e], \purgeinv{h'} \right] \\
    &=\pr^\sigma_q\left[M[e], \purgeinv{h'a} \mid \purgeinv{h'}\right] \\
    & \qquad\qquad \cdot \pr^{\sigma}_{q}\left[M'[e], \purgeinv{h'} \right],\\
    &=\pr^\sigma_q\left[M[e], \purgeinv{h'a} \right],
  \end{align*}
  where we used the induction hypothesis to apply \eqref{eqn:nondist-cec-eqn1} on the second line.
  This proves \eqref{eqn:nondist-cec-eqn2}.

  Consider now $r \in Q'$.
  \begin{align*}
    \pr^{\sigma'}_{f(q)}\left[M'[e], h'ar \right] &= \pr^{\sigma'}_{f(q)}\left[M'[e], h'ar \mid h'a\right]\pr^{\sigma'}_{f(q)}\left[M'[e], h'a \right]\\
    &=\delta_e'(\last(h'),a)(r)\pr^\sigma_q\left[M[e], \purgeinv{h'a} \right].
  \end{align*}
  We distinguish two cases. If $\last(h)$ does not have the form of $s_D$, then it also belongs to $Q$, $a \in A(\last(q))$,
  with $\delta_e(\last(h'),a)(r) = \delta_e'(\last(h'),a)(r)$.
  In this case, the above is equal to $\pr^\sigma_q\left[M[e], \purgeinv{h'ar} \right]$.
  Assume now that $\last(h) = s_D$ for some non-distinguishing MCEC $D$, and that 
  $a=\freshaction{(r', a')}$ for some pair $(r',a')$.
  Then $\purgeinv{h'a}$ only contains histories that end at $r'$, followed by action $a'$.
  We have, moreover, $\delta_e'(\last(h'),a)(r) = \sum_{q\in f^{-1}(r)} \delta_e(r',a')(q)$,
  by the definition of $\purge{M}$, so
  
  \begin{align*}
    \pr^{\sigma'}_{f(q)}\left[M'[e], h'ar \right] &=    \pr^\sigma_q\left[M[e], \purgeinv{h'a}f^{-1}(r) \right] \\
    &= \pr^\sigma_q\left[M[e], \purgeinv{h'ar} \right].
  \end{align*}
  This proves \eqref{eqn:nondist-cec-eqn1}.

  Last, consider history $h$ ending at some state $s_D$, and write 
  \begin{align*}
    \pr^{\sigma'}_{f(q)}\left[M'[e], h \cdot \actionstay \right] &= \pr^{\sigma'}_{f(q)}\left[M[e], h\cdot \actionstay \mid h\right]
          \pr^{\sigma'}_{f(q)}\left[M[e], h\right]\\
          &= \pr^{\sigma'}_{f(q)}\left[M[e], h\cdot \actionstay \mid h\right]
          \pr^\sigma_q\left[M[e], \purgeinv{h}\right]\\
          &= \pr_q^\sigma[M[e], \purgeinv{h} D^\omega \mid \purgeinv{h}]
          \pr^\sigma_q\left[M[e], \purgeinv{h}\right]\\
          &= \pr_q^\sigma[M[e], \purgeinv{h} D^\omega],
  \end{align*}
  where we used the induction hypothesis on the second line, and the definition of $\sigma'$ on  the third line.
  This proves \eqref{eqn:nondist-cec-eqn3}.

  \medskip
  We now show that $\pr_{q}^\sigma(M,\varphi) \leq \pr_{f(q)}^{\sigma'}(M',\varphi)$ follows from these properties.
  In fact, for all $e \in E$, one can write
  \begin{equation}
    \label{eqn:m-mp-final}
    \begin{array}{ll}
    \pr_{f(q)}^{\sigma'}(M'[e],\varphi) = &
    \sum_{\substack{D \in \ecs(M[e]), \text{ $\varphi$-winning }\\\text{ non-distinguishing MCEC}}} \pr_{f(q)}^{\sigma'}(M'[e],(Q'A')^*\cdot s_D \cdot \actionstay \cdot \winabsorb)
    \\
    &+\sum_{\substack{D \in \ecs(M'[e]), \text{ $\varphi$-winning } \\ D \neq \{(\winabsorb,\_)\} }}\pr_{f(q)}^{\sigma'}(M'[e],\Inf=D),
    \end{array}
  \end{equation}
  by separating winning end-components of $M'$ into two: the winning absorbing state $\winabsorb$ reached via some $s_D$ for a non-distinguishing MCEC of $M$, 
  and any other end-component of $M'$.

  For the first term of \eqref{eqn:m-mp-final}, we have, from the above properties of $\sigma'$, 
  \begin{align*}
    & \sum_{\substack{D \in \ecs(M[e]), \text{ $\varphi$-winning }\\\text{ non-distinguishing MCEC}}} \pr_{f(q)}^{\sigma'}(M'[e],(Q'A')^*\cdot s_D \cdot \actionstay\cdot \winabsorb)\\
    &= \sum_{\substack{D \in \ecs(M[e]), \text{ $\varphi$-winning }\\\text{ non-distinguishing MCEC}}}\sum_{\substack{D' \in \ecs(M[e])\\D' \subseteq D}} \pr_q^\sigma[M[e], \Inf=D']\\
    &\geq \sum_{\substack{D \in \ecs(M[e]), \text{ $\varphi$-winning }\\\text{ non-distinguishing MCEC}}}\sum_{\substack{D' \in \ecs(M[e]), \text{ $\varphi$-winning} \\D' \subseteq D}} \pr_q^\sigma[M[e], \Inf=D']
    \\&=\sum_{\substack{D \in \ecs(M[e]), \text{ $\varphi$-winning }\\\text{ non-distinguishing MCEC}}} \pr_q^\sigma[M[e], \Inf=D].
  \end{align*}

  For the second term of \eqref{eqn:m-mp-final}, let $D\in \ecs(M')\setminus \{(\winabsorb,\_)\}$, and observe that $\Inf=D$ does not contain the action $\actionstay$.
  Notice how we only have an inequality because $\sigma$ might actually have a nonzero probability of realizing $\Inf=D'$ for some non-winning $D'$ included in a winning $D$.

  We established above that for all histories $h$ of $M'$ without the action $\actionstay$,
  $\pr_{f(q)}^{\sigma'}[M'[e], h] = \pr_{q}^{\sigma}[M[e], \purgeinv{h}]$,
  that is, cylinders generated by $h$ and $\purgeinv{h}$ have the same probabilities in $M'$ under
  $\sigma'$, and, respectively, in $M$ under $\sigma$. 
  It follows that
  \begin{align*}
    &~\sum_{\substack{D \in \ecs(M'[e]), \text{ $\varphi$-winning } \\ D \neq \{(\winabsorb,\_)\} }}\pr_{f(q)}^{\sigma'}[M'[e], \Inf=D]\\
    &=\sum_{\substack{D \in \ecs(M'[e]), \text{ $\varphi$-winning } \\ D \neq \{(\winabsorb,\_)\} }}\pr_{q}^{\sigma}[M[e], \purgeinv{\Inf = D}]\\
    &= \sum_{\substack{D \in \ecs(M'[e]), \text{ $\varphi$-winning } \\ D \neq \{(\winabsorb,\_)\} }}\sum_{D' \in \ecs(M), f(D') = D} \pr_{q}^{\sigma}[M[e], \Inf = D']\\
    &\geq \sum_{\substack{D \in \ecs(M'[e]), \text{ $\varphi$-winning } \\ D \neq \{(\winabsorb,\_)\} }}\sum_{\substack{D' \in \ecs(M), f(D') = D\\ D' \text{ is $\varphi$-winning}}} \pr_{q}^{\sigma}[M[e], \Inf = D'],\\
    &=\sum_{\substack{D \in \ecs(M[e]), \text{ $\varphi$-winning} \\ \text{not a non-distinguishing MCEC}}} \pr_{q}^{\sigma}[M[e], \Inf = D].
  \end{align*}
  where we extend the definition of $f$ to state-action pairs so that $f(D')$ denotes an end-component of $M'$.

  Combining these bounds on both terms of \eqref{eqn:m-mp-final}, we conclude
  \begin{align*}
    \pr_{f(q)}^{\sigma'}(M'[e],\varphi) &\geq
    \sum_{\substack{D \in \ecs(M[e]), \text{ $\varphi$-winning }\\\text{ non-distinguishing MCEC}}} \pr_q^\sigma[M[e], \Inf=D] 
    \\ & ~~+  \sum_{\substack{D \in \ecs(M[e]), \text{ $\varphi$-winning} \\ \text{not a non-distinguishing MCEC}}} \pr_{q}^{\sigma}[M[e], \Inf = D]\\
    &\geq \pr_{q}^{\sigma}(M[e],\varphi).
  \end{align*}
\end{proof}

The following lemma is the dual, and shows that any strategy for $\purge{M}$ can be replicated in $M$,
albeit with a bit more memory. The additional memory is required to implement behaviors inside non-distinguishing MCECs.

\begin{lemma}
  \label{lemma:removing-nondist-cecs-stratmap-2}
    Consider an MEMDP $M=\tuple{Q,A,(\delta_e)_{e\in E}}$, and objective $\varphi=\Parity(p)$,
    and the map $f:Q\rightarrow Q'$ relating states of $M$ and that of $\purge{M} = \tuple{Q',A',(\delta_e')_{e\in E}}$.
    For all states $q'\in Q'$ and strategies $\sigma'$ for~$\purge{M}$, and all $q \in f^{-1}(q')$, 
    there exists a strategy $\sigma$ with
    $\pr_{q}^\sigma(M,\varphi) = \pr_{q'}^{\sigma'}(\purge{M},\varphi)$.
    Furthermore, if $\sigma'$ is a $m$-memory strategy, $\sigma$ can be chosen to be a
    $(m+\abs{Q}\abs{\Act})$-memory strategy.
\end{lemma}
\begin{proof}
  Consider $q'\in Q'$, an $m$-memory strategy $\sigma'$ for~$M'$, and $q \in f^{-1}(q')$,
  where $m$ can be finite or infinite.
  We show that there exists an $(m+\abs{Q}\abs{\Act})$-memory strategy $\sigma$ with $\pr_{q}^\sigma(M,\varphi) = \pr_{q'}^{\sigma'}(M',\varphi)$.
  We define $\sigma$ as follows. Consider a history $h$ of $M$.
  \begin{itemize}
    \item If $f(\last(h)) \in Q$, we let $\sigma(h) = \sigma'(\purge{h})$.
    \item Assume $f(\last(h)) = s_D$ for some non-distinguishing MCEC~$D$. 
    With probability $\sigma'(\purge{h})(\actionstay)$, we let $\sigma$ switch to a pure memoryless strategy that maximizes the probability 
    of $\varphi$ inside $D$ (this strategy is independent from the environment).    
    For each $\freshaction{(q,a)}$,
    with probability $\sigma'(\purge{h})(\freshaction{(q,a)})$ we let $\sigma$ run a pure memoryless strategy until state $q$ is reached (which happens probability 1), 
    and from $q$ take~$a$.
  \end{itemize}
  The memory bound for $\sigma$ is $m+\abs{Q}\abs{\Act}$ where $m$ is the memory size of $\sigma'$, because inside each collapsed MCEC,
  and for each pair $\freshaction{(q,a)}$, a pure memoryless strategy is executed until reaching $q$ and taking action~$a$.

  By construction, for all histories~$h$ that start at~$q$ in~$M$ and end outside of non-distinguishing end-components, 
  we have:
  \begin{equation}
    \pr_{q}^{\sigma}(M[e], h) = \pr_{f(q)}^{\sigma'}(M'[e], \purge{h}) \text{ for all environments } e \in E.
  \end{equation}
  So if $R$ denotes a measurable set of infinite runs of $M$ such that for all $\rho \in R$, $\purge{\rho}$ is infinite 
  (in other terms, $\rho$ does not stay inside a non-distinguishing MCEC), then 
  \begin{equation}
    \label{eqn:same-prob-R-M-M'}
    \pr_q^\sigma(M[e], R) = \pr_{f(q)}^{\sigma'}(M'[e], \purge{R}),
  \end{equation}
  writing $\purge{R} = \{\purge{\rho} \mid \rho \in R\}$.

  Furthermore, for those histories $h=h'as$ where $q \in D$ is a non-distinguishing MCEC and $\last(h') \not \in D$, we have:
  \[
    \pr_{q}^{\sigma}(M[e], h) = \pr_{f(q)}^{\sigma'}(M'[e], \purge{h}) \text{ for all environments } e \in E.
  \]
  Then, by definition of $\sigma$, for a non-distinguishing MCEC $D$,
  \begin{equation}
    \label{eqn:same-prob-ndmcec}
    \sum_{D' \in \ecs(M), D'\subseteq D}\pr_q^\sigma(M[e], \Inf=D') = \pr_{f(q)}^{\sigma'}(M'[e], (Q'A')^*\cdot s_D \cdot \actionstay).
  \end{equation}

  Observe that any end-component $D$ of $M$ that is not a non-distinguishing MCEC maps to an end-component of $M'$.
  We have for all $e\in E$,
  \begin{align*}
    \pr_q^\sigma(M[e], \varphi) &= \sum_{\substack{D \in \ecs(M), \text{$\varphi$-winning} \\\text{not a non-distinguishing MCEC}}}
    \pr_{q}^{\sigma}(M[e], \Inf=D)
    \\& + \sum_{\substack{D \in \ecs(M), \text{$\varphi$-winning}\\\text{non-distinguishing MCEC}}}
                                    \pr_{q}^{\sigma}(M[e], \Inf=D)
    \\&=\sum_{\substack{D \in \ecs(M'), \text{$\varphi$-winning}\\\text{not a non-distinguishing MCEC}}}
    \pr_{f(q)}^{\sigma'}(M'[e], \Inf=D)
    \\&+\sum_{\substack{D \in \ecs(M'), \text{$\varphi$-winning}\\\text{non-distinguishing MCEC}}}
    \pr_{f(q)}^{\sigma'}(M'[e], (Q'A')^*\cdot s_D \cdot \actionstay )\\
    &=\pr_{f(q)}^{\sigma'}(M'[e], \varphi),
  \end{align*}
  using \eqref{eqn:same-prob-R-M-M'} and \eqref{eqn:same-prob-ndmcec}.
\end{proof}

\subsection{Learning While Playing}
\label{section:learning-while-playing}
In this section, we show that after collapsing non-distinguishing MCECs, over $n$ steps (for $n$ large enough),
with high probability,
we either reach a 
MCEC (which is either distinguishing or trivial) 
or collect a large number of samples of distinguishing transitions whose empirical average is close to their mean.
Intuitively, this means that either the knowledge can be improved after $n$ steps using the collected samples while bounding the probability of error, or a 
MCEC is reached. 

If the MCEC is distinguishing, the strategy can improve the knowledge as in Lemma~\ref{lemma:strategy-for-dist-mecs},
and if not, then the MCEC is trivial and there is a unique way to play.
These results will be used in the next section to build a finite-memory strategy with approximately the same probability of winning,
given any arbitrary strategy.

For a history $h$, let $\abs{h}_{q,a}$ denote the number of occurrences of the state-action pair $(q,a)$,
and $\abs{h}_{q,a,q'}$ the number of times these are followed by $q'$, where $q' \in\Supp(\delta(q,a))$.
For a distinguishing transition $t=(q,a,q')$, we say that a history $h$ is a
\emph{bad $(t,\eta)$-classification} in MDP $M[e]$ if 
$\left\lvert\frac{\abs{h}_{q,a,q'}}{\abs{h}_{q,a}} - \delta_{e}(q,a)(q')\right\rvert \geq \eta/2$, that is the measured and theoretical frequency of $t$ are too far apart.
It is a \emph{good $(t,\eta)$-classification} otherwise.
Intuitively, over long histories, good classifications have high probability.

We first prove the following technical lemma, bounding the difference between the empirical average and the mean when sampling among a finite number of
transitions, when the transitions to sample are chosen at each step by an \emph{adversary}. This adversary corresponds to strategies in an MDP, is arbitrary,
and can depend on the history and use randomization.

We state the following lemma for (single-environment) MDPs, and apply it to each environment in an MEMDP.
\begin{lemma}
  \label{lemma:adversarial-sampling}
  Consider MDP~$M$, state~$q_0$, and $T=\{t_i = (q_i,a_i,q_i')\}_{1\leq i \leq k}$ a subset of transitions
  such that $(q_i,a_i)=(q_j,a_j)$ implies $q_i'=q_j'$ for all $i,j$.
  For all $\eta,\epsilon>0$, all $n_0>\frac{k^3}{\epsilon\eta^2}$, and any strategy $\sigma$ with
  $\pr_{q_0}^\sigma\left[\{h : \sum_{(q,a,q') \in T} \abs{h_{q,a}} \geq n_0\}\right]=1$, the set of histories~$h$ that satisfy the following conditions
  has probability at most $\epsilon$:
  \begin{itemize}
    \item $\sum_{(q,a,q') \in T} \abs{h_{q,a}} \geq n_0$
    \item there exists $1{\leq} i {\leq} k$ such that $\abs{h}_{q_i,a_i} {=} \max_{i'} \abs{h}_{q_{i'},a_{i'}}$
    and $h$ is a bad $(t_i,\eta)$-classification.
  \end{itemize}
\end{lemma}

Here the assumption on $T$ simplifies the proofs since it means that for each state-action pair, 
we will be observing the frequency of a unique successor state.
The lemma also requires that at least $n_0$ occurrences of $T$ is visited with probability 1. This hypothesis ensures that we have enough samples
to obtain a good approximation (that is, a good $(t_i,\eta)$-classification) with high probability (at least $1-\epsilon$).
In fact, if a strategy $\sigma$ avoids visiting transitions from $T$, say, with probability $1/2$, then it cannot ensure a good approximation with high probability because half the cases, there are just not enough samples of $T$.

The lemma is easy for $k=1$. In fact, all trials are identical and independent, so one can use e.g. Hoeffding's inequality to derive a bound.
When $k>1$, trials are no longer independent since $\sigma$ might react to the success or failure of a given transition to make its decisions in the future.
In fact, the lemma is not trivial to prove due to the possible dependency between the trials.

Here is such a situation of dependency. 
Consider a state $q$ from which action $a$ leads to either to $q_1$ or $q_2$, each with probability 0.5, from which a deterministic transition comes back to $q$.
Another action $b$ from $q$ deterministically loops back at $q$.
Consider $\sigma$ that picks $(q,a)$ first. As long as we reach $q_1$, $\sigma$ continues to pick $(q,a)$. Whenever $q_2$ is reached, $\sigma$ switches definitively to $(q,b)$.
Now the probability of observing $(q,a,q_1)$ at step $n>1$ depends on the result of the first $n-1$ trials. 
For example, conditioned on 
observing $(q,a,q_1)$ on the first $n-1$ trials, the probability of observing $(q,a,q_1)$ again is $0.5$.
But conditioned on not observing $(q,a,q_1)$ on the $(n-1)$-th trial, this probability is $0$. This shows that given such $\sigma$, the successive trials are not independent,
and theorems such as Hoeffding's inequality cannot be applied.

In turns out that although the trials can be dependent, their covariance is 0. We exploit this observation to derive a good bound using
Chebyshev's inequality:
\begin{theorem}[Chebyshev's Inequality]
  \label{thm:chebyshev}
  Let $X$ be a random variable with mean $\mu$, and standard deviation $q$. Then, for all $a>0$, we have
  \(
    \pr[ \abs{X-\mu} \geq s a ] \leq \frac{1}{a^2}.
  \)
\end{theorem}
This inequality clearly also applies if $q$ is an upper bound on the standard deviation of $X$.

\begin{proof}[Proof of Lemma~\ref{lemma:adversarial-sampling}]
  We consider a slightly more abstract setting where there are $k$ independent arms,
  each with a probability of success of $p_i$.
  In MDPs, each arm corresponds to a state-action pair $(q_i,a_i)$ and it  succeeds when reaching $q_{i}'$, with probability $p_i=\delta(q_i,a_i)(q_i')$.

  Consider a strategy $\sigma$ that chooses, at each step, $i \in \{1,\ldots,k\}$, an arm to pull based on the full history and randomization.
  Consider $\epsilon,\eta>0$.

  We model the problem as follows.
  For each $i \in \{1,\ldots,k\}$, define a sequence $X_1^{(i)}, X_2^{(i)}, \ldots$ of identical and independent Bernoulli variables with probability $p_i$.
  Let $\Choice_j$ denote the arm selected by $\sigma$ at step $j$.
  At each step $j$, $\Choice_j$ selects an arm, and
  all types of arms are pulled.
  While $\Choice_j$ can depend on the history, $X_j^{(i)}$ does not depend on the history, and in particular on $\Choice_j$.

  Define the \emph{weight} of arm $i$ at step $j$ as the following random variable.
  \[
  \wgt_j^{(i)} = \left\{\begin{array}{ll}
    X_j^{(i)} - p_i & \text{if } \Choice_j=i,\\
    0 & \text{otherwise.}
  \end{array}\right.
  \]
  Define $\wgt_{\leq n}^{(i)} = \sum_{j=1}^n \wgt_j^{(i)}$, for any $n\geq 1$.
  Let us also define $\occ_j^{(i)} =1 $ iff $\Choice_j=i$, and $\occ_{\leq n}^{(i)}= \sum_{j=1}^n \occ_j^{(i)}$.
  Observe that 
  \[\wgt_{\leq n}^{(i)} = \sum_{1\leq j \leq n, \Choice_j=i} X_j^{(i)} - \occ_{\leq n}^{(i)} p_i,\]
  that is, this is the difference between the empirical sum and the mean of the sum of the subsequence of $X_j^{(i)}$ where $\Choice_j=i$.

  Then $\frac{\wgt_{\leq n}^{(i)}}{\occ_{\leq n}^{(i)}}$ is the difference between the empirical average of the $X^{(i)}_j$ and $p_i$,
  assuming that $\occ_{\leq n}^{(i)}>0$.

  We have, by the definition of variance,
  \[
    \expect^\sigma[\wgt_j^{(i)}] = \pr^\sigma[\Choice_j=i](p_i (1 - p_i) + (1-p_i)(-p_i)) = 0,
  \]
  so $\expect^\sigma[\wgt_{\leq n}^{(i)}] = 0$ as well.

  We are going to apply Theorem~\ref{thm:chebyshev} on the variable $\wgt_{\leq n}^{(i)}$; so we need a bound on the variance of $\wgt_{\leq n}^{(i)}$.
  We show that $\variance^\sigma[\wgt_{\leq n}^{(i)}]\leq np_i(1-p_i)$.
  We have 
  \[
    \variance^\sigma[\wgt_{\leq n}^{(i)}] = \sum_{j=1}^n \variance^\sigma[\wgt_j^{(i)}] + 2\sum_{1\leq j<j'\leq n} \cov(\wgt_j^{(i)},\wgt_{j'}^{(i)})
  \]
  For each $j$, because $\expect^\sigma[\wgt_j^{(i)}] = 0$, we have $\variance^\sigma[\wgt_j^{(i)}]=\expect^\sigma[(\wgt_j^{(i)})^2]$,
  which can be calculated as 
  \begin{align*}
    &\pr^\sigma[\Choice_j=i](p_i (1-p_i)^2 + (1-p_i)(-p_i)^2)\\
    &=\pr^\sigma[\Choice_j=i]p_i(1-p_i)((1-p_i) + p_i)\\
    &\leq p_i(1-p_i),
  \end{align*}
  so that the first term of the variance is at most $np_i(1-p_i)$.

  Now, as noted above, $\wgt_j^{(i)}$ and $\wgt_{j'}^{(i)}$ are not independent variables since $\sigma$ can choose the arm at step $j'$ depending on the result of 
  $\wgt_j^{(i)}$; we nevertheless show that the covariance is equal to $0$.
  We have $\cov(\wgt_j^{(i)},\wgt_{j'}^{(i)}) = \expect^\sigma[\wgt_j^{(i)} \cdot \wgt_{j'}^{(i)}] - \expect^\sigma[\wgt_j^{(i)}][\wgt_{j'}^{(i)}]$ by definition of covariance;
  so this is equal to $\expect^\sigma[\wgt_j^{(i)} \cdot \wgt_{j'}^{(i)}]$ which can be calculated as follows.
  \begin{align*}
    & \pr^\sigma[ \Choice_j=i \land \Choice_{j'}=i \land X_j^{(i)}=1 \land X_{j'}^{(i)}=1](1-p_i)^2\\
    & +\pr^\sigma[ \Choice_j=i \land \Choice_{j'}=i \land X_j^{(i)}=1 \land X_{j'}^{(i)}=0](1-p_i)(-p_i)\\
    & +\pr^\sigma[ \Choice_j=i \land \Choice_{j'}=i \land X_j^{(i)}=0 \land X_{j'}^{(i)}=1](-p_i)(1-p_i)\\    
    & +\pr^\sigma[ \Choice_j=i \land \Choice_{j'}=i \land X_j^{(i)}=0 \land X_{j'}^{(i)}=0](-p_i)^2.
  \end{align*}

  Now $X_{j'}^{(i)}$ and the variables $\Choice_j, \Choice_{j'}, X_{j}^{(i)}$ are independent; in fact, the values of $\Choice_j, \Choice_{j'}$  cannot depend on
  $X_{j'}^{(i)}$ since the latter is revealed after $\Choice_j, \Choice_{j'}$.
  In contrast,
  $X_{j}^{(i)}$ and $\Choice_{j'}$ can be dependent since the latter can depend on the value of $X_{j}^{(i)}$.
  
  We can rewrite $\pr^\sigma[ \Choice_j=i \land \Choice_{j'}=i \land X_j^{(i)}=1 \land X_{j'}^{(i)}=1]$ as follows.
   \begin{align*}
     & \pr^\sigma[ X_{j'}^{(i)} = 1 \mid \Choice_j=i \land \Choice_{j'}=i \land X_j^{(i)}=1  ] \pr^\sigma[\Choice_j=i \land \Choice_{j'}=i \land X_j^{(i)}=1]\\
     &=\pr^\sigma[ X_{j'}^{(i)} = 1 ] \pr^\sigma[\Choice_j=i \land \Choice_{j'}=i \land X_j^{(i)}=1]\\
     &= p_i \pr^\sigma[\Choice_j=i \land \Choice_{j'}=i \land X_j^{(i)}=1]
  \end{align*}
  by independence.
  
  Applying this to all four terms, $\expect^\sigma[\wgt_j^{(i)} \cdot \wgt_{j'}^{(i)}]$ can be written as
  \begin{align*}
    &\pr^\sigma[\Choice_j=i \land \Choice_{j'}=i \land X_j^{(i)}=1]p_i (1-p_i)^2\\
    &+\pr^\sigma[\Choice_j=i \land \Choice_{j'}=i \land X_j^{(i)}=1](1-p_i) (1-p_i)(-p_i)\\
    &+\pr^\sigma[\Choice_j=i \land \Choice_{j'}=i \land X_j^{(i)}=0]p_i (-p_i)(1-p_i)\\
    &+\pr^\sigma[\Choice_j=i \land \Choice_{j'}=i \land X_j^{(i)}=0](1-p_i) (-p_i)^2,\\
    =&\pr^\sigma[\Choice_j=i \land \Choice_{j'}=i \land X_j^{(i)}=1](p_i (1-p_i)^2 + (1-p_i) (1-p_i)(-p_i))\\
    &+\pr^\sigma[\Choice_j=i \land \Choice_{j'}=i \land X_j^{(i)}=0](p_i (-p_i)(1-p_i) + (1-p_i) (-p_i)^2)\\
    =&0.
  \end{align*}

  So all covariance terms are $0$, and we have $\variance^\sigma[\wgt_{\leq n}^{(i)}]\leq np_i(1-p_i)$.

  We now apply Theorem~\ref{thm:chebyshev}: For all $1\leq i \leq k$, and for all $a>0$,
  \[
    \pr^\sigma\left[ \left\lvert\wgt^{(i)}_{\leq n}\right\rvert \geq a\sqrt{np_i(1-p_i)} \right] \leq \frac{1}{a^2}.
  \]

  Using $\pr[X \cup Y] = \pr[X] + \pr[Y] - \pr[X \cdot Y]$, it follows that   
  \[
    \pr^\sigma\left[ \exists i, \left\lvert{\wgt^{(i)}_{\leq n}}\right\rvert \geq {a\sqrt{np_i(1-p_i)}} \right] \leq \frac{k}{a^2}.
  \]
  
  We have,
  \[
    \pr^\sigma\left[ 
        \exists i, 
        \occ_{\leq n}^{(i)} = \max_{i'}\occ_{\leq n}^{(i')}
        \land \left\lvert\frac{\wgt^{(i)}_{\leq n}}{\occ^{(i)}_{\leq n}}\right\rvert 
              \geq \frac{a\sqrt{np_i(1-p_i)}}{{\occ^{(i)}_{\leq n}}} \right] 
        \leq \frac{k}{a^2}.
  \]
  
  Here we divided the inequality by $\occ^{(i)}_{\leq n}$ (since for $n>0$, $\max_{i'}\occ_{\leq n}^{(i')}>0$);
  moreover, the probability bound holds since each event has become smaller.
  
  Notice that for all $n>0$, $\sum_{i=1}^k \occ_{\leq n}^{(i)}=n$ since $\sigma$ picks one of the arms at each step;
  so $\max_{i'}\occ_{\leq n}^{(i')} \geq n/k$ with probability 1. 
  We get
  \[
    \pr^\sigma\left[ 
      \exists i, 
      \occ_{\leq n}^{(i)} = \max_{i'}\occ_{\leq n}^{(i')}
      \land \left\lvert\frac{\wgt^{(i)}_{\leq n}}{\occ^{(i)}_{\leq n}}\right\rvert 
              \geq \frac{ak\sqrt{p_i(1-p_i)}}{\sqrt{n}} \right] 
        \leq \frac{k}{a^2}.
  \]

  Now, given $\epsilon,\eta>0$, we pick $a=\sqrt{k/\epsilon}$ so that $k/a^2\leq \epsilon$; and then $n$ large enough so that 
$\frac{ak\sqrt{p_i(1-p_i)}}{\sqrt{n}}\leq \eta/2$; this means it suffices to pick $n$ such that
\(
  \max_{1\leq i \leq k}\left(\frac{ak \sqrt{p_i(1-p_i)}}{\eta/2}\right)^2 \leq n
\), so $(\frac{ak}{\eta/2})^2 = \frac{4k^3}{\epsilon\eta^2} \leq n$ suffices. 
\end{proof}

We use Lemma~\ref{lemma:adversarial-sampling} to prove that in MEMDPs without non-trivial and non-distinguishing MCECs (for example, obtained by $\purge{\cdot}$),
after $n$ steps, we either reach a MCEC or collect a large number of samples of distinguishing transitions whose empirical average is close to their mean.
Given MEMDP $M$, let $T_M$ denote a set obtained by selecting one distinguishing transition $(q,a,q')$ for each state-action pair $(q,a)$ whose probability 
distribution differs in a pair of different environments.
We select one representative distinguishing transition $(q,a,q')$ for each pair $(q,a)$ because this simplifies the calculations.
We let $\abs{h}_{T_M} = \sum_{(q,a,q') \in T_M}\abs{h}_{q,a}$.

Let us fix $\eta$ as follows
\[
  \eta<\frac{1}{2}\min\left(\{\abs{\delta_e(q,a)(q') - \delta_{f}(q,a)(q')} \mid e,f \in E, q,q'\in Q, a \in \Act\}\setminus\{0\}\right).
\]

\def\good{\textrm{\sf Good}}
Let us define the set of \emph{good histories with $n_0$ samples}, denoted $\Good_{n_0}$, as the set of histories $h$ satisfying
\begin{itemize}
  \item $\abs{h}_{T_M}\geq n_0$,
  \item for all $t=(q,a,\_)\in T_M$ satisfying $\abs{h}_{q,a} {=} \max_{(q',a',\cdot) \in T_M} \abs{h}_{q',a'}$, 
  $h$ is a good $(t,\eta)$-classification.
\end{itemize}

  \begin{lemma}
    \label{lemma:avoiding-cecs}
    Consider an MEMDP~$M$ whose only non-distinguishing MCECs are trivial, and fix $\epsilon>0$, 
    Let $n_0 = \lceil\frac{2(\abs{Q}\abs{\Act})^3}{\epsilon\eta^2} \rceil$, and 
    $n \geq 2 p^{-2\abs{Q}} \max(\log(\frac{4}{\epsilon}),  n_0)$
    where $p$ is the smallest nonzero probability that appears in $M$.
    Then, from any starting state, and under any strategy, with probability at least $1-\epsilon$,
    within $n$ steps, the history either visits a MCEC
    (distinguishing or trivial), or belongs to $\Good_{n_0}$.
  \end{lemma}

\begin{proof}
  We show that in all $M[e]$, under any strategy $\sigma$, from every state~$q_0$, there is a path of size at most $\abs{Q}$ compatible with the strategy that 
  reaches a MCEC or a distinguishing transition. 
  Consider first the case of a pure strategy $\sigma$.
  To prove this, towards a contraction, assume that MCECs and distinguishing transitions are not visited within $\abs{Q}$ steps under $\sigma$. Consider the execution tree that starts at~$q_0$ in $M$ under $\sigma$:
  this is a tree labeled by $Q$, in which the children of a given node at history $h$ are labeled by all possible successors $\Supp(\delta_e(\last(h),\sigma(h)))$ for some $e\in E$.
  Since all transitions are non-distinguishing, the choice of $e$ is irrelevant here.
  We build this tree and cut each branch whenever a MCEC or a distinguishing transition is seen, or a state is repeated.
  Since we assumed that MCECs and distinguishing transitions are not reachable under $\sigma$, all branches of this tree are cut only when a state is repeated.
  It follows that the set of states in this tree, together with the actions prescribed by $\sigma$ from these histories
  form a closed set of states. But then a strongly-connected subset must exist, which is a non-distinguishing CEC. This is thus included in a MCEC,  
  contradicting our assumption.
  
  If $\sigma$ is pure, then in all $M[e]$, from every history, there is a probability of at least $p^{\abs{Q}}$ of either taking a distinguishing transition,
  or visiting a MCEC within $\abs{Q}$ steps, independently of the current state. 
  If $\sigma$ is randomized, then the probability of such a single run can be smaller since $\sigma$ might assign small probabilities to its actions.
  In this case, since we are only interested in the behaviors in the first $n$ steps, we can see $\sigma$ as a \emph{mixed} strategy which consists in randomly choosing
  among a set of pure strategies that stop after $n$ steps.
  Since the above argument can be applied to each pure strategy in the support of $\sigma$ (when $\sigma$ is seen as a mixed strategy), it follows that under $\sigma$,
  there is a probability of at least $p^{\abs{Q}}$ of taking a distinguishing transition
  or visiting a MCEC within the next $\abs{Q}$ steps, as well. 

  Viewing runs as the concatenation of finite segments of size $\abs{Q}$, we call each such segment a trial.
  Consider the random Bernoulli variables $X_1,X_2,\ldots$
  such that the value of $X_i$ is 1 iff a MCEC or a distinguishing transition is visited at the $i$-th trial.

  So by Hoeffding's inequality, for all states $q_0$, $n>0$ and $t>0$,
  \footnote{Note that Hoeffding's inequality requires an independent sequence of random variables
  which is not the case of the $X_i$'s. We can nevertheless still apply this inequality 
  here using a coupling argument:
  Define $U_i$ as a sequence of independent and continuous variables uniformly distributed over $[0,1]$.
  Define the Bernoulli variable $Y_i=1$ iff $U_i \leq p^{\abs{Q}}$. Furthermore, 
  define the sequence of Bernoulli variables $\tilde{X}_1,\tilde{X}_2,\ldots$ inductively, by
  $\tilde{X}_i=1$ iff $U_i \leq \pr(X_i=1 \mid X_1=\tilde{X}_1,\ldots, X_{i-1}=\tilde{X}_{i-1})$.
  Because $\pr(X_i=1 \mid h)\geq p^{\abs{Q}}$
  regardless of the history $h$, we have $Y_i \leq \tilde{X}_i$.
  Furthermore, $\pr(X_1=x_1, \ldots,X_n=x_n) = \pr(\tilde{X}_1=x_1, \ldots, \tilde{X}_n=x_n)$
  for all $x_1,\ldots,x_n \in \{0,1\}$. It follows that Hoeffding's inequality can be applied on
  the i.i.d. sequence $Y_i$, and we get for all $A>0$, $\pr[\sum_i X_i \leq A] \leq \pr[\sum_i Y_i \leq A]$.
  }
  \[
    \pr_{q_0}^\sigma\left[\sum_{i=1}^n X_i \leq \sum_{i=1}^n \expect^\sigma[X_i]  - t \right]\leq 2e^{-2\frac{t^2}{n}}.
  \]
  Given $n>0$, we  choose here $t = np^{\abs{Q}}/2$. 
  This yields,
  \[
    \pr_{q_0}^\sigma\left[\sum_{i=1}^n X_i \leq \sum_{i=1}^n \expect^\sigma[X_i]  - np^{\abs{Q}}/2 \right]\leq 2e^{-2\frac{n^2p^{2\abs{Q}}}{4n}} \leq \epsilon/2,
  \]
  which is the case since, by taking the $\log$ of both sides,
  \begin{align*}
    -\frac{np^{2\abs{Q}}}{2} &\leq \log(\epsilon/4)\\
    \Leftrightarrow&n \geq 2\log(4/\epsilon)p^{-2\abs{Q}}.
  \end{align*}

  Because $\expect^\sigma(X_i) \geq p^{\abs{Q}}$,
  $\sum_{i=1}^n \expect^\sigma[X_i]\geq np^{\abs{Q}}$. 
  This means that with probability at least $1-\epsilon/2$, 
  $\sum_{i=1}^n X_i\geq np^{\abs{Q}}/2$.
  As $n\geq 2  \lceil \frac{2(\abs{Q}\abs{\Act})^3}{\epsilon\eta^2} \rceil p^{-2\abs{Q}}$,
  we have $\sum_{i=1}^n X_i\geq\lceil\frac{2(\abs{Q}\abs{\Act})^3}{\epsilon\eta^2}\rceil$,
  that is, 
  with probability at least $1-\epsilon/2$, 
  either a MCEC or $\lceil\frac{2(\abs{Q}\abs{\Act})^3}{\epsilon\eta^2} \rceil$ occurrences of distinguishing transitions are seen
  (which can be good or bad classifications).

  Let us write $n_0=\lceil \frac{2(\abs{Q}\abs{\Act})^3}{\epsilon\eta^2} \rceil$.
  It remains to bound the probability of visiting either a MCEC or $\Good_{n_0}$.
  Let us define a tree-shaped MDP $M_n$ from $M$ as follows. First, we unfold $M$ by stopping each branch either when a MCEC is reached, or after $n$ steps.
  Then, each leaf that belongs to a MCEC is extended with 
  fresh states and transitions so that the branch contains $n_0$ instances of distinguishing transitions.
  More precisely, we pick some distinguishing transition $(q,a,q')$ of $M$, and extend a given leaf $l_0$ of $M_n$ as follows.
  The only enabled action at $l_0$ is $a$, and it goes to $l_0'$ with probability $\delta_e(q,a,q')$ to $l_0'$ in $M_n[e]$,
  and to $l_0''$ with probability $1-\delta_e(q,a,q')$; and both $l_0',l_0''$ deterministically go to $l_1$. We repeat this until $n_0$ 
  occurrences of distinguishing transitions are obtained.
  Last, all leafs are made into absorbing states.
  
  Let $\diamond \textrm{CEC}$ denote the set of histories that reach a MCEC.
  For all $e\in E$,
  \[
    \pr_{q_0}^\sigma[M[e], \diamond \textrm{CEC} \lor \Good_{n_0}] \geq
    \pr_{q_0'}^\sigma[M_n[e], \Good_{n_0}]
  \]
  where $q_0'$ is the root of $M_n[e]$,
  since MCECs are replaced with a gadget that might not satisfy $\Good_{n_0}$ with probability 1.

  As an additional step, we obtain $M_n'$ by modifying $M_n$ as follows: we extend each leaf whose branch does not contain $n_0$ occurrences of distinguishing transitions (nor visit a MCEC),
  by adding fresh states and transitions as described above so that a total of $n_0$ distinguishing transitions is obtained at each branch.
  We get for all $ e\in E$.
  \[
    \pr_{q_0'}^\sigma[M_n[e], \Good_{n_0}] \geq \pr_{q_0'}^\sigma[M_n'[e], \Good_{n_0}] - \epsilon/2
  \]
  since the probability of the modified branches was shown to be at most $\epsilon$ above.
  
  Now, by construction, for all strategies $\sigma$ and $e\in E$,
  $n_0$ occurrences of distinguishing transitions are seen in $M_n'[e]$ with probability 1.
  By Lemma~\ref{lemma:adversarial-sampling} with $k=\abs{Q}\abs{\Act}$, applied for $\epsilon/2$, we get
  \[
    \pr_{q_0'}^\sigma[M_n'[e], \Good_{n_0}] \geq 1 - \epsilon/2.
  \]
  It follows that
  $\pr_q^\sigma[M[e], \diamond \textrm{CEC} \lor \Good_{n_0}] \geq
  \pr_{q_0'}^\sigma[M_n'[e], \Good_{n_0}] \geq 1 - \epsilon$ for all $e \in E$, as required.
\end{proof}

\subsection{Constructing Approximate Finite-Memory Strategies}
\label{section:constructing-finite-memory-strategies}

We are now ready to construct a finite-memory strategy that approximates an arbitrary strategy~$\sigma$.
We construct a finite-memory strategy for $\purge{M}$ and then transfer it to $M$ 
using Lemmas~\ref{lemma:removing-nondist-cecs-stratmap}-\ref{lemma:removing-nondist-cecs-stratmap-2}.
The finite-memory strategy we construct consists in imitating the strategy~$\sigma$ for~$n$ steps, where $n$ is defined in Lemma~\ref{lemma:avoiding-cecs}.
Because all nontrivial MCECs of $\purge{M}$ are distinguishing, when we play for $n$ steps, with high probability,
we either visit a trivial MCEC (which is either winning for all environments or losing for all environments), 
or reach a distinguishing MCEC, or observe enough samples of distinguishing transitions.
The strategy is extended arbitrarily in trivial MCECs. Inside distinguishing MCECs, it gathers samples of distinguishing transitions
as in Lemma~\ref{lemma:strategy-for-dist-mecs}, which improves the knowledge (with an arbitrarily small probability of error). The knowledge is also correctly improved with high probability if enough samples are gathered outside of MCECs. 
In both cases, the strategy switches to a finite-memory strategy for the improved knowledge constructed recursively for smaller sets of environments.

Lemma~\ref{lemma:quantitative-bounded-memory} formalizes this reasoning and gives a bound $N$ on the memory of the resulting strategy. In the memory bound, 
the term
$\lceil 8\frac{\log(8/\epsilon)}{\eta^2}^2\rceil$ comes from the application of Lemma~\ref{lemma:strategy-for-dist-mecs} for distinguishing MCECs for each subset of $E$; and the term $(2\abs{Q})^{n(\abs{E}+1)}$ corresponds to the recursive analysis, since the strategy is defined inductively for each subset of $E$.

  \begin{lemma}
    \label{lemma:quantitative-bounded-memory}
    Consider an MEMDP $M=\tuple{Q,A,(\delta_e)_{e\in E}}$, state $q \in Q$, parity objective $\varphi$.
    For all strategies $\sigma$, and $\epsilon>0$,
    there exists a strategy $\sigma'$ 
    using at most 
    $N = (2\abs{Q})^{n(\abs{E}+1)}\abs{\Act}\lceil 8\frac{\log(8/\epsilon)}{\eta^2}^2\rceil$
    memory 
    where $n = \left\lceil 2p^{-2\abs{Q}}\max(\frac{8(\abs{Q}\abs{\Act})^3}{\epsilon\eta^2}, \log(16/\epsilon))\right\rceil$,
    with $p$ the smallest nonzero probability in $M$,
    and that satisfies $\pr_q^{\sigma'}(M, \varphi) \geq \pr_q^{\sigma}(M, \varphi)-\epsilon$.  
 \end{lemma}

\begin{proof}
  Given $\epsilon>0$, let    
  \begin{align*}
    &n_0=\left\lceil \frac{8(\abs{Q}\abs{\Act})^3}{\epsilon\eta^2}\right\rceil,\\
    &n = \left\lceil 2p^{-2\abs{Q}}\max(n_0, \log(16/\epsilon))\right\rceil,
  \end{align*}
  Notice that the bounds on $n$ and $n_0$ come from Lemma~\ref{lemma:avoiding-cecs} applied for $\epsilon/4$.
  Define the sequence $(g_i)_{i\geq 1}$ by $g_1 = 1$, and
  $g_i = \alpha(g_{i-1} + \beta) + \gamma$
  where $\alpha=2\abs{Q}^n$, $\beta=\lceil 8\frac{\log(8/\epsilon)}{\eta^2}^2\rceil$, and $\gamma = \abs{Q}\abs{A}$.
  Note that we have, for $i>1$, 
  $g_i = \alpha^{i-1} + (\gamma + \alpha\beta)(\frac{\alpha^{i-1}-1}{\alpha-1})$.  
  Observe that $g_i \leq \alpha^{i-1}(1+\gamma+\alpha\beta)$.

  We prove, by induction on $\abs{E}$, that for all states~$q$, strategies $\sigma$, 
  there exists a $g_{\abs{E}}$-memory strategy $\sigma'$ such that $\pr_q^{\sigma'}(M, \varphi) \geq \pr_q^{\sigma}(M, \varphi) - \epsilon$.
  
  We have $g_{\abs{E}} \leq \alpha^{\abs{E}-1}(1+\gamma+\alpha\beta) \leq \alpha^{\abs{E}}(1 + \abs{Q}\abs{A} + 2\abs{Q}^n\beta)\leq
  \alpha^{\abs{E}}(3\abs{Q}^n\abs{\Act}\beta) \leq \alpha^{\abs{E}}(2\alpha\abs{\Act}\beta) $, which is at most
  $(2\abs{Q})^{n(\abs{E}+1)}\abs{\Act}\lceil 8\frac{\log(8/\epsilon)}{\eta^2}^2\rceil$, and proves the lemma.
  
  The base case $\abs{E}=1$ is obvious since there exists optimal memoryless strategies for parity objectives
  in MDPs. Assume $\abs{E}\geq 2$.

  Let $M'=\purge{M}$ and $\sigma'$ be given by Lemma~\ref{lemma:removing-nondist-cecs-stratmap} such that 
  $\pr_{q}^\sigma(M,\varphi) \leq \pr_{q'}^{\sigma'}(M',\varphi)$ where $q' = f(q)$.
  We prove the property for $M'$ and transfer the result back to $M$ using Lemma~\ref{lemma:removing-nondist-cecs-stratmap-2}.
  More precisely, we show below that there exists a $(g_{\abs{E}}-\gamma)$-memory strategy $\sigma''$ with 
  $\pr_{q'}^{\sigma''}(M',\varphi) \geq \pr_{q'}^{\sigma'}(M',\varphi) - \epsilon$.
  It follows, by Lemma~\ref{lemma:removing-nondist-cecs-stratmap}, that there exists a $g_{\abs{E}}$-memory strategy $\sigma'''$
  for $M$  such that 
  \[
    \pr_{q}^{\sigma'''}(M,\varphi) = \pr_{q'}^{\sigma''}(M',\varphi) \geq \pr_{q'}^{\sigma'}(M',\varphi)-\epsilon \geq \pr_{q}^\sigma(M,\varphi)-\epsilon 
  \]
  which proves the result.

  We construct $\sigma''$ by imitating $\sigma'$ for $n$ steps, and stopping if a MCEC is reached (thus, either trivial or distinguishing, by Lemma~\ref{lemma:removing-nondist-cecs-1}).
  More precisely, consider history $h$ in $M'$ that starts at $q'$. We define $\sigma''(h) = \sigma'(h)$, except in the following cases
  where $\sigma''$ switches to a strategy as described below:
  \begin{enumerate}
    \item If $\last(h)$ belongs to a trivial MCEC, then $\sigma''$ is memoryless from that history (as there is only one possible action to choose).
    Notice that $\pr_{q'}^{\sigma''}(M',\varphi \mid h ) = \pr_{q'}^{\sigma'}(M',\varphi\mid h)$ since this MCEC is either winning or losing with probability 1,
    in each $e\in E$.
    
    \item Assume $\last(h)$ belongs to a distinguishing MCEC~$D$ with partition $(K_1,K_2)$. 
    Let $\vec{\beta} = \pr^{\sigma'}(M',\varphi \mid h)$, the probability values achieved from history $h$ under strategy $\sigma'$ starting with history $h$.
    One can define a strategy $\sigma_h'$ such that 
    $\vec{\beta} = \pr_{\last(h)}^{\sigma_h'}(M',\varphi)$, by $\sigma_h' : h' \mapsto \sigma'(h \cdot h')$.
    By induction applied to $M'$, $\last(h)$, $\sigma_h'$, environment set $K_i$, 
    and $\epsilon/8$,
    there exist $g_{\abs{K_i}}$-memory strategies $\sigma_i$, 
    with $\pr_{\last(h)}^{\sigma_i}(M'[K_i],\varphi)\geq \vec{\beta}{\vert}_{K_i} - \epsilon/8$.
    We apply Lemma~\ref{lemma:strategy-for-dist-mecs} to build strategy $\sigma''_h$ satisfying the following:
    \begin{equation}
      \label{eqn:quant-switch-dist-cec}
      \pr_{\last(h)}^{\sigma''_h}(M'[e], \varphi) \geq \pr_{\last(h)}^{\sigma_i}(M'[e],\varphi) - \epsilon/4 \text{ for all environments } e \in K_i.
    \end{equation}
    
    At $h$, we let $\sigma''$ switch to $\sigma''_h$.
    It follows that $\pr_{q'}^{\sigma''}(M',\varphi \mid h ) \geq \pr_{q'}^{\sigma'}(M',\varphi \mid h) - \epsilon/4$.

    \item Assume that $h$ contains $n_0$ occurrences of distinguishing state-action pairs, that is, $\abs{h}_{T_{M'}} = n_0$.
    Let $(q,a,q') \in T_M$ be a distinguishing transition with the largest number of occurrences in $h$; and let $(K_1,K_2)$ be the partition of $E$ induced by this transition.
    For each $i=1,2$, let $\sigma_i$ be the $g_{\abs{K_i}}$-memory strategy given by induction hypothesis applied to $M'$, state $\last(h)$, environment set $K_i$,
    bound $\epsilon/4$, and strategy $\sigma_h' : h' \mapsto \sigma'(h \cdot h')$ that achieves 
    $\pr_{\last(h)}^{\sigma_i}(M'[K_i],\varphi)\geq \pr_{\last(h)}^{\sigma_h'}(M'[K_i],\varphi) - \epsilon/4$ for each $i=1,2$.
    We let $\sigma''$ switch to:
    \begin{itemize}
      \item $\sigma_1$ if $\left\lvert\frac{\abs{h}_{q,a,q'}}{\abs{h}_{q,a}} - \delta_{e}(q,a)(q')\right\rvert < \eta/2$ for some $e \in K_1$,
      \item $\sigma_2$ otherwise.
    \end{itemize}
    
    The above shows that if $h$ is a good classification in $e$, then 
    $\pr_{q'}^{\sigma''}(M'[e],\varphi \mid h ) \geq \pr_{q'}^{\sigma'}(M'[e],\varphi \mid h) - \epsilon/4$.

    \item If $\abs{h} = n$ and none of the above applies, then $\sigma''$ switches to an arbitrary memoryless strategy.
    These histories that satisfy case 4 has probability at most $\epsilon/4$ by Lemma~\ref{lemma:avoiding-cecs}.
  \end{enumerate}
  
  \medskip
  Let us show that $\pr_{q'}^{\sigma''}(M',\varphi) \geq \pr_{q'}^{\sigma'}(M',\varphi) - \epsilon$.  
  To prove this, we distinguish histories $h$ according to the cases above, 
  and relate $\pr_{q'}^{\sigma''}(M',\varphi \mid h )$ and $\pr_{q'}^{\sigma'}(M',\varphi\mid h)$,
  and bound the probability of some histories $h$.
  
  Let us write
  \begin{align*}
    \pr_{q'}^{\sigma''}(M'[e], \varphi ) &= \sum_{h : \text{ case 1}} \pr_{q'}^{\sigma''}(M'[e],\varphi, h) + \sum_{h \text{ case 2}} \pr_{q'}^{\sigma''}(M'[e],h) \pr_{q'}^{\sigma''}(M'[e],\varphi \mid h)\\
                                      &+ \sum_{\substack{h : \text{ case 3}\\\text{bad classification}}} \pr_{q'}^{\sigma''}(M'[e],h) \pr_{q'}^{\sigma''}(M'[e],\varphi \mid h)\\
                                      &+ \sum_{\substack{h : \text{ case 3}\\\text{good classification}}} \pr_{q'}^{\sigma''}(M'[e],h) \pr_{q'}^{\sigma''}(M'[e],\varphi \mid h)\\
                                      &+ \sum_{h : \text{ case 4}} \pr_{q'}^{\sigma''}(M'[e],h) \pr_{q'}^{\sigma''}(M'[e],\varphi \mid h).\\
  \end{align*}
  Since $\pr_{q'}^{\sigma''}(M',h ) = \pr_{q'}^{\sigma'}(M', h)$ for histories satisfying any of the cases (because $\sigma''$ imitates $\sigma'$ until such a case occurs),
  and because the terms $\pr_{q'}^{\sigma''}(M',h )$ at the second and forth lines are each at most $\epsilon/4$, using the cases above, we get
  \begin{align*}
    \pr_{q'}^{\sigma''}(M'[e], \varphi ) &{\geq} \sum_{h : \text{case 1}} \pr_{q'}^{\sigma'}(M'[e],\varphi, h) {+} \sum_{h:\text{ case 2}} \pr_{q'}^{\sigma'}(M'[e],h) (\pr_{q'}^{\sigma'}(M'[e],\varphi{\mid} h){-}\epsilon{/}4)\\
                                      &+ \left(\sum_{\substack{h : \text{ case 3}\\\text{bad classification}}} \pr_{q'}^{\sigma'}(M'[e],h) \pr_{q'}^{\sigma'}(M'[e],\varphi \mid h)\right)-\epsilon/4\\
                                      &+ \sum_{\substack{h : \text{ case 3}\\\text{good classification}}} \pr_{q'}^{\sigma'}(M'[e],h) (\pr_{q'}^{\sigma'}(M'[e],\varphi \mid h)-\epsilon/4)\\
                                      &+ \left(\sum_{h : \text{ case 4}} \pr_{q'}^{\sigma''}(M'[e],h) \pr_{q'}^{\sigma''}(M'[e],\varphi \mid h)\right)-\epsilon/4.\\
                                      &\geq \pr_{q'}^{\sigma'}(M'[e], \varphi ) -\epsilon.
  \end{align*}
  Last, we argue that $\sigma''$ uses memory of size $g_{\abs{E}}$. 
  Strategy $\sigma''$ must store the histories until one of the four cases occur: this happens in at most $n$ steps, which means $\abs{Q}^{n}$ memory is required for this phase.
  In addition, for each history of case 2, $g_{\abs{E}-1} + g_{\abs{E}-1} + 8\frac{\log(8/\epsilon)}{\eta^2}^2$ memory states are needed by Lemma~\ref{lemma:strategy-for-dist-mecs};
  where the terms $g_{\abs{E}-1}$ are upper bounds on the memory requirement of the strategies to which
  we switch, given by induction.
  Case 3 does not require additional memory since the decision is made depending on the current history, which is already in the memory. In total, we thus need
  $\abs{Q}^n(2g_{\abs{E}-1} + \beta)$ memory states, which is at most $\alpha(g_{\abs{E}-1} + \beta) =
  g_{\abs{E}}-\gamma$.
\end{proof}

\subsection{Approximation Algorithm}\label{sec:approx_alg}

We now provide a procedure solving the gap problem with threshold $\alpha$ for parity objectives in MEMDPs.
Informally, given bound $N$, the procedure \emph{guesses} an $N$-memory strategy by solving 
a set of polynomial constraints over the reals,
and checks that the strategy ensures winning with probability at least $\alpha - \epsilon$
in all environments.
We first give the construction for reachability, then explain the extension to parity conditions.

\smallskip
\myparagraph{Reachability in MDPs}
\def\Qno{Q^{\textsf{no}}}
\def\Qquestion{Q^{?}}
Let us start by recalling the linear constraints that characterize 
reachability probabilities in single-environment MDPs under memoryless strategies.
Consider an MDP $M=\tuple{Q,A,\delta}$ and objective $\Reach(T)$.
Let $\Qno \subseteq Q$, and $\Qquestion = Q\setminus (\Qno \cup T)$.
$\Qno$ will be the set of states from which the reachability probability is $0$; it is necessary to make sure all such states are in $\Qno$ 
so that the equation given below has a unique solution.
Define the unknown $x_q$ representing the probability of reaching $T$ from $q$ under the strategy that is being guessed,
and $p_{q}(a)$ the probability of the strategy to pick action $a$ from~$q$, for $a \in A_q$.
Consider the following constraints:
\begin{equation}
  \label{eqn:mdp-reach}
  \!\!\!\!\begin{array}{lll}
    &x_q = 0 & \text{ for all } q \in \Qno, \\
    &x_q = 1 & \text{ for all } q\in T, \\   
    &x_q = \sum_{a \in A_q} p_{q}(a) \cdot \sum_{q' \in Q} \delta(q,a,q') \cdot x_{q'} & \text{ for all } q \in \Qquestion,\\
    &0 \leq x_q \leq 1 \text{ and } 0 \leq p_{q}(a) \leq 1 & \text{ for all } q \in Q, a \in A_q,\\
    &\sum_{a \in A_q} p_{q}(a) =1 & \text{ for all } q \in Q.
  \end{array}
\end{equation}

Any solution $(\vec{x}, \vec{p})$ of \eqref{eqn:mdp-reach} yields a strategy $\sigma^{\vec{p}}$, 
which is defined as picking action $a$ from state $q$ with probability $p_{q}(a)$. 
The following theorem shows that $\vec{x}$ does capture the reachability probabilities
of $\sigma^{\vec{p}}$, provided that $\Qno$ is the set of states from which 
the reachability probability is $0$. 

\begin{theorem}[Theorem 10.19, \cite{BK08}]
  \label{thm:mdp-reach-charact}
  Consider any subset $\Qno \subseteq Q$, and a solution $(\vec{x}, \vec{p})$ of \eqref{eqn:mdp-reach}.
  If for all states $q \in \Qno$, $\pr^{\sigma^{\vec{p}}}_q[M, \Reach(T)] = 0$,
  then for all $q\in Q$, $x_q = \pr^{\sigma^{\vec{p}}}_q[M, \Reach(T)]$.
  Conversely, for any memoryless strategy $\tau$, if $\Qno$ denotes the set 
  of states $q$ with $\pr^{\tau}_q[M, \Reach(T)] = 0$, then 
  \eqref{eqn:mdp-reach} has a unique solution $(\vec{x}, \vec{p})$ where $\tau = \sigma^{\vec{p}}$,
  and $x_q = \pr^{\tau}_q[M, \Reach(T)]$ for all $q \in Q$.
\end{theorem}

\myparagraph{Finite-Memory Reachability in MEMDPs}
We now show how to solve the gap problem for an instance of the quantitative reachability problem for MEMDPs.
Consider MEMDP $M=\tuple{Q,A,(\delta_e)_{e\in E}}$, objective $\Reach(T)$, a memory bound $N$,
an initial state $q_0$, and bounds $\epsilon,\alpha>0$.
We want to check whether there exists a strategy $\sigma$ such that, for all environments $e \in E$, 
we have $\pr_{q_0}^\sigma[M[e],\Reach(T)]\geq \alpha$,
or whether for all $\sigma$, there exists an environment $e \in E$ with 
$\pr_{q_0}^\sigma[M[e],\Reach(T)]< \alpha-\epsilon$.

We guess a memoryless randomized strategy on combined states $(q,i)$ for $q \in Q$ and $0\leq i< N$,
which correspond to $N$-memory strategies on $M$.
In the sequel, we write $[N] = \{0,1,\dots,N-1\}$.
We define the unknown variable $x_{q,i}^e$ for each $e \in E$, and combined state $(q,i)$
representing the probability of reaching $T$ from state $q$ and memory value~$i$ in $M[e]$, under the strategy that is being guessed.
Furthermore, define $p_{q,i}(a,i')$ for each action $a \in A_q$ and $i' \in [N]$, as the unknown representing 
the probability of the strategy picking action $a$ from $(q,i)$ and updating the memory value to $i'$.

Consider subsets $\Qno_e \subseteq Q$ for each $e \in E$,
and let $\Qquestion_e = Q \setminus (\Qno_e \cup T_e)$.
We write the following constraints in a slightly more general setting, 
where a possibly different target set $T_e$ is considered for each environment $e$
(this will be useful when generalizing to parity conditions below):
\begin{equation}
  \label{eqn:memdp-reach}
  \!\!\!\!\begin{array}{lll}
    &x_{q,i}^e = 0  & \text{ for all } e \in E, q \in \Qno_e, i \in [N], \\
    &x_{q,i}^e = 1  & \text{ for all } e \in E, q \in T_e,    i \in [N], \\
    &x_{q,i}^e = \sum_{a \in A_q} p_{q,i}(a,i') \cdot \sum_{q' \in Q} \delta_e(q,a,q') \cdot x_{q',i'}^e  & \text{ for all } e \in E, q \in \Qquestion_e, i \in [N], \\
    &0 \leq x_{q,i}^e \leq 1  & \text{ for all } e \in E, q \in Q, i \in [N], \\
    &0 \leq p_{q,i}(a,i') \leq 1  & \text{ for all } q \in Q, a \in A_q, i,i' \in [N], \\
    &\sum_{a \in A_q}\sum_{i' \in [N]} p_{q,i}(a,i') = 1  & \text{ for all } q \in Q, i \in [N],\\
    &x_{q_0}^e \geq \alpha-\epsilon  & \text{ for all } e \in E.
  \end{array}
\end{equation}

Notice how the choice of the action and memory updates $p_{q,i}$ does not depend on the environment.
The constraints \eqref{eqn:memdp-reach} simply combine $\abs{E}$ copies of \eqref{eqn:mdp-reach}
over a state space augmented with $N$ memory values. In addition 
we added the constraints $x_{q_0}^e \geq \alpha-\epsilon$ for all $e \in E$, 
which restrict the solution sets to those strategies that ensure the threshold $\alpha-\epsilon$.

\smallskip
\myparagraph{The Gap Problem for Reachability}
The full procedure is as follows. We let $T_e = T$ for all $e \in E$.
Let $N$ be the memory bound given in Lemma~\ref{lemma:quantitative-bounded-memory}.

We enumerate all possibles choices for the sets $\Qno_e$.
For each choice $(\Qno_e)_{e\in E}$, we solve the corresponding constraints~\eqref{eqn:memdp-reach}.
If there is no solution, we continue with the next choice.
Otherwise let $\sigma^{\vec{p}}$ be the $N$-memory strategy given by the solution to this equation.
If $\pr_q^{\sigma^{\vec{p}}}[M[e], \Reach(T_e)] = 0$ for each $e \in E$ and $q \in \Qno_e$, 
then we return \textsf{Yes};
otherwise we continue with the next choice $(\Qno_e)_{e\in E}$. 
We return \textsf{No} at the end of if no solution was found.

Let us show that this procedure solves the gap problem.
Assume that there exists a strategy $\tau$ such that for all environments $e \in E$, 
we have $\pr_{q_0}^{\tau}[M[e],\Reach(T)]\geq \alpha$.
By Lemma~\ref{lemma:quantitative-bounded-memory}, there exists a $N$-memory strategy $\tau'$ 
such that in all environments $e \in E$, we have $\pr_{q_0}^{\tau'}[M[e],\Reach(T)]\geq \alpha-\epsilon$.
Hence \eqref{eqn:memdp-reach} must have a solution corresponding to this strategy for some choice of the sets $(\Qno_e)_{e\in E}$, and
the procedure returns \textsf{Yes}.
Assume now that no strategy achieves the threshold $\alpha-\epsilon$. In particular, no $N$-memory strategy
achieves this threshold, and the procedure returns \textsf{No}.

We now analyze the complexity of the procedure.
The value $N$ is double exponential in the size of the input, which means that the size of \eqref{eqn:memdp-reach} is also double exponential.
Polynomial equations can be solved in polynomial space in the size of the equations~\cite{Canny-stoc1988},
so here we can solve \eqref{eqn:memdp-reach} in double exponential space.

\smallskip
\myparagraph{The Gap Problem for Parity}
We now extend the previous procedure to solve the quantitative parity gap problem based on the following observations.
In $M[e]$, any finite-memory strategy $\sigma$ induces a Markov chain. Then the probability $\pr_{q_0}^\sigma[M[e], \varphi]$
of satisfying a parity condition $\varphi$ is equal to the probability of reaching bottom strongly connected components (BSCC) 
that are winning\footnote{Recall that a BSCC is winning for a parity condition if the smallest priority of its states is even.} 
for~$\varphi$ in the induced Markov chain~\cite{BK08}. 
But the set of BSCCs only depends on the support of $\sigma$, that is, the set of state-action pairs that have positive probability.
When considering an MDP under $N$-memory strategies, the support is the set of tuples $(q,i,a,i')$ such that
from combined state $(q,i)$ the strategy has a nonzero probability of picking action $a$ and updating memory to $i'$.

We proceed as follows. Given MEMDP~$M$, initial state $q_0$, parity condition $\varphi$, and bound $N$, we enumerate 
all supports $S \subseteq Q \times [N] \times A \times [N]$. 
For each support $S$, let $T^S_e$ be the set of $\varphi$-winning BSCCs in $M[e]$ under a strategy with support $S$.
We apply the reachability procedure described above based on \eqref{eqn:memdp-reach} for the target sets $(T^S_e)_e$ 
augmented with the following constraints:
for all $(q,i,a,i') \in S$, we add the constraint $p_{q,i}(a,i')>0$, and for all others $p_{q,i}(a,i')=0$.
If the answer is \textsf{Yes} for some support $S$, then we return \textsf{Yes}; otherwise we return \textsf{No}.

This solves the gap problem:
if there is $\tau$ such that for all environments $e \in E$ we have $\pr_{q_0}^{\tau}[M[e],\varphi]\geq \alpha$,
then by Lemma~\ref{lemma:quantitative-bounded-memory} there exists a $N$-memory strategy $\tau'$ 
such that for all environments $e \in E$ we have $\pr_{q_0}^{\tau'}[M[e],\varphi]\geq \alpha-\epsilon$.
Let $S$ denote the support of the strategy $\tau'$, and let $T_e^S$ be the set of winning BSCCs in $M[e]$ under $\tau$.
So \eqref{eqn:memdp-reach}, instantiated for $S$
has a solution, for some choice of the sets $(\Qno_e)_{e\in E}$, and
the procedure returns \textsf{Yes}.
Assume now that no strategy achieves the threshold $\alpha-\epsilon$. In particular, there is no $N$-memory strategy
with any support $S$ that 
achieves this threshold, and the procedure returns \textsf{No}.

There are an exponential number of possibilities for the choice of support. Moreover, given a support~$S$,
each set $T^S_e$ can be determined in polynomial time. Overall, the procedure remains in double exponential space.

  \begin{theorem}
    The gap problem can be solved in double exponential space for MEMDPs with parity objectives.
  \end{theorem}

\bibliographystyle{plain}
\bibliography{biblio}

\end{document}

%% file: figures/as3.tex
\begin{gpicture}(100,36)(0,0)


\gasset{Nw=5, Nh=5, Nmr=2.5, ilength=4}

\node[Nmarks=i, iangle=90](q0)(15,22){$0$}
\node[Nmarks=n](q1)(5,10){$1$}
\node[Nmarks=n](q2)(15,10){$2$}
\node[Nmarks=n](q3)(25,10){$3$}

\node[Nframe=n,Nmarks=n, Nw=0, Nh=0](u1)(5,5){}
\node[Nframe=n,Nmarks=n, Nw=0, Nh=0, ExtNl=y, NLangle=270, NLdist=2](u2)(15,5){back to $0$}
\node[Nframe=n,Nmarks=n, Nw=0, Nh=0](u3)(25,5){}

\node[Nframe=n](label)(15,32){$M[e_1]$}

\drawarc[linegray=0](15,22,4,270,316.5)

\drawedge[ELside=r,ELpos=50, ELdist=1, curvedepth=0](q0,q2){$\frac{1}{2}$}
\drawedge[ELside=l,ELpos=50, ELdist=1, eyo=2.5, curvedepth=0](q0,q3){$\frac{1}{2}$}

\drawedge[ELside=l,ELpos=50, ELdist=1, curvedepth=0](q1,u1){}
\drawedge[ELside=l,ELpos=50, ELdist=1, curvedepth=0](q2,u2){}
\drawedge[ELside=l,ELpos=50, ELdist=1, curvedepth=0](q3,u3){}


\node[Nmarks=i, iangle=90](q0)(50,22){$0$}
\node[Nmarks=n](q1)(40,10){$1$}
\node[Nmarks=n](q2)(50,10){$2$}
\node[Nmarks=n](q3)(60,10){$3$}

\node[Nframe=n,Nmarks=n, Nw=0, Nh=0](u1)(40,5){}
\node[Nframe=n,Nmarks=n, Nw=0, Nh=0, ExtNl=y, NLangle=270, NLdist=2](u2)(50,5){back to $0$}
\node[Nframe=n,Nmarks=n, Nw=0, Nh=0](u3)(60,5){}

\node[Nframe=n](label)(50,32){$M[e_2]$}

\drawarc[linegray=0](50,22,4,223.5,316.5)

\drawedge[ELside=r,ELpos=50, ELdist=1, eyo=2.5, curvedepth=0](q0,q1){$\frac{1}{2}$}
\drawedge[ELside=l,ELpos=50, ELdist=1, eyo=2.5, curvedepth=0](q0,q3){$\frac{1}{2}$}

\drawedge[ELside=l,ELpos=50, ELdist=1, curvedepth=0](q1,u1){}
\drawedge[ELside=l,ELpos=50, ELdist=1, curvedepth=0](q2,u2){}
\drawedge[ELside=l,ELpos=50, ELdist=1, curvedepth=0](q3,u3){}


\node[Nmarks=i, iangle=90](q0)(85,22){$0$}
\node[Nmarks=n](q1)(75,10){$1$}
\node[Nmarks=n](q2)(85,10){$2$}
\node[Nmarks=n](q3)(95,10){$3$}

\node[Nframe=n,Nmarks=n, Nw=0, Nh=0](u1)(75,5){}
\node[Nframe=n,Nmarks=n, Nw=0, Nh=0, ExtNl=y, NLangle=270, NLdist=2](u2)(85,5){back to $0$}
\node[Nframe=n,Nmarks=n, Nw=0, Nh=0](u3)(95,5){}

\node[Nframe=n](label)(85,32){$M[e_3]$}

\drawarc[linegray=0](85,22,4,223.5,270)

\drawedge[ELside=r,ELpos=50, ELdist=1, eyo=2.5, curvedepth=0](q0,q1){$\frac{1}{2}$}
\drawedge[ELside=l,ELpos=50, ELdist=1, curvedepth=0](q0,q2){$\frac{1}{2}$}

\drawedge[ELside=l,ELpos=50, ELdist=1, curvedepth=0](q1,u1){}
\drawedge[ELside=l,ELpos=50, ELdist=1, curvedepth=0](q2,u2){}
\drawedge[ELside=l,ELpos=50, ELdist=1, curvedepth=0](q3,u3){}


\end{gpicture}

%% file: figures/ls3.tex
\begin{gpicture}(100,36)(0,0)


\gasset{Nw=5, Nh=5, Nmr=2.5, ilength=4}

\node[Nmarks=i, iangle=90](q0)(15,22){$0$}
\node[Nmarks=n](q1)(5,10){$1$}
\node[Nmarks=n](q2)(15,10){$2$}
\node[Nmarks=n](q3)(25,10){$3$}

\node[Nframe=n,Nmarks=n, Nw=0, Nh=0](u1)(5,5){}
\node[Nframe=n,Nmarks=n, Nw=0, Nh=0, ExtNl=y, NLangle=270, NLdist=2](u2)(15,5){back to $0$}
\node[Nframe=n,Nmarks=n, Nw=0, Nh=0](u3)(25,5){}

\node[Nframe=n](label)(15,32){$M[e_1]$}

\drawarc[linegray=0](15,22,4,223.5,316.5)

\drawedge[ELside=r,ELpos=50, ELdist=1, eyo=2.5, curvedepth=0](q0,q1){$\red{\frac{1}{2}}$}
\drawedge[ELside=r,ELpos=50, ELdist=1, curvedepth=0](q0,q2){$\frac{1}{4}$}
\drawedge[ELside=l,ELpos=50, ELdist=1, eyo=2.5, curvedepth=0](q0,q3){$\frac{1}{4}$}

\drawedge[ELside=l,ELpos=50, ELdist=1, curvedepth=0](q1,u1){}
\drawedge[ELside=l,ELpos=50, ELdist=1, curvedepth=0](q2,u2){}
\drawedge[ELside=l,ELpos=50, ELdist=1, curvedepth=0](q3,u3){}


\node[Nmarks=i, iangle=90](q0)(50,22){$0$}
\node[Nmarks=n](q1)(40,10){$1$}
\node[Nmarks=n](q2)(50,10){$2$}
\node[Nmarks=n](q3)(60,10){$3$}

\node[Nframe=n,Nmarks=n, Nw=0, Nh=0](u1)(40,5){}
\node[Nframe=n,Nmarks=n, Nw=0, Nh=0, ExtNl=y, NLangle=270, NLdist=2](u2)(50,5){back to $0$}
\node[Nframe=n,Nmarks=n, Nw=0, Nh=0](u3)(60,5){}

\node[Nframe=n](label)(50,32){$M[e_2]$}

\drawarc[linegray=0](50,22,4,223.5,316.5)

\drawedge[ELside=r,ELpos=50, ELdist=1, eyo=2.5, curvedepth=0](q0,q1){$\frac{1}{4}$}
\drawedge[ELside=r,ELpos=50, ELdist=1, curvedepth=0](q0,q2){$\red{\frac{1}{2}}$}
\drawedge[ELside=l,ELpos=50, ELdist=1, eyo=2.5, curvedepth=0](q0,q3){$\frac{1}{4}$}

\drawedge[ELside=l,ELpos=50, ELdist=1, curvedepth=0](q1,u1){}
\drawedge[ELside=l,ELpos=50, ELdist=1, curvedepth=0](q2,u2){}
\drawedge[ELside=l,ELpos=50, ELdist=1, curvedepth=0](q3,u3){}


\node[Nmarks=i, iangle=90](q0)(85,22){$0$}
\node[Nmarks=n](q1)(75,10){$1$}
\node[Nmarks=n](q2)(85,10){$2$}
\node[Nmarks=n](q3)(95,10){$3$}

\node[Nframe=n,Nmarks=n, Nw=0, Nh=0](u1)(75,5){}
\node[Nframe=n,Nmarks=n, Nw=0, Nh=0, ExtNl=y, NLangle=270, NLdist=2](u2)(85,5){back to $0$}
\node[Nframe=n,Nmarks=n, Nw=0, Nh=0](u3)(95,5){}

\node[Nframe=n](label)(85,32){$M[e_3]$}

\drawarc[linegray=0](85,22,4,223.5,316.5)

\drawedge[ELside=r,ELpos=50, ELdist=1, eyo=2.5, curvedepth=0](q0,q1){$\frac{1}{4}$}
\drawedge[ELside=r,ELpos=50, ELdist=1, curvedepth=0](q0,q2){$\frac{1}{4}$}
\drawedge[ELside=l,ELpos=50, ELdist=1, eyo=2.5, curvedepth=0](q0,q3){$\red{\frac{1}{2}}$}

\drawedge[ELside=l,ELpos=50, ELdist=1, curvedepth=0](q1,u1){}
\drawedge[ELside=l,ELpos=50, ELdist=1, curvedepth=0](q2,u2){}
\drawedge[ELside=l,ELpos=50, ELdist=1, curvedepth=0](q3,u3){}


\end{gpicture}

%% file: figures/limit-sure1.tex
\begin{gpicture}(117,38)(0,0)

\gasset{Nw=7, Nh=7, Nmr=3.5, loopdiam=5}

\node[Nmarks=i](q1)(12,15){$q_1$}
\node[Nmarks=n](q2)(32,15){$q_2$}
\node[Nmarks=r](q3)(52,25){$q_3$}
\node[Nmarks=n](q4)(52,5){$q_4$}

\node[Nframe=n](label)(32,32){$M[e_1]$}

\drawline[AHnb=0](60,1)(60,34)

\drawarc[linegray=0](12,15,6,43,66)

\drawedge[ELside=l,ELpos=50, ELdist=1, curvedepth=6](q1,q2){$\frac{1}{3}$}
\drawedge[ELside=l,ELpos=50, ELdist=1, curvedepth=6](q2,q1){$c$}

\drawedge[ELside=l, ELpos=50, ELdist=1](q2,q3){$a$}
\drawedge[ELside=r, ELpos=50, ELdist=1](q2,q4){$b$}


\drawloop[ELside=r,loopCW=n, loopangle=90](q1){$\frac{2}{3}$}
\drawloop[ELside=l,loopCW=y, loopangle=90](q3){}
\drawloop[ELside=l,loopCW=y, loopangle=90](q4){}

\node[Nmarks=i](q1)(72,15){$q_1$}
\node[Nmarks=n](q2)(92,15){$q_2$}
\node[Nmarks=r](q3)(112,25){$q_3$}
\node[Nmarks=n](q4)(112,5){$q_4$}

\node[Nframe=n](label)(92,32){$M[e_2]$}

\drawarc[linegray=0](72,15,6,43,66)

\drawedge[ELside=l,ELpos=50, ELdist=1, curvedepth=6](q1,q2){$\frac{2}{3}$}
\drawedge[ELside=l,ELpos=50, ELdist=1, curvedepth=6](q2,q1){$c$}

\drawedge[ELside=l, ELpos=50, ELdist=1](q2,q3){$b$}
\drawedge[ELside=r, ELpos=50, ELdist=1](q2,q4){$a$}

\drawloop[ELside=r,loopCW=n, loopangle=90](q1){$\frac{1}{3}$}
\drawloop[ELside=l,loopCW=y, loopangle=90](q3){}
\drawloop[ELside=l,loopCW=y, loopangle=90](q4){}


\end{gpicture}

%% file: figures/limit-sure2.tex
\begin{gpicture}(80,64)(0,0)

\gasset{Nw=7, Nh=7, Nmr=3.5, loopdiam=5}


\node[Nmarks=i](q1)(15,20){$q_1$}
\node[Nmarks=n](q2)(35,20){$q_2$}
\node[Nmarks=r](q3)(55,20){$q_3$}
\node[Nmarks=r](q4)(75,20){$q_4$}
\node[Nmarks=r](q5)(25,5){$q_5$}
\node[Nmarks=n](q6)(65,5){$q_6$}

\node[Nframe=n](label)(5,7){$M[e_2]$}

\drawline[AHnb=0](1,30)(79,30)


\drawedge[ELside=l,ELpos=50, ELdist=1, curvedepth=6](q1,q2){}
\drawedge[ELside=l,ELpos=50, ELdist=1, curvedepth=6](q2,q1){$a$}

\drawedge[ELside=l, ELpos=50, ELdist=1](q2,q3){$b$}
\drawedge[ELside=l,ELpos=50, ELdist=1, curvedepth=6](q3,q4){}
\drawedge[ELside=l,ELpos=50, ELdist=1, curvedepth=6](q4,q3){}

\drawloop[ELside=l,loopCW=y, loopangle=180](q5){}
\drawloop[ELside=l,loopCW=y, loopangle=180](q6){}


\node[Nmarks=i](q1)(15,55){$q_1$}
\node[Nmarks=n](q2)(35,55){$q_2$}
\node[Nmarks=r](q3)(55,55){$q_3$}
\node[Nmarks=r](q4)(75,55){$q_4$}
\node[Nmarks=r](q5)(25,40){$q_5$}
\node[Nmarks=n](q6)(65,40){$q_6$}

\node[Nframe=n](label)(5,42){$M[e_1]$}

\drawarc[linegray=0](35,55,6,223,236)
\drawarc[linegray=0](75,55,6,223,236)

\drawedge[ELside=l,ELpos=50, ELdist=1, curvedepth=6](q1,q2){}
\drawedge[ELside=l,ELpos=50, ELdist=1, curvedepth=6](q2,q1){$a$}
\drawedge[ELside=l,ELpos=50, ELdist=1, curvedepth=0](q2,q5){$a$}

\drawedge[ELside=l, ELpos=50, ELdist=1](q2,q3){$b$}
\drawedge[ELside=l,ELpos=50, ELdist=1, curvedepth=6](q3,q4){}
\drawedge[ELside=l,ELpos=50, ELdist=1, curvedepth=6](q4,q3){}
\drawedge[ELside=l,ELpos=50, ELdist=1, curvedepth=0](q4,q6){}

\drawloop[ELside=l,loopCW=y, loopangle=180](q5){}
\drawloop[ELside=l,loopCW=y, loopangle=180](q6){}


\end{gpicture}

%% file: figures/gexample-purge.tex
\begin{gpicture}(130,80)(0,0)

\gasset{Nw=7, Nh=7, Nmr=3.5, loopdiam=5}


\node[Nmarks=n](q1)(10,50){$q_1$}
\node[Nmarks=n](q2)(10,65){$q_2$}
\node[Nmarks=n](q3)(27,50){$q_3$}
\node[Nmarks=n](q4)(44,50){$q_4$}
\node[Nmarks=n](q5)(61,50){$q_5$}
\node[Nmarks=n](q6)(61,65){$q_6$}

\node[Nframe=n](label)(35.5,75){$M[e_1]$}

\drawarc[linegray=0](10,50,5,0,60)       
\drawarc[linegray=0](10,65,5,163,240)    
\drawarc[linegray=0](10,65,5,341,17)     
\drawarc[linegray=0](27,50,5,343,0)      
\drawarc[linegray=0](44,50,5,118,153)    
\drawarc[linegray=0](44,50,5,343,41.4)   


\drawedge[ELside=l,ELpos=50, ELdist=1, curvedepth=-3](q1,q2){}
\drawedge[ELside=l,ELpos=50, ELdist=1, curvedepth=-3](q2,q1){}

\drawedge[ELside=l,ELpos=50, ELdist=1, curvedepth=3](q2,q3){}

\drawedge[ELside=r, ELpos=30, ELdist=1](q1,q3){$a$}
\drawedge[ELside=r, ELpos=45, ELdist=1](q3,q4){$a$}
\drawedge[ELside=r, ELpos=33, ELdist=.5, curvedepth=-3](q4,q3){$a$}

\drawedge[ELside=l, ELpos=38, ELdist=.5](q4,q5){$b$}
\drawedge[ELside=l, ELpos=50, ELdist=1](q4,q6){}

\drawedge[ELside=r, ELpos=30, ELdist=.5, curvedepth=-3](q5,q6){$a$}
\drawedge[ELside=l, ELpos=30, ELdist=.5, curvedepth=-3](q6,q5){$a$}


\drawloop[ELside=l, ELpos=25, ELdist=.5,loopCW=y, loopangle=135](q2){$a$}
\drawloop[ELside=r, ELpos=25, ELdist=.5,loopCW=n, loopangle=45](q2){$b$}
\drawloop[ELside=l,loopCW=y, loopangle=315](q3){}
\drawloop[ELside=l,loopCW=y, loopangle=90](q4){}
\drawloop[ELside=l,loopCW=y, loopangle=315](q4){}


\node[Nmarks=n](q1)(10,10){$q_1$}
\node[Nmarks=n](q2)(10,25){$q_2$}
\node[Nmarks=n](q3)(27,10){$q_3$}
\node[Nmarks=n](q4)(44,10){$q_4$}
\node[Nmarks=n](q5)(61,10){$q_5$}
\node[Nmarks=n](q6)(61,25){$q_6$}

\node[Nframe=n](label)(35.5,35){$M[e_2]$}

\drawline[AHnb=0, dash={1}{0}](70,1)(70,34)

\drawarc[linegray=0](10,10,5,0,60)    
\drawarc[linegray=0](10,25,5,163,240) 
\drawarc[linegray=0](27,10,5,343,0)   
\drawarc[linegray=0](44,10,5,118,153) 
\drawarc[linegray=0](44,10,5,0,41.4)   


\drawedge[ELside=l,ELpos=50, ELdist=1, curvedepth=-3](q1,q2){}
\drawedge[ELside=l,ELpos=50, ELdist=1, curvedepth=-3](q2,q1){}


\drawedge[ELside=r, ELpos=30, ELdist=1](q1,q3){$a$}
\drawedge[ELside=r, ELpos=45, ELdist=1](q3,q4){$a$}
\drawedge[ELside=r, ELpos=33, ELdist=.5, curvedepth=-3](q4,q3){$a$}

\drawedge[ELside=l, ELpos=38, ELdist=.5](q4,q5){$b$}
\drawedge[ELside=l, ELpos=50, ELdist=1](q4,q6){}

\drawedge[ELside=r, ELpos=30, ELdist=.5, curvedepth=-3](q5,q6){$a$}
\drawedge[ELside=l, ELpos=30, ELdist=.5, curvedepth=-3](q6,q5){$a$}


\drawloop[ELside=l, ELpos=25, ELdist=.5,loopCW=y, loopangle=135](q2){$a$}
\drawloop[ELside=r, ELpos=25, ELdist=.5,loopCW=n, loopangle=45](q2){$b$}
\drawloop[ELside=l,loopCW=y, loopangle=315](q3){}
\drawloop[ELside=l,loopCW=y, loopangle=90](q4){}

\drawline[AHnb=0, dash={1}{0}](70,1)(70,78)


\node[Nmarks=n](q1)(85,50){$q_1$}
\node[Nmarks=n](q2)(85,65){$q_2$}
\node[Nmarks=n](sD)(102,50){$s_D$}
\node[Nmarks=n](qwin)(102,65){$\winabsorb$}
\node[Nmarks=n](sDp)(119,50){$s_{D'}$}
\node[Nmarks=n](qlose)(119,65){$\loseabsorb$}

\drawarc[linegray=0](85,50,5,0,60)       
\drawarc[linegray=0](85,65,5,163,240)    
\drawarc[linegray=0](85,65,5,341,17)     
\drawarc[linegray=0](102,50,5,343,0)      


\drawedge[ELside=l,ELpos=50, ELdist=1, curvedepth=-3](q1,q2){}
\drawedge[ELside=l,ELpos=50, ELdist=1, curvedepth=-3](q2,q1){}

\drawedge[ELside=l,ELpos=50, ELdist=1, curvedepth=3](q2,sD){}

\drawedge[ELside=r, ELpos=30, ELdist=1](q1,sD){$a$}
\drawedge[ELside=r, ELpos=50, ELdist=1](sD,sDp){}

\drawedge[ELside=r, ELpos=35, ELdist=1](sD,qwin){$\actionstay$}
\drawedge[ELside=r, ELpos=35, ELdist=1](sDp,qlose){$\actionstay$}


\drawloop[ELside=l, ELpos=25, ELdist=.5,loopCW=y, loopangle=135](q2){$a$}
\drawloop[ELside=r, ELpos=25, ELdist=.5,loopCW=n, loopangle=45](q2){$b$}
\drawloop[ELside=l, ELpos=35, ELdist=.5,loopCW=y, loopangle=315](sD){$F_{q_4,b}$}

\drawloop[ELside=l,loopCW=y, loopangle=0](qwin){}
\drawloop[ELside=l,loopCW=y, loopangle=0](qlose){}


\node[Nmarks=n](q1)(85,10){$q_1$}
\node[Nmarks=n](q2)(85,25){$q_2$}
\node[Nmarks=n](sD)(102,10){$s_D$}
\node[Nmarks=n](qwin)(102,25){$\winabsorb$}
\node[Nmarks=n](sDp)(119,10){$s_{D'}$}
\node[Nmarks=n](qlose)(119,25){$\loseabsorb$}

\drawarc[linegray=0](85,10,5,0,60)       
\drawarc[linegray=0](85,25,5,163,240)    


\drawedge[ELside=l,ELpos=50, ELdist=1, curvedepth=-3](q1,q2){}
\drawedge[ELside=l,ELpos=50, ELdist=1, curvedepth=-3](q2,q1){}


\drawedge[ELside=r, ELpos=30, ELdist=1](q1,sD){$a$}
\drawedge[ELside=r, ELpos=40, ELdist=1](sD,sDp){$F_{q_4,b}$}

\drawedge[ELside=r, ELpos=35, ELdist=1](sD,qwin){$\actionstay$}
\drawedge[ELside=r, ELpos=35, ELdist=1](sDp,qlose){$\actionstay$}


\drawloop[ELside=l, ELpos=25, ELdist=.5,loopCW=y, loopangle=135](q2){$a$}
\drawloop[ELside=r, ELpos=25, ELdist=.5,loopCW=n, loopangle=45](q2){$b$}

\drawloop[ELside=l,loopCW=y, loopangle=0](qwin){}
\drawloop[ELside=l,loopCW=y, loopangle=0](qlose){}


\end{gpicture}